\documentclass{article}
\usepackage[utf8]{inputenc}
\usepackage{cite}
\usepackage{graphicx}
\usepackage{authblk}
\usepackage{datetime2,soul,amsthm,mathtools}
\usepackage{amsmath}
\usepackage{tikz}
\usepackage{mathdots}
\usepackage{yhmath}
\usepackage{cancel}
\usepackage{color}
\usepackage{siunitx}
\usepackage{array}
\usepackage{multirow}
\usepackage{amssymb}
\usepackage{gensymb}
\usepackage[all]{xy}
\usepackage{tabularx}
\usepackage{extarrows}
\usepackage{booktabs}
\usetikzlibrary{fadings}
\usetikzlibrary{patterns}
\usetikzlibrary{shapes}
\usepackage{comment}

\usepackage{indentfirst,csquotes}
\usepackage[indentafter]{titlesec}
\titleformat{name=\section}{}{\thetitle.}{0.8em}{\centering\scshape}
\titleformat{name=\subsection}[runin]{}{\thetitle.}{0.5em}{\bfseries}[.]
\titleformat{name=\subsubsection}[runin]{}{\thetitle.}{0.5em}{\itshape}[.]
\titleformat{name=\paragraph,numberless}[runin]{}{}{0em}{}[.]
\titlespacing{\paragraph}{0em}{0em}{0.5em}
\titleformat{name=\subparagraph,numberless}[runin]{}{}{0em}{}[.]
\titlespacing{\subparagraph}{0em}{0em}{0.5em}
\usepackage{lipsum}

\usepackage{amssymb,amsthm,amsmath}
\usepackage{xcolor,paralist,hyperref,titlesec,fancyhdr,etoolbox}
\usepackage[all]{xy}
\theoremstyle{definition}
\newtheorem{definition}{Definition}[section]
\newtheorem{remark}[definition]{Remark}
\newtheorem{example}[definition]{Example}
\theoremstyle{plain}
\newtheorem{theorem}[definition]{Theorem}
\newtheorem{corollary}[definition]{Corollary}

\newtheorem{lemma}[definition]{Lemma}
\newtheorem{proposition}[definition]{Proposition}
\renewenvironment{proof}{\noindent\textsc{Proof.}\quad}{\qed}
\newenvironment{proofof}[1]{\noindent {\it Proof of #1. $\,$}}{\hfill $\Box$ \vskip 5pt }

\newenvironment{lettertheorem}[1]
  {\renewcommand{\thedefinition}{#1}%
   \begin{theorem}}
  {\end{theorem}}

\hypersetup{ colorlinks=true, linkcolor=black, filecolor=black, urlcolor=black }

\newcommand\ZZ{\mathop{}\!\mathbb{Z}}

\newcommand\RR{\mathop{}\!\mathbb{R}}
\newcommand\CC{\mathop{}\!\mathbb{C}}
\newcommand\LL{\mathop{}\!\mathcal{L}}

\newcommand\HDet{\mathop{}\!\widetilde{\mathrm{Det}}}

\begin{document}
\setlength{\columnsep}{5pt}
\title{\vspace{-1.25cm} Two Regularized Determinants of Laplacian through Resurgence theory}


\author[1]{Wen SHEN\thanks{Email: sw@cnu.edu.cn}}
\author[1, 2]{Shanzhong SUN\thanks{Email: sunsz@cnu.edu.cn}}

\affil[1]{Department of Mathematics, Capital Normal University, Beijing 100048, P. R. China}
\affil[2]{Academy for Multidisciplinary Studies, Capital Normal University, Beijing 100048, P. R. China}

\date{}
\maketitle

\begin{abstract}
We study two types of regularizations of the determinant of Laplacian on Riemann  manifold from the viewpoint of resurgence theory. One is the formal logarithmic derivative of the determinant, and the other is its exponential deformation. Under appropriate conditions, the close formulas for both regularized determinant are established through Borel-Laplace resummation which takes into account the contribution of the singularities along the analytic continuation  of Theta series $\hat{\Theta}_{D_X}$. The series resembles the trace of the heat kernel, but is defined via the spectrum of the square-root of the Laplacian.
As applications, we revisit the well known formal logarithmic derivative of  determinant on $S^1$ and compact Riemann surface with higher genus ($\geq2$) corresponding to the Poisson summation formula and Selberg trace formula respectively. Furthermore, the 1-Gevrey asymptotic behavior of the exponential deformation regularization at infinity is considered whose coefficients are determined by the trace of the heat kernel. In the end, we establish the relationship between the two regularized determinants. In fact, they have the same derivatives when the deformation parameter tends to $0$ in exponentially deformed regularization.
\end{abstract}

\tableofcontents

\section{Introduction}\label{1}

\subsection{Background}\label{1.1}

{Functional determinant of Schr\"{o}dinger operator gets involved in the famous Gutzwiller semiclassical trace formula  relating it to the period/length spectra(\cite{gutzwiller1971periodic, gutzwiller1991chaos,sun2017gutzwiller}) through  the trace of Green's function, which is one of our main motivations to address the regularization problem via resurgence theory. 
A typical case is for the Laplace operator on a Riemann surface of genus great than $1$ where Gutzwiller semiclassical trace formula reduces to the well known Selberg trace formula.}

Given a differential operator $P$ for which we will focus on the Laplace operator on a Riemannian manifold, it is natural and a common practice to encode its discrete spectra $\{\lambda_n\}$ into the functional determinant defined by $\prod_{n=1}^{\infty}\lambda_n$ generalizing the familiar determinant of matrix. Unfortunately it is divergent, and regularization is needed to make it sense.

The most classical approach is the zeta regularization, where the regularized determinant is defined as $e^{-\zeta_P^\prime(0)}$. The associated spectral zeta function
$$\zeta_P(s):=\sum_n{\lambda_n^{-s}}$$ admits a contour integral representation involving either the Green's function (\cite[\S5]{Welylaw}) or eigenfunctions satisfying specific boundary conditions (\cite{hur2008fast, kirsten2003functional}). These methods serve distinct purposes: the former reveals the singularity structure of the spectral zeta function by the short time expansion of heat kernel when $P$ is the Laplacian, while the latter yields the ratio of the regularized determinant constructed from the eigenfunctions. Recently, Carosi (\cite{carosi2026equivalent}) tries to incorporate these results into a unified framework based on the observation that the spectral zeta function can be represented as a contour integral
$$\zeta_P(s)=\frac{1}{2\pi i}\int_C\left(\lambda^{-s}\frac{d\log F_P(\lambda)}{d\lambda}\right)d\lambda,$$
with $C$ a contour enclosing the spectrum, and $F_P(\lambda)$ a holomorphic function whose simple zeros correspond to the eigenvalues of $P$.

We will focus on another type of functional determinant which is also well studied: the infinite dimensional analogue of the characteristic polynomial of a matrix. It is defined as
\begin{equation}{\label{def:formal product det}}
\HDet_P(\lambda):=\prod_{n\geq1}(\lambda+\lambda_n),
\end{equation}
where the notation $\HDet$ highlights the fact that the product is formal and generally divergent. There are several techniques to regularize it in the literature.

\subsubsection{Zeta Regularization}\label{1.1.1}
Like the zeta regularization of $\prod_{n=1}^{\infty}\lambda_n$, 
$\HDet_P(\lambda)$ 
can be regularized by the form
$$D_P(\lambda):=\exp{\left(-\frac{\partial}{\partial s}Z_P(s,\lambda)\Big|_{s=0}\right)},$$
where $Z_P(s,\lambda):=\sum_{n=0}^{\infty}(\lambda+\lambda_n)^{-s}$ is a two-variable zeta function. The zeta-regularized determinant is well-defined following from the fact that $Z_P(s,\lambda)$ is the Mellin transform of $\hat{\Theta}_P(t)e^{-\lambda t}$, where 
$$\hat{\Theta}_P(t):=\sum_{n=1}^{\infty}e^{-\lambda_nt}$$ is also one of our main objects, and it has a short time asymptotic expansion when $t\to0^+$, at least for 
Laplace operator on a compact manifold. 

As a classical example, the functional determinant of the Laplacian on a compact Riemann surface $X$ with genus $g \geq2$ can be evaluated by applying the Selberg trace formula to a resolvent-type test function (\cite{voros1987spectral, cartierVoros1988}). Consequently, the regularized determinant of $-\Delta_X+\frac{1}{4}$ can be expressed as the product of the Selberg zeta function and the regularized determinant of the round sphere.

\subsubsection{Weierstrass Product Form}\label{1.1.2}
The Weierstrass product form of the determinant is
$$\Delta_{P}(\lambda):=\prod_{n=0}^{\infty}\left(1-\frac{\lambda}{\lambda_n}\right) \exp\left(\frac{\lambda}{\lambda_n}+\frac{\lambda^2}{2 \lambda_n^2}+\cdots+\frac{\lambda^m}{m \lambda_n^m}\right),$$
where $m$ is a constant depending on the dimension of the manifold. 
Voros (\cite{voros1987spectral}) considered the relation between this regularized methods and the one in \S\ref{1.1.1}. They can both be viewed as the $m$-th formal integral of the spectral sum $\sum_n(\lambda+\lambda_n)^{-m-1}$ up to integration ambiguities, but are normalized at different places ($\lambda\to-\infty$ and $\lambda=0$, respectively). Thus, they differ by a factor $e^{Q(\lambda)}$ with $Q(\lambda)$ a polynomial of degree $m$. 

Furthermore, Illies introduced a general $\delta$-regularized determinant (\cite[Def. 3.4]{illies2001regularized})  via $$\Delta_{P,\delta}(\lambda) := \exp\!\left(-\mathrm{CT}_{s=0}\,\delta(s)Z_p(s,\lambda)\right),$$ 
where $\delta(s)$ is a formal series about $s$ encoding the spectra data and $\mathrm{CT}_{s=0}$ means the constant term in the Laurent expansion at $s=0$. Zeta-regularization, Weierstrass canonical product and other regularizations are the special cases of the $\delta$-regularized determinant by the different choices of $\delta(s)$(\cite[Examples in p.73]{illies2001regularized}). Moreover, every $\delta$-regularized determinant also differs from the Weierstrass product by an exponential polynomial $e^{P(z)}$ (\cite[Thm. 1]{illies2001regularized}). He also provided a criterion for determining whether a function is a $\delta$-regularized determinant based on its asymptotic expansion at infinity(\cite[Thm. 5]{illies2001regularized}).

\subsubsection{Laplace Method}\label{1.1.3}
In (\cite{voros1992spectral}), Voros further studied the relations among the spectral functions $\hat{\Theta}_P(t)$, $Z_P(s,\lambda)$, and $D_P(\lambda)$. Beyond the Mellin transform relation between $\hat{\Theta}_P(t)e^{-\lambda t}$ and $Z_P(s,\lambda)$, he established that the finite part of the log-determinant can be viewed as the Laplace transform of the regularized Theta function
\begin{equation}\label{laplace with F.P}
-\mathrm{F.P.}\left[\log D_P(\lambda)\right]=\int_0^{\infty}\mathrm{F.P.}\left[ \frac{\hat{\Theta}_P(t)}{t}\right] e^{-t\lambda}dt,
\end{equation}
where $\mathrm{F.P.}$ denotes the subtraction of divergent terms as $t\to0$ or $\lambda\to\infty$. Based on it, Voros derived functional equations for $D_P(\lambda)$ via contour integration around the singularities of
$\hat{\Theta}_P(t)$ for some specific examples: $\{\lambda_n\}\subseteq\ZZ$; the spectrum of the Laplacian on some compact hyperbolic surface; and the spectrum of the Schr\"{o}dinger operator with homogeneous potential $q^{2M}(M\in\ZZ_{>0})$ (\cite[Eqs. 3.15, 4.15, 5.26]{voros1992spectral}). Futhermore, he proposed a general framework (\cite[Fig. 3.18, \S6]{voros1992spectral} and \cite{voros1993quantumresur}) whereby the functional relations for determinants would arise from the global analytic structure of "resurgent" Theta series. While he computed explicit examples via contour integration, a general rigorous mathematical description of determinants in terms of the "resurgent structure", and a precise condition for $\hat{\Theta}_P$ is applicable to his framework, were lacking. Curiosity on such resurgent structures is another motivation to the current work. 


\subsection{Main Results}\label{1.2}
We introduce two novel regularized determinants that are different from the classical zeta-regularized determinant $D_P$, providing a unified framework and new perspective on spectral invariants through the resurgence theory(\cite{ecallebook1, ecallebook2, ecallebook3, sauzinbook}). By using the Stokes formula (lemma \ref{lemma: Laplace and stokes}), we establish explicit computational formulas in Theorem \ref{theorem about derivative regularization in introduction} and Theorem \ref{thm about the computation of deform in intro} to express these determinants as a sum of alien operators (see Definition \ref{Def:alien operator}) acting on the modified Theta series under precise analytical conditions. {This method of using Alien operators to encode the singularities of the holomorphic extension of a germ along paths originating from the origin, provides a rigorous and explicit description of the "global analytic structure" within Voros' framework (\cite{voros1992spectral}).}

For the Laplace-type operator $P=-\Delta_X+C$ on a compact Riemann manifold $X$ of dimension $d$, it is well-known that $P$ is self-adjoint and has discrete spectrum, whose non-zero eigenvalues are positive $$0<\lambda_1\leq\cdots\leq\lambda_n\leq\cdots\to\infty.$$
Instead of subtracting the divergent terms of $\hat{\Theta}_{-\Delta_X+C}(t)$ in \eqref{laplace with F.P}, we propose two methods that make $\hat{\Theta}_{-\Delta_X+C}(t)$ integrable. Bearing in mind the asymptotic behavior, we multiply  $(-t)^m$ to $\hat{\Theta}_{-\Delta_X+C}(t)$. In fact, $\hat{\Theta}_{-\Delta_X+C}(t)\in O(t^{-\frac{d}{2}})$ as $t\to 0$, so the choice for $m=\lfloor\frac{d}{2}\rfloor+1$ is enough. Another method is deformation, and we shift $\hat{\Theta}_{-\Delta_X+C}(t)$ by a negative constant $s_0$. Both methods make the Laplace transform over the modified Theta series well-defined, thus

\begin{equation}{\label{relation:derivative laplace in intro}}
\int_0^{e^{i\theta}\infty}(-t)^{m-1}\hat{\Theta}_{-\Delta_X+C}(t)e^{-\lambda t}dt=\sum_{n\geq1}\frac{(-1)^{m-1}(m-1)!}{(\lambda+\lambda_n)^m},\quad\theta\in\left(-\frac{\pi}{2},\frac{\pi}{2}\right);
\end{equation}
\begin{equation}{\label{relation:deform laplace in intro}}
\int_0^{e^{i\theta}\infty}\hat{\Theta}_{-\Delta_X+C}(t-s_0)e^{-\lambda t}dt=\sum_{n\geq1}\frac{e^{\lambda_n s_0}}{\lambda+\lambda_n},\quad\theta\in\left(-\frac{\pi}{2},\frac{\pi}{2}\right).
\end{equation}
We will present two regularization methods emerging from these identities, and deduce the expressions for their corresponding regularized determinants through the analysis of the singularities of two modified Theta series.

\textbf{The First Regularization.} In terms of the formal product $\HDet_{-\Delta_X+C}(\lambda)$ in \eqref{def:formal product det}, the right hand side of \eqref{relation:derivative laplace in intro} can be viewed as
$$\widetilde{\frac{d^m}{d\lambda^m}}\widetilde{\log} \HDet_{-\Delta_X+C}(\lambda),$$
which we call the \textbf{logarithmic derivative regularization of the determinant}. We add the "tilde" to emphasize the formal nature of taking the logarithm and the derivatives. If we write $\rho^2=\lambda$ and the spectrum $\rho_n^2=\lambda_n$, the regularized determinant
$$\widetilde{\frac{d^{m_0}}{d\rho^{m_0}}}\widetilde{\log}\HDet_{-\Delta_X+C}(\rho^2):=\sum_{n\geq 1}\frac{d^{m_0-1}}{d\rho^{m_0-1}}\frac{2\rho}{\rho^2+\rho_n^2},\quad\text{where }m_0=d+1,$$
can be characterized by the singularities and the {holomorphic} extension of $(-t)^d\hat\Theta_{D_X}(t)$, 
where
$$\hat\Theta_{D_X}(t):=\sum_{n=1}^{\infty}e^{-\rho_n t}\in O(t^{-d}),\quad t\to 0.$$
To be more precise, we have the following

\begin{lettertheorem}{A}\label{theorem about derivative regularization in introduction}
Let $m_0=d+1$. Assume:
\begin{itemize}
\item[(i)] There exists a non-empty closed discrete set $\Omega=\big\{i\tau_k:\tau_k\in\RR,\tau_0=0,\tau_k<\tau_{k+1},k\in\ZZ\big\}$, such that $(-t)^d\hat{\Theta}_{D_X}(t)\in\hat{\mathcal{R}}_{\Omega^*}^{simp}$;
\item[(ii)] There exists $\delta>0$, such that $\hat{\Theta}_{D_X}(t)$ can be holomorphically extended to
$$U_\delta:=\left\{re^{i\theta}:r>0,-\frac{\pi}{2}-\delta<\theta<\frac{\pi}{2}+\delta\right\}\setminus\big(i(\tau_1,+\infty)\cup i(-\infty,\tau_{-1})\big)$$
along the path in it, with $t^d\hat{\Theta}_{D_X}(t)\in\mathcal{N}\left(\left(\frac{\pi}{2},\frac{\pi}{2}+\delta\right),0\right)$.

\end{itemize}
Then, for $0<\epsilon<\delta$ and $\arg(\rho)\in\left(-\frac{\pi}{2}+\epsilon,\frac{\pi}{2}-\epsilon\right)$,
\begin{equation}\label{formula:reg1 in intro}
\begin{aligned}
\widetilde{\frac{d^{m_0}}{d\rho^{m_0}}}\widetilde{\log}\HDet_{-\Delta_X+C}(\rho^2)=&i^{m_0}\left(\sum_{\tau_k>0}e^{-\tau_k\rho}\LL^{\left(\frac{\pi}{2}-\epsilon\right)}\left(\Delta_{i\tau_k}^-\left(t^{m_0-1}\hat{\Theta}_{D_X}(t)\right)\right)(-i\rho)\right) \\&-i^{m_0}\LL^{\left(\frac{\pi}{2}+\epsilon\right)}\left(t^{m_0-1}f(t)\right)(-i\rho),
\end{aligned}
\end{equation}
if the right side converges, where $$f(t):=\hat{\Theta}_{D_X}(te^{-i\pi})+\hat{\Theta}_{D_X}(t).$$
We call $f(t)$ the "\textbf{asymmetry part}" of $\hat{\Theta}_{D_X}(t)$.
\end{lettertheorem}
As demonstrated in (\cite{voros1987spectral}), this form of $\widetilde{\frac{d^{m}}{d\lambda^{m}}}\widetilde{\log} \widetilde{Det}_{-\Delta_X+C}(\lambda)$ serves as the common $m$-th derivative for both the zeta-regularized determinant and the Weierstrass product. Specifically, these two regularized determinants can be obtained by
integrating $m$ times, then normalizing at $-\infty$ and $0$, respectively. Therefore, such characterization of $\widetilde{\frac{d^{m}}{d\lambda^{m}}}\widetilde{\log} \widetilde{Det}_{-\Delta_X+C}(\lambda)$ provides a unified approach to deriving both zeta-regularized and Weierstrass product determinants.
\begin{remark}
Several explanations are in the sequel:
\begin{itemize}
    \item[(i)] $\Delta_{i\tau_k}^-$ is the alien operator in resurgence theory which describes the singularity behaviour at $i\tau_k$ (Definition \ref{Def:alien operator}). In particular, if all singular points of $t^{m_0-1}\hat{\Theta}_{D_X}(t)$ are simple poles, then $\Delta_{i\tau_k}^-t^{m_0-1}\hat{\Theta}_{D_X}(t)$ is actually the residue of $-2\pi it^{m_0-1}\hat{\Theta}_{D_X}(t)$ at $i\tau_k$ multiplying Dirac delta, and
\begin{equation*}
\begin{aligned}
    \widetilde{\frac{d^{m_0}}{d\rho^{m_0}}}\widetilde{\log}\HDet_{-\Delta_X+C}(\rho^2)=&-2\pi i^{m_0+1}\left(\sum_{\tau_k>0}e^{-\tau_k\rho}\mathrm{Res}\left(t^{m_0-1}\hat{\Theta}_{D_X}(t);i\tau_k\right)\right) \\&-i^{m_0}\mathcal{L}^{\left(\frac{\pi}{2}+\varepsilon\right)}\left(t^{m_0-1}f(t)\right)(-i\rho),\quad\arg(\rho)\in\left(-\frac{\pi}{2}+\varepsilon,\frac{\pi}{2}-\varepsilon\right).
\end{aligned}
\end{equation*}
Thus, \eqref{formula:reg1 in intro} means that the regularized determinant is the sum due to the Stokes structure and the Laplace transform over the "asymmetry part" of the modified Theta series.

\item[(ii)] The reason for choosing the new series $\hat\Theta_{D_X}(t)$ by letting $\rho_n^2=\lambda_n$ is that if we consider the original $\hat\Theta_{-\Delta_X+C}(t)$, the singular set is always dense even when $d=1$. This is due to Weyl's law (\cite{Welylaw}) $\lambda_n\sim n^{2/d}(n\to\infty)$ and the Fabry gap theorem (\cite{Fabrygap}): $\frac{n}{\lambda_n}\to0(n\to\infty)$ implies the line of imaginary numbers is singular. This situation violates the condition on the discreteness of the singularities which is necessary to apply resurgence theory.
{Moreover, singular points of $\hat\Theta_{D_X}(t)$ in the vertical axis correspond directly to the length spectra of the Hamiltonian $H=-\Delta_X+C$ in the examples for $X=S^1$ and higher-genus compact Riemann surfaces. This correspondence perfectly aligns with our motivation to understand the relationship between length and operator spectrum.}

\item[(iii)] {Our theorem 
is established by resurgence theory. 
One would say that the first regularization of determinant can be characterized by the resurgent structure of the modified Theta series. This provides certain precise mathematical conditions when Voros' framework (\cite{voros1992spectral}) holds. }
%
\item[(iv)] Our result unifies several known formulas. For instance, if $X=S^1$, \eqref{formula:reg1 in intro} implies the Poisson summation formula; if $X$ is a higher-genus compact Riemann surface, \eqref{formula:reg1 in intro} reproduces the functional equation satisfied by the Selberg zeta function, which is compatible with the formula in (\cite{voros1992spectral}).
It is a right place to point that our formula for the regularized determinant from the point of view of resurgence theory can be heuristically viewed as the Laplace transform of the Poisson-Newton formula in \cite{Marco2015unified} which also unifies the Poisson summation formula and the Selberg trace formula from a distributional perspective.            
\end{itemize}
\end{remark}

\textbf{The Second Regularization.} The identity \eqref{relation:deform laplace in intro} provides another method to regularize the determinant. For $s_0<0$ a negative real number, we introduce
    \begin{equation}{\label{Def:exp deform2 in intro}}
\widetilde{\frac{d}{d\lambda}}\widetilde{\log} \HDet_{-\Delta_X+C}^{\sharp_2}(\lambda):=\sum_{n=1}^{\infty} \frac{e^{s_0\sqrt{\lambda_n}}}{\lambda+\lambda_n},
\end{equation}
and call it the "\textbf{exp deformed regularization of the determinant}". Similar to the previous regularization, the second one can also be expressed as the sum of the Stokes terms and the Laplace transform of the "asymmetry part" of the $\hat{\Theta}_{D_X}(t-s_0)$.

\begin{lettertheorem}{B}{\label{thm about the computation of deform in intro}}
Fix a real number $s_0<0$. Assume
\begin{itemize}
\item[(i)] There exists a non-empty closed discrete set $\Omega_{s_0}=\big\{s_0+i\tau_k:\tau_0=0,\tau_k\in\RR,\tau_k<\tau_{k+1},k\in\ZZ\big\}$,
        such that
        $\hat{\Theta}_{D_X}(t-s_0)\in \hat{\mathcal{R}}_{\Omega_{s_0}^*}^{\text{simp}}$;

\item[(ii)] There exists $\delta\in\RR$, such that $\arg(s_0+i\tau_1)<\frac{\pi}{2}+\delta<\pi$, and $\hat{\Theta}_{D_X}(t-s_0)$ can be holomorphic extended to
    $$U_\delta:=\left\{s_0+re^{i\theta}:r>0,-\frac{\pi}{2}-\delta<\theta<\frac{\pi}{2}+\delta\right\}\setminus\left\{s_0+it:t\in(\tau_1,+\infty)\cup(-\infty,\tau_{-1})\right\}$$ along the path in it, with $\hat{\Theta}_{D_X}(t-is_0\cot\delta))\in\mathcal{N}\left(\left(\frac{\pi}{2},\frac{\pi}{2}+\delta\right),0\right)$.
\end{itemize}
Then, for $\arg(s_0+i\tau_1)<\pi-\epsilon<\frac{\pi}{2}+\delta$,
and $\arg(\rho)\in\left(-\frac{\pi}{2}+\epsilon,\epsilon\right)$, we have
\begin{equation}\label{eq:compute det1 formula}
\begin{aligned}
\widetilde{\frac{d}{d\rho}}\widetilde{\log}\HDet^{\sharp_2}_{-\Delta_X+C}(\rho^2)=&i\left(\sum_{\tau_k>0}e^{i(s_0+\tau_ki)\rho}\LL^{\left(\frac{\pi}{2}-\epsilon\right)}\left(\Delta_{s_0+i\tau_k}^-\left(\hat{\Theta}_{D_X}(t-s_0)\right)\right)(-i\rho)\right)\\&-i\LL^{(\pi-\epsilon)}(f_{s_0}(t))(-i\rho),
\end{aligned}
\end{equation}
if the right side converges, where
$$
f_{s_0}(t):=\hat{\Theta}_{D_X}(e^{-i\pi}t-s_0)+\hat{\Theta}_{D_X}(t-s_0).
$$
We also call $f_{s_0}(t)$ the "\textbf{asymmetry part}" of $\hat{\Theta}_{D_X}(t-s_0)$.
\end{lettertheorem}

As an application, we consider the exp deformed regularization of the determinant of $-\Delta_{S^1}$ and obtain a new formula \eqref{deformed distributon PSF} which implies the deformed Poisson summation formula (Theorem \ref{thm: deformed psf}). One would say that \eqref{deformed distributon PSF} is more essential than the deformed Poisson summation formula since the latter captures only the residue contributions on the left hand side of \eqref{deformed distributon PSF}.

Furthermore, we study the asymptotic behavior of the exp deformed regularization of the determinant by Borel-Laplace transform when $\lambda\to\infty$ in Theorem \ref{thm:deform 1-Gevrey expansion}. The asymptotic coefficients are derivatives of Theta series, which is also compatible with the asymptotic behavior of the Weierstrass product of $\HDet_P(\lambda)$ (\cite{cartierVoros1988}). We also consider the relationship between these two regularizations (Theorem \ref{relation Reg1 and Reg2}). Indeed, the $(m-1)$-th derivative of the exp deformed regularization tends to the log derivative regularization if the deform parameter $s_0\to0$.
{Although this can be shown by a direct calculation, our more involved method explains that the underlying relationship results from the singularity translations of $\hat{\Theta}(t-s_0)$ and the integrability in Laplace transform.}

\subsection{Outline}\label{1.3}
An overview of the main results is presented schematically as follows:
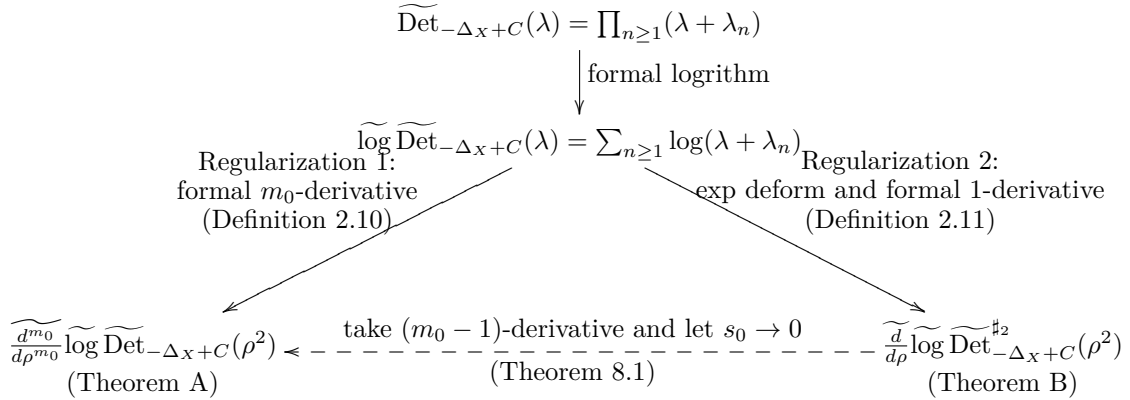
\[
\xymatrix{
& \txt{$\HDet_{-\Delta_X+C}(\lambda)=\prod_{n\geq1}(\lambda+\lambda_n)$}
\ar^{\txt{formal logrithm}}[d]\\
& \txt{$\widetilde{\log}\HDet_{-\Delta_X+C}(\lambda)=\sum_{n\geq1}\log(\lambda+\lambda_n)$}
\ar_{\txt{~\\Regularization 1:\\formal $m_0$-derivative\\(Definition \ref{def:deriva reg})}}[ddl]
\ar^{\txt{~\\Regularization 2:\\exp deform and formal $1$-derivative\\
(Definition \ref{def:exp reg})}}[ddr] \\ \\
\txt{$\widetilde{\frac{d^{m_0}}{d\rho^{m_0}}}\widetilde{\log}\HDet_{-\Delta_X+C}(\rho^2)$\\
(Theorem \ref{theorem about derivative regularization in introduction})
}
&&
\txt{$\widetilde{\frac{d}{d\rho}}\widetilde{\log}\HDet^{\sharp_2}_{-\Delta_X+C}(\rho^2)$\\
(Theorem \ref{thm about the computation of deform in intro})
}
\ar@{-->}^{\txt{(Theorem \ref{relation Reg1 and Reg2})}}_{\txt{take $(m_0-1)$-derivative and let $s_0\to0$}}[ll]
}\]
This paper is organized as follows:

{Section \ref{Section:Preliminary}: In this section, we introduce the Theta series and the distribution about its singularities. Inspired by the formal Laplace transform relation between Theta series and determinants, we derive two determinant regularization methods (Definition \ref{def:deriva reg} and \ref{def:exp reg}) which are based on two ways to make the formal Laplace transform rigorous—specifically by modifying the Theta series (corollary \ref{relation:laplace relation four}) to ensure that is integrable near $t=0$. Finally, we introduce the resurgence theory which provides a unifying mathematical framework expressing these regularized determinants through the singularities of modified Theta series in \S\ref{section:computation about derivative regularization} and \S\ref{section:The computation of  exp deformed determinant}.}

{Section \ref{section:deformed Borel-laplace}: We study the 1-Gevrey asymptotic behavior of formula \eqref{relation:deform laplace in intro} at infinity, the right hand side of which
is related to the exp deformed regularization. As a result (Theorem \ref{thm:deform 1-Gevrey expansion}), the coefficients in the asymptotic expansion can be described in terms of the derivative of Theta series.}

{Section \ref{section:computation about derivative regularization}: We prove Theorem \ref{theorem about derivative regularization in introduction} to compute the logarithmic derivative of the regularized determinant, in which the logarithmic derivative regularization of determinant is expressed as the sum of contributions from the singularities of $(-t)^d\hat{\Theta}_{D_X}(t)$ and asymmetry of its holomorphic extension.}

{Section \ref{section:The computation of  exp deformed determinant}: Theorem \ref{thm about the computation of deform in intro} is established to compute the exp deformed regularization of determinants. 
The singularities of shifted Theta series $\hat{\Theta}_{D_X}(t-s_0)$ are involved.}

{Section \ref{section:example S1}: We compute these two regularized determinants for the Laplacian on $S^1$, and rederive the well-known Poisson summation formula (Theorem \ref{theorem:PSF}) and its deformed form (Theorem \ref{thm: deformed psf}). The details of the proof are given in appendix.}

{Section \ref{section:example X}: We study the derivative regularized determinant about the Laplacian on Riemann surface of higher genus. The result \eqref{3-derivative of reg det in Rie-Surface.} formula is the combination of the derivative regularized determinant on $S^2$ and Selberg Zeta function, which is also compatible with the Weierstrass factorization of determinant in (\cite{cartierVoros1988}).}

{Section \ref{section:relation of two reg}: Finally we investigate the relationship (Theorem \ref{relation Reg1 and Reg2}) between two regularizations. If we take the ($m-1$)-th derivative of exp deformed determinant, it will tend to the derivative regularized determinant when the deform parameter $s_0\to 0$.}
\section{Preliminaries}{\label{Section:Preliminary}}
The Theta series and determinants are related through the formal Laplace transform. Inspired by this fact, we study two determinant regularization methods based on the ways to make this formal transform rigorous by ensuring that the modified Theta series is integrable near $t=0$. Finally, we introduce the resurgence theory which provides an analytical framework expressing these regularized determinants through the singularities of the modified Theta series (see \S\ref{section:computation about derivative regularization} and \S\ref{section:The computation of  exp deformed determinant}).

\subsection{Spectrum and Theta series}
Let $(X,g)$ be a closed connected Riemann manifold of dimension $\dim X=d$. The spectral theorem (\cite[Theorem 7.2.6]{spectralthm}) implies that the eigenvalues of Laplace operator $-\Delta_X$ are discrete and positive. We denote them by $0<\lambda_1\leq\lambda_2\leq\cdots\leq\lambda_n\leq\cdots$. The asymptotic behavior of $\lambda_n$ is governed by Weyl's law:

\begin{lemma}[Weyl's law, {\cite[pp. 256]{Welylaw}}]\label{lem:astomptotic of spectrum}
Under the above assumptions, we have $$N(\lambda):=\#\{n\in\ZZ_{>0}:\lambda_k\leq\lambda\}=\frac{\mathrm{vol}(X)}{(4\pi)^{\frac{d}{2}}\Gamma(\frac{d}{2}+1)}\lambda^{\frac{d}{2}}(1+o(1)),\quad\lambda\to+\infty.
$$
In particular, if we take $\lambda=\lambda_n$, then $N(\lambda_n)=n$ and
$$\lim_{n\to\infty}\frac{\lambda_n}{n^{\frac{2}{d}}}=\frac{(4\pi)^{\frac{d}{2}}\Gamma(\frac{d}{2}+1)}{\mathrm{vol}(X)}.$$
\end{lemma}

We define the Theta series
$$\hat{\Theta}_{-\Delta_X}(t):=\sum_{n=1}^{\infty}e^{-\lambda_nt},\quad t\in\CC.$$ The asymptotic behavior of the spectra ensures that the Theta series is convergent and holomorphic on the right half-plane. As the trace of the heat kernel, this series serves as a fundamental object in Riemannian geometry. 

\begin{proposition}{\label{prop:laplace Theta propty}}
For the Theta series, we have
\begin{itemize}
\item[(i)] $\hat{\Theta}_{-\Delta_X}(t)$ is holomorphic in $\Re(t)>0$;
\item[(ii)] $\hat{\Theta}_{-\Delta_X}(t)\sim Ct^{-\frac{d}{2}}$, as $t\to0$;
\item[(iii)] $|\hat{\Theta}_{-\Delta_X}(t)|\leq C'(\delta)e^{-\lambda_1\Re(t)}$, when $\Re(t)>\delta$, for all fixed $\delta>0$.
\end{itemize}
\end{proposition}
\begin{proof}

(i) According to Lemma \ref{lem:astomptotic of spectrum}, for any constant $M$, $0<M<\frac{(4\pi)^{\frac{d}{2}}\Gamma(\frac{d}{2}+1)}{\mathrm{vol}(X)}$, there exists $N>0$, such that for any $n>N$, we have $\lambda_n>Mn^{\frac{2}{d}}$. Thus, for $\Re(t)>0$,
$$\begin{aligned}
|\hat{\Theta}_{-\Delta_X}(t)|
&\leq\sum_{n=1}^Ne^{-\lambda_n\Re(t)}+\sum_{n=N+1}^{\infty}e^{-\lambda_n\Re(t)}\\
&\leq\sum_{n=1}^Ne^{-\lambda_n\Re(t)}+\sum_{n=N+1}^{\infty}e^{-Mn^{\frac{2}{d}}\Re(t)}\\
&\leq\sum_{n=1}^Ne^{-\lambda_n\Re(t)}+\frac{d!}{(M\cdot\Re(t))^d}\sum_{n=N+1}^{\infty}\frac{1}{n^2}<\infty.
\end{aligned}$$
Hence, $\hat{\Theta}_{-\Delta_X}(t)$ converges uniformly and is holomorphic in $\Re(t)>0$.

(ii) This is (\cite[Corollary 2.2]{DuistermaatGuillemin}).

$(iii)$ 
Again by Lemma \ref{lem:astomptotic of spectrum}, the number $r$ of elements in $\{n\in\ZZ_{>0}:\lambda_1\leq\lambda_n\leq2\lambda_1\}$ is finite. Therefore,
$$
\begin{aligned}
|\hat{\Theta}_{-\Delta_X}(t)|&=\left|\sum_{n=1}^re^{-\lambda_nt}+\sum_{n=r+1}^{\infty}e^{-\frac{\lambda_n}{2}t}e^{-\frac{\lambda_n}{2}t}\right|\\
&\leq\left|\sum_{n=1}^re^{-\lambda_nt}\right|+\sum_{\lambda_n>2\lambda_1}\left|e^{-\frac{\lambda_n}{2}t}\right|\left|e^{-\frac{\lambda_n}{2}t}\right|\\
&\leq re^{-\lambda_1\Re(t)}+\sum_{n=1}^{\infty}e^{-\lambda_n\frac{\delta}{2}}e^{-\lambda_1\Re(t)}\\
&=\left(r+\hat{\Theta}_{-\Delta_X}\left(\frac{\delta}{2}\right)\right)e^{-\lambda_1\Re(t)}.
\end{aligned}
$$
This completes the proof.
\end{proof}

If one replaces $-\Delta_X$ with the shifted Laplace operator $-\Delta_X + C$ ($C>-\lambda_1$), Proposition \ref{prop:laplace Theta propty} still holds. 
We retain the same notation $\{\lambda_n\}$ for the spectrum of $-\Delta_X+C$ by a bit abusing the notations.

Proposition \ref{prop:laplace Theta propty} (ii) implies that the Theta series $\hat{\Theta}_{-\Delta_X+C}$ is singular at $0$. Furthermore, the following Fabry's gap theorem tell us more about singular points.

\begin{lemma}[Fabry's gap theorem, {\cite[Theorem 20.1]{Fabrygap}}]\label{lem:fabry gap}
Let $t\in\CC$ and $\lambda_1<\lambda_2<\cdots\to\infty$. If $f(t)=\sum_{n=1}^{\infty}a_ne^{-\lambda_nt}$
converges in $\Re(t)>0$ and
$$
\lim _{n\to\infty}\frac{n}{\lambda_n}=0,
$$
then either all points of $\Re(t)=0$, or none, belong to the set of non-singular points of $f(t)$.
\end{lemma}

So all points in $\Re(t)=0$ are singular points of $\hat{\Theta}_{-\Delta_X+C}$ if $d=1$. However, to study the spectral determinant by analyzing the singularities of Theta series using resurgence theory requires the discreteness of singularities. One possible approach is to use eigenvalues appropriately to make the conditions of Lemma \ref{lem:fabry gap} invalid.  For instance, one can choose $\rho_n:=\lambda_n^\alpha$ with $\alpha$ depending on the dimension by Lemma \ref{lem:astomptotic of spectrum}. In our paper, the choice $\rho_n=\sqrt{\lambda_n}$ is enough to handle examples, and we denote the corresponding Theta series as follows:
$$\hat{\Theta}_{D_X}(t):=\sum_{n=1}^{\infty}e^{-\rho_n t},\quad t\in\CC.$$
This series still possesses some properties similar to those in Proposition \ref{prop:laplace Theta propty}.

\begin{proposition}{\label{proposition:Dirac-theta property}}
For the modified Theta series, we have
\begin{itemize}
\item[(i)] $\hat{\Theta}_{D_X}(t)$ is holomorphic in $\Re(t)>0$;
\item[(ii)] $\hat{\Theta}_{D_X}(t)\sim Ct^{-d}$, as $t\to0$ (See (\cite[Corollary 2.2]{DuistermaatGuillemin}) for the details);
\item[(iii)] $|\hat{\Theta}_{D_X}(t)|\leq C'(\delta)e^{-\rho_1\Re(t)}$, when $\Re(t)>\delta$, for all fixed $\delta>0$.
\end{itemize}
\end{proposition}

Proposition \ref{prop:laplace Theta propty} and Proposition \ref{proposition:Dirac-theta property} indicate that the Laplace transform of our Theta series are well-defined since they are integrable at $0$ to which we now turn.

\subsection{Laplace Transform of Theta series}
Let us start from the Laplace transform.

\begin{definition}[{\cite[Definition 5.29]{sauzinbook}}]\label{def:Laplace transform}
Let $I\subseteq\RR$ be an open interval and $\gamma:I\to\RR$ be a locally bounded function. We denote by $\mathcal{N}(I,\gamma)$ the set consisting of all measurable and locally integrable complex functions $\hat{\varphi}:\{re^{i\theta}:r>0,\theta\in I\}\to\CC$ such that
\begin{itemize}
\item[(i)] $\hat{\varphi}$ is integrable near the origin\footnote{We have removed the requirement for the holomorphicity of $\hat{\varphi}(t)$ at $0$ as did in (\cite[Definition 5.29]{sauzinbook}), and replaced it with the integrability condition. In this paper, we will use the same notations if the results also hold under this weaker assumption, and we will provide the proof if necessary.};

\item[(ii)] There exists a locally bounded function $\alpha:I\to\RR_{>0}$, such that for any $\epsilon>0$, $r>0$ and $\theta\in I$, we have $|\hat{\varphi}(re^{i\theta})|\leq\alpha(\theta)e^{(\gamma(\theta)+\epsilon)\cdot r}$.
\end{itemize}
We write $\hat{\varphi}\in\mathcal{N}(I)$, if $\hat{\varphi}\in\mathcal{N}(I,\gamma)$ for some locally bounded function $\gamma:I\to\RR$. Since we always take $\gamma$ to be a constant function, we usually replace $\gamma$ with a real number in $\mathcal{N}(I,\gamma)$.
\end{definition}
Similarly, for a single point set $I=\{\theta\}\subseteq\RR$, the function $\gamma$ in Definition \ref{def:Laplace transform} degenerates to $\gamma:\{\theta\}\to\RR$. Now, the set $\mathcal{N}(I,\gamma)$ degenerates to $\mathcal{N}(\theta,\gamma)$.

The variable $\theta$ appearing in the Laplace transform should be viewed as the direction to take the integral,
$$\begin{aligned}
\LL^{\theta}:\mathcal{N}(\theta,\gamma)&\longrightarrow\mathcal{O}(\Pi_{\gamma}^{\theta})\\
\hat{\varphi}(t)&\longmapsto(\LL^{\theta}\hat{\varphi})(\rho):=\int_0^{e^{i\theta}\infty}e^{-\rho t}\hat{\varphi}(t)dt,
\end{aligned}$$
where $\Pi_{\gamma}^{\theta}$ is the half-plane in $\CC$ defined by $\Pi_{\gamma}^{\theta}:=\{\rho\in\CC:\Re(\rho e^{i\theta})>\gamma(\theta)\}$. Note that if $\hat{\varphi}(t)$ is not holomorphic but integrable near the origin, the Laplace transform is understood as $\lim_{r\to 0}\int_{re^{i\theta}}^{e^{i\theta}\infty}e^{-\rho t}\hat{\varphi}(t)dt$. Furthermore, this transform is well-defined on $\mathcal{N}(\theta,\gamma)$ by Lebesgue dominated convergence theorem.

Now we can define the Laplace transform in directions $I$, where $I$ is an open interval of $\RR$ with length $<\pi$. Let $\gamma:I\to\RR$ be a locally bounded function as before, define
$$\begin{aligned}
\LL^I:\mathcal{N}(I,\gamma)&\longrightarrow\mathcal{O}(\mathcal{D}(I,\gamma))\\
\hat{\varphi}(t)&\longmapsto(\LL^{\theta}\hat{\varphi})(\rho),\quad\text{for some }\theta\in I,
\end{aligned}$$
where $\mathcal{D}(I,\gamma):=\bigcup_{\theta\in I}\Pi_{\gamma}^{\theta}$.

\begin{proposition}\label{def:Laplace transform is welldefined}
The Laplace transform $\LL^I$ is well-defined, i.e., the definition does not depend on the choice of $\theta\in I$.
\end{proposition}
\begin{proof}
Let $\theta_1,\theta_2\in I$, we have $0<\theta_2-\theta_1<\pi$ and $\rho\in\Pi_{\gamma(\theta_1)}^{\theta_1}\cap\Pi_{\gamma(\theta_2)}^{\theta_2}$. For $\hat{\varphi}\in\mathcal{N}(I,\gamma)$,
$$
\begin{aligned}
(\LL^{\theta_1}\hat{\varphi}-\LL^{\theta_2}\hat{\varphi})(\rho)&=\lim_{r\to0}\left(\int_{re^{i\theta_1}}^{e^{i\theta_1}\infty}-\int_{re^{i\theta_2}}^{e^{i\theta_2}\infty}\right)\hat{\varphi}(t)e^{-\rho t}dt\\
&=\lim_{r\to0}\int_{re^{i\theta_1}}^{re^{i\theta_2}}\hat{\varphi}(t)e^{-\rho t}dt-\lim_{R\to \infty}\int_{Re^{i\theta_1}}^{Re^{i\theta_2}}\hat{\varphi}(t)e^{-\rho t}dt.
\end{aligned}$$
The first term vanishes because $\lim_{t\to 0}|t\hat{\varphi}(t)|=0$, which can be deduced from the integrability of $\hat{\varphi}(t)$ near $0$; the second term also vanishes, by (ii) in Definition \ref{def:Laplace transform}. See (\cite[Lemma 5.31]{sauzinbook}) for more details.
\end{proof}

By comparing Propositions \ref{prop:laplace Theta propty} and \ref{proposition:Dirac-theta property} with the conditions in Definition \ref{def:Laplace transform}, we observe that the primary problem encountered when applying the Laplace transform to the Theta series lies in the lack of integrability near the origin, $\hat{\Theta}_{D_X},\hat{\Theta}_{-\Delta_X+C}\notin\mathcal{N}\left(-\frac{\pi}{2},\frac{\pi}{2}\right)$. To address this, we propose two methods to modify the Theta series to satisfy Definition \ref{def:Laplace transform} (i): the first one is to multiply by a suitable power $t^m$ to remove the divergence at $0$; the second one is to shift the domain.

\begin{proposition}{\label{prop:Theta series in N}}
 For $m\geq d+1$, $m^{\prime}\geq\lfloor\frac{d}{2}\rfloor+1$ and $s_0<0$, we have
 $$t^{m-1}\hat{\Theta}_{D_X}(t),~~t^{m^{\prime}-1}\hat{\Theta}_{-\Delta_X+C}(t),~~\hat{\Theta}_{D_X}(t-s_0),~~ \hat{\Theta}_{-\Delta_X+C}(t-s_0)~~\in\mathcal{N}\left(\left(-\frac{\pi}{2},\frac{\pi}{2}\right),0\right).$$
\end{proposition}
\begin{proof}
We only prove the first one and omit the other similar cases. For each $\theta\in\left(-\frac{\pi}{2},\frac{\pi}{2}\right)$, one can choose a positive real number $\delta=\delta(\theta)$ such that $|t|^{m-1}\leq e^{|t|\frac{\rho_1\cos\theta}{2}}$ for all $|t|>\delta$. Then $|t^{m-1}\hat{\Theta}_{D_X}(t)|\leq C'(\delta\cos\theta)e^{-|t|\frac{\rho_1\cos\theta}{2}}\leq C'(\delta\cos\theta)$ for $\Re(t)>\delta\cos\theta$ by Proposition \ref{proposition:Dirac-theta property} (iii). One can verify $C'(\delta)=r+\hat{\Theta}_{D_X}(\delta/2)$ for some number $r$ (depending on $\rho_1$), similar to the proof of (iii) in Proposition \ref{prop:laplace Theta propty}, so $C'(\delta\cos\theta)$ is a locally bounded function in $\theta$. Furthermore, $t^{m-1}\hat{\Theta}_{D_X}$ is holomorphic in $\Re(t)>0$ by Proposition \ref{proposition:Dirac-theta property} (i), which means that it is bounded in any compact region in $\Re(t)>0$. Hence, there is a locally bounded function $\alpha_1(\theta)$, such that $|t^{m-1}\hat{\Theta}_{D_X}(t)|\leq\alpha_1(\theta)$ for $0<|t|\leq\delta$, $\arg(t)=\theta$. Now, the function $\alpha(\theta):=\max\{\alpha_1(\theta),C'(\delta\cos\theta)\}$ is what we want.
\end{proof}

It is helpful to list all the necessary calculations as follows:
\begin{corollary}{\label{relation:laplace relation four}}
For $I=\left(-\frac{\pi}{2},\frac{\pi}{2}\right)$, $m\geq d+1$, $m^{\prime}\geq\lfloor\frac{d}{2}\rfloor+1$ and $s_0<0$, we have
\begin{itemize}
\item[(i)]$\LL^I\big((-t)^{m^{\prime}-1}\hat{\Theta}_{-\Delta_X+C}(t)\big)(\lambda)=\sum_{n\geq1}\frac{(-1)^{m^{\prime}-1}(m^{\prime}-1)!}{(\lambda+\lambda_n)^{m^{\prime}}}\text{, for }\lambda\in\mathcal{D}(I,0)$;

\item[(ii)]$\LL^I\big(\hat{\Theta}_{-\Delta_X+C}(t-s_0)\big)(\lambda)=\sum_{n\geq1}\frac{e^{\lambda_ns_0}}{\lambda+\lambda_n}\text{, for }\lambda\in\mathcal{D}(I,0)$;

\item[(iii)]$\LL^I\big((-t)^{m-1}\hat{\Theta}_{D_X}(t)\big)(\rho)=\sum_{n\geq1}\frac{(-1)^{m-1}(m-1)!}{(\rho+\rho_n)^m}\text{, for }\rho\in\mathcal{D}(I,0)$;

\item[(iv)]$\LL^I\big(\hat{\Theta}_{D_X}(t-s_0)\big)(\rho)=\sum_{n\geq1}\frac{e^{\rho_n s_0}}{\rho+\rho_n}\text{, for }\rho\in\mathcal{D}(I,0)$.
\end{itemize}
\end{corollary}
\begin{proof}
 We can exchange summation and integration by the Lebesgue dominated convergence theorem. For example, let $\theta\in\left(-\frac{\pi}{2},\frac{\pi}{2}\right)$, one can compute
$$
\LL^{\theta}\left((-t)^{m-1}\hat{\Theta}_{D_X}(t)\right)(\rho)=\int_0^{e^{i\theta}\infty}(-t)^{m-1}\sum_{n\geq1}e^{-\rho_nt}e^{-\rho t}dt=\sum_{n\geq1}\frac{(-1)^{m-1}(m-1)!}{(\rho+\rho_n)^m}.
$$
The calculations are routine, and we leave it to the reader.
\end{proof}

$(i)$ and $(ii)$ in Corollary \ref{relation:laplace relation four} provide two holomorphic functions with domain $\mathcal{D}\left((-\frac{\pi}{2},\frac{\pi}{2}),0\right)$, which lead to the definitions of two regularized determinants. Based on the singularities of $(-t)^{m-1}\hat{\Theta}_{D_X}(t)$ and $\hat{\Theta}_{-D_X}(t-s_0)$, Corollary \ref{relation:laplace relation four} $(iii)$ and $(iv)$ can be used to calculate these regularized determinants respectively, see \S\ref{section:computation about derivative regularization} and \S\ref{section:The computation of  exp deformed determinant}.

As the end of this subsection, we present a useful fact.

\begin{proposition}[{\cite[Lemma 5.17]{sauzinbook}}]\label{prop:multiple in borel and derivative in laplace}
If $\hat{\varphi}\in\mathcal{N}(I,\gamma)$, then $t\hat{\varphi}\in\mathcal{N}(I,\gamma)$, and
$$
\LL^I(-t\hat{\varphi})(\rho)=\frac{d}{d\rho}(\LL^I\hat{\varphi})(\rho),\quad\rho\in\mathcal{D}(I,\gamma).
$$
 \end{proposition}

\subsection{Two Regularizations for Determinant}
In this subsection, we introduce two regularized determinants.

Recall that the determinant series for the shifted Laplacian $-\Delta_X+C$ is the formal product
\[\HDet_{-\Delta_X+C}(\lambda):=\prod_{n=1}^{\infty}(\lambda+\lambda_n)\]
which is obviously divergent. Corollary \ref{relation:laplace relation four} provides two distinct approaches to regularize it.

For the formal product $A(\lambda)=\prod_na_n(\lambda)$ and the formal sum $B(\lambda)=\sum_mb_m(\lambda)$, we use  $\widetilde{\log}A(\lambda):=\sum_n\log a_n(\lambda)$ and  $\widetilde{\frac{d}{d\lambda}}B(\lambda):=\sum_m\frac{d}{d\lambda}b_m(\lambda)$ to denote
the \textbf{formal logarithm} and
\textbf{formal derivative} respectively.

The first regularization comes from Corollary \ref{relation:laplace relation four} $(i)$.
\begin{definition}[Logarithmic Derivative Regularization]{\label{def:deriva reg}}
The \textbf{formal logarithmic determinant} of $-\Delta_X+C$ is defined as a formal sum $\widetilde{\log}\HDet_{-\Delta_X+C}(\lambda)$. Furthermore, for $m\geq1$, the \textbf{$m$-th formal logarithmic derivative determinant} is defined as
$$
\widetilde{\frac{d^m}{d\lambda^m}}\widetilde{\log}\HDet_{-\Delta_X+C}(\lambda):=\underset{m}{\underbrace{\widetilde{\frac{d}{d\lambda}}\cdots\widetilde{\frac{d}{d\lambda}}}}\widetilde{\log}\HDet_{-\Delta_X+C}(\lambda)=\sum_{n=1}^{\infty}\frac{d^m}{d\lambda^m}\log(\lambda+\lambda_n)=\sum_{n=1}^{\infty}\frac{(-1)^{m-1}(m-1)!}{(\lambda+\lambda_n)^m}.
$$
According to Lemma \ref{lem:astomptotic of spectrum}, if $m\geq\lfloor\frac{d}{2}\rfloor+1$, the formal series $\widetilde{\frac{d^m}{d\lambda^m}}\widetilde{\log}\HDet_{-\Delta_X+C}(\lambda)$ is convergent in $\Re(\lambda)>0$. In this case, we call the sum function the \textbf{logarithmic derivative regularization of the determinant}.
\end{definition}
As analyzed in (\cite{voros1987spectral}), this sum function is equal to the $m$-th derivative of the zeta-regularized determinant or the Weierstrass product. We will characterize this function in terms of the singularities of $(-t)^{m-1}\hat{\Theta}_{D_X}(t)$ in \S\ref{section:computation about derivative regularization}.

The second regularization is a slight modification of Corollary \ref{relation:laplace relation four} $(ii)$.
\begin{definition}[Exponentially Deformed Regularization]{\label{def:exp reg}}
Let $s_0<0$, denote
$$
\widetilde{\log}\HDet^{\sharp_1}_{-\Delta_X+C}(\lambda):=\sum_{n=1}^{\infty}e^{s_0\lambda_n}\log(\lambda+\lambda_n)\quad\text{and}\quad\widetilde{\log}\HDet^{\sharp_2}_{-\Delta_X+C}(\lambda):=\sum_{n=1}^{\infty}e^{s_0\sqrt{\lambda_n}}\log(\lambda+\lambda_n).
$$
where $\lambda_n\geq 0$ are the spectrums of $-\Delta_X+C$. The corresponding formal derivatives are
$$
\widetilde{\frac{d}{d\lambda}}\widetilde{\log}\HDet^{\sharp_1}_{-\Delta_X+C}(\lambda)=\sum_{n=1}^{\infty}\frac{e^{s_0\lambda_n}}{\lambda+\lambda_n}\quad\text{and}\quad\widetilde{\frac{d}{d\lambda}}\widetilde{\log}\HDet_{-\Delta_X+C}^{\sharp_2}(\lambda)=\sum_{n=1}^{\infty}\frac{e^{s_0\sqrt{\lambda_n}}}{\lambda+\lambda_n},
$$
they converge uniformly and holomorphically in any compact subsets of $\mathbb{C}\setminus\{-\lambda_1,-\lambda_2,\cdots\}$. We refer to the sum function $\widetilde{\frac{d}{d\lambda}}\widetilde{\log}\HDet_{-\Delta_X+C}^{\sharp_2}(\lambda)$ as the \textbf{exp deformed regularization of the determinant}.
\end{definition}
The second regularization is motivated by (\cite[pp. 124--126, 154--158]{sauzinbook}). In \S\ref{section:The computation of exp deformed determinant}, we will characterize it in terms of the discrete singularities of $\hat{\Theta}_{D_X}(t-s_0)$ based on the Stokes structure (Lemma \ref{lemma: Laplace and stokes}) and Corollary \ref{relation:laplace relation four} $(iv)$.

If $X=S^1$, the first regularization recovers the classical Poisson summation formula (Theorem \ref{theorem:PSF}), while the second regularization provides its deformed version (Theorem \ref{thm: deformed psf}).

\subsection{Resurgence Theory}{\label{resu theoty}}
Resurgence theory is a powerful framework to study divergent series systematically through their singularity structures in Borel plane. It was initiated by J. Ecalle (\cite{ecallebook1,ecallebook2,ecallebook3}) and was further developed by Pham (\cite{pham1987introduction,pham1999resurgent}) and Sauzin (\cite{sauzinbook}) among others. Here, we follow \cite{sauzinbook} to briefly review this theory with appropriate modifications adapted to our setting.

As an illustrative example, consider the series $$\sum_{i\geq0}\sum_{n=1}^{\infty}(-1)^ie^{s_0\lambda_n}\lambda_n^i\cdot\lambda^{-i-1},$$ which is the formal expansion of $\widetilde{\frac{d}{d\lambda}}\widetilde{\log}\HDet^{\sharp_1}_{-\Delta_X+C}(\lambda):=\sum_{n=1}^{\infty}\frac{e^{s_0\lambda_n}}{\lambda+\lambda_n}$.
In order to study the relationship between this illegal expansion and the original object, we need to introduce formal Borel transform and $1$-Gevrey asymptotic expansion.

\begin{definition}[{\cite[Definition 5.6]{sauzinbook}}]{\label{Def:borel transform}}
The \textbf{formal Borel transform} is a $\CC$-linear map $\mathcal{B}:\rho^{-1}\CC[\![\rho^{-1}]\!]\to\CC[\![t]\!]$, given by
$$\mathcal{B}:\tilde{\varphi}=\sum_{i=0}^{\infty}a_i\rho^{-i-1}\longmapsto\hat{\varphi}=\sum_{i=0}^{\infty}a_i\frac{t^i}{i!}.$$
\end{definition}

We mainly focus on those $\hat{\varphi}=\mathcal{B}\tilde{\varphi}\in\CC[\![t]\!]$, whose sum functions exist locally at the origin and belong to $\mathcal{N}(I,\gamma)$ for some $I$ and $\gamma$ (see Definition \ref{def:Laplace transform}). By Proposition \ref{def:Laplace transform is welldefined}, these functions admit well-defined Laplace transforms. One can compare $\tilde{\varphi}$ and $\LL^I\hat{\varphi}$ through the $1$-Gevrey asymptotic expansion.

\begin{definition}[{\cite[Definition 5.21]{sauzinbook}}]\label{def:1-Gevrey expansion}
Let $\mathcal{D}\subseteq\CC\setminus\{0\}$ be an unbounded domain. Given a function $\varphi:\mathcal{D}\to\CC$ and a formal series $\tilde{\varphi}(\rho)=\sum_{i\geq0}a_i\rho^{-i-1}\in\rho^{-1}\CC[\![\rho^{-1}]\!]$, we say that \textbf{$\varphi$ admits $\tilde{\varphi}$ as uniform $1$-Gevrey asymptotic expansion in $\mathcal{D}$}, denoted by
$$\varphi(\rho)\sim_1\tilde{\varphi}(\rho)\quad\text{uniformly for }\rho\in\mathcal{D},$$
if there exist $L,M>0$ such that
$$|\varphi(\rho)-a_0\rho^{-1}-\cdots-a_{N-1}\rho^{-N}|\leq LM^NN!|\rho|^{-N-1},\quad\text{for all }\rho\in\mathcal{D},N\in\ZZ_{>0}.$$
\end{definition}

\begin{proposition}[{\cite[\S5.9]{sauzinbook}}]{\label{prop:1-Gevrey asymptotic}}
If there exists a locally bounded function $\gamma:I\to\RR$ such that $\hat{\varphi}=\mathcal{B}\tilde{\varphi}\in\mathcal{N}(I,\gamma)$ and holomorphic at $0$, then 
\begin{equation}{\label{eqn: uniformly 1-gevrey}}
\LL^I\hat{\varphi}(\rho)\sim_1\tilde{\varphi}(\rho)\quad\text{uniformly for }\rho\in\mathcal{D}(J,\gamma|_J),~J\subseteq I\text{ relatively compact subinterval}.
\end{equation}
For convenience, we denote \eqref{eqn: uniformly 1-gevrey} by $$\LL^I\hat{\varphi}(\rho)\sim_1\tilde{\varphi}(\rho) \text{ for } \rho\in\mathcal{D}(I,\gamma).$$
$$
\xymatrix{(\LL^I\hat{\varphi})(\rho)\in\mathcal{O}(\mathcal{D}(I,\gamma))\\&\hat{\varphi}=\sum_{i\geq0}a_i\frac{t^i}{i!}\in\mathcal{N}(I,\gamma)\cap\mathcal{O}_0\ar[ul]_{\LL^I}\\\tilde{\varphi}(\rho)=\sum_{i\geq0}a_i\rho^{-i-1}\in\rho^{-1}\CC[\![\rho^{-1}]\!]\ar[ur]_{\text{Borel}}\ar@{-->}[uu]^{\text{$1$-Gevrey}}}
$$
\end{proposition}

One can show (look at Theorem \ref{thm:deform 1-Gevrey expansion} for a proof),
$$\mathcal{B}\left(\sum_{i\geq0}\sum_{n=1}^{\infty}(-1)^ie^{s_0\lambda_n}\lambda_n^i\cdot\lambda^{-i-1}\right)=\hat{\Theta}_{-\Delta_X+C}(t-s_0)\in\mathcal{N}\left(-\frac{\pi}{2},\frac{\pi}{2}\right),$$ which is holomorphic at $t=0$. Then by Corollary \ref{relation:laplace relation four} $(ii)$, Definition \ref{def:exp reg} and Proposition \ref{prop:1-Gevrey asymptotic}, we have
$$\left(\LL^{(-\frac{\pi}{2},\frac{\pi}{2})}\hat{\Theta}_{-\Delta_X+C}(t-s_0)\right)(\lambda)=\widetilde{\frac{d}{d\lambda}}\widetilde{\log}\HDet_{-\Delta_X+C}^{\sharp_1}(\lambda)\sim_1\sum_{i\geq0}\sum_{n=1}^{\infty}(-1)^ie^{s_0\lambda_n}\lambda_n^i\cdot\lambda^{-i-1}$$
for $\lambda\in\mathcal{D}\left(\left(-\frac{\pi}{2},\frac{\pi}{2}\right),0\right)=\CC\setminus\RR_{\leq0}$. Moreover, direct computation gives us
$$\sum_{i\geq0}\sum_{n=1}^{\infty}(-1)^ie^{s_0\lambda_n}\lambda_n^i\cdot\lambda^{-i-1}=\sum_{i\geq0}\frac{d^i}{ds^i}\hat{\Theta}_{-\Delta_X+C}(s)\Big|_{s=-s_0}\cdot\lambda^{-i-1}.$$
Hence, the Laplace transform of $\hat{\Theta}_{-\Delta_X+C}(t-s_0)$ admits $\sum_{i\geq0}\frac{\frac{d^i}{ds^i}\hat{\Theta}_{-\Delta_X+C}(s)|_{-s_0}}{\lambda^{i+1}}\in\lambda^{-1}\CC[\![\lambda^{-1}]\!]$ as uniform $1$-Gevrey asymptotic expansion in $\mathcal{D}(J,0)$, where $J\subseteq\left(-\frac{\pi}{2},\frac{\pi}{2}\right)$ is relatively compact. For $\rho_n=\sqrt{\lambda_n}$, we can still obtain similar asymptotic formulas for $\hat{\Theta}_{D_X}(t-s_0)$ (Corollary \ref{deform 1-Gevrey thm about rho}).
$$
\xymatrix{\widetilde{\frac{d}{d\lambda}}\widetilde{\log}\HDet_{-\Delta_X+C}^{\sharp_1}(\lambda)&&\widetilde{\frac{d}{d\rho}}\widetilde{\log}\HDet_{D_X}^{\sharp_1}(\rho)&\\&\hat{\Theta}_{-\Delta_X+C}(t-s_0)\ar[ul]_{\LL^{(-\frac{\pi}{2},\frac{\pi}{2})}}&&\hat{\Theta}_{D_X}(t-s_0)\ar[ul]_{\LL^{(-\frac{\pi}{2},\frac{\pi}{2})}}\\\sum_{i\geq0}\frac{\frac{d^i}{ds^i}\hat{\Theta}_{-\Delta_X+C}(s)|_{-s_0}}{\lambda^{i+1}}\ar@{-->}[uu]^{\text{$1$-Gevrey}}\ar[ur]_{\text{Borel}}&&\sum_{i\geq0}\frac{\frac{d^i}{ds^i}\hat{\Theta}_{D_X}(s)|_{-s_0}}{\rho^{i+1}}\ar@{-->}[uu]^{\text{$1$-Gevrey}}\ar[ur]_{\text{Borel}}&}
$$

The resurgence theory can also be used to study the singularities of $\hat{\varphi}$, especially those with only simple singularities. 

   Unless otherwise specified, we let $D$ denote an open disk in the complex plane. In particular, the open disk centered at the origin with radius $r$ is denoted by $D_r$.

\begin{definition}[{\cite[Definition 6]{Sauzinpaper}}]{\label{def: simple singularity}}

Let $\omega\in\CC$, we say that a function $\hat{\varphi}$, which is holomorphic in an open disc $D\subseteq\CC$ to which $\omega\in\overline{D}\setminus D$, \textbf{has a simple singularity at $\omega$}, if there exist a constant $C\in\CC$ and holomorphic germs $\hat{\Phi}(t),R(t)\in\mathcal{O}_0$, such that
    $$\hat{\varphi}(t)=\frac{C}{2\pi i(t-\omega)}+\frac{1}{2\pi i}\hat{\Phi}(t-\omega)\log(t-\omega)+R(t-\omega),\quad\text{for all }t\in D.$$
\end{definition}

We form a convolution $\CC$-algebra $\CC\boldsymbol{\delta}\oplus\mathcal{O}_0$, where the Dirac delta $\boldsymbol{\delta}$ is convolution unit\footnote{For each $I$, the Laplace transform $\LL^I$ maps $\boldsymbol{\delta}$ to $1$.}. For a function $\hat{\varphi}$ with a simple singularity at $\omega$, we associate one and only one vector $\mathrm{sing}_{\omega}(\hat{\varphi}):=C\boldsymbol{\delta}+\hat{\Phi}\in\CC\boldsymbol{\delta}\oplus\mathcal{O}_0$, which completely records the singularity of $\hat{\varphi}$ at $\omega$.

\begin{definition}{\label{Def:Omega Resurgent}}
Let $\Omega \subset \mathbb{C}$ be a non-empty, discrete and closed subset containing the origin. A function $\hat{\varphi}$ integrable near the origin is called \textbf{$\Omega_{int}$-continuable}\footnote{In this paper, we relax the condition that the resurgent function be holomorphic at $0$ (see \cite[Definition 6.1]{sauzinbook}) to being integrable near $0$, and the corresponding notation will be slightly modified accordingly.}, if there exists an open disk $D_r^*\subseteq\CC$ centered at $0$ and punctured at $0$, such that
\begin{itemize}
\item $D_r^*\cap\Omega=\emptyset$;

\item $\hat{\varphi}\in\mathcal{O}(D_r^*)$;

\item For each path $\gamma$ in $\CC\setminus\Omega$ with initial point in $D_r^*$, $\hat{\varphi}$ admits a holomorphic extension $\mathrm{cont}_{\gamma}\hat{\varphi}$ along $\gamma$.
\end{itemize}
The set of all $\Omega_{int}$-continuable functions is denoted as $\hat{\mathcal{R}}_{\Omega}^{int}$.
\end{definition}

\begin{definition}
For $\Omega$ as in Definition \ref{Def:Omega Resurgent}, we define $\hat{\mathcal{R}}_{\Omega^*}^{simp}$ as the set of $\hat{\varphi}\in\hat{\mathcal{R}}_{\Omega}^{int}$ such that for each $\omega\in\Omega^{*}$, and any path $\gamma \subseteq\mathbb{C}\setminus\Omega$ from $D_r^*$ (as in Definition \ref{Def:Omega Resurgent}) to $D$ (as in Definition \ref{def: simple singularity}), the holomorphic extension $\mathrm{cont}_{\gamma}\hat{\varphi}$ has a simple singularity at $\omega$.
\end{definition}

In resurgence theory, singularities of an analytically continued germ $\hat{\varphi}$ along different paths are manifested by their alien operators ({\cite[Definition 6.48, Definition 6.62, pp. 227]{sauzinbook}). For our problem, the non-empty closed, discrete set is of the form 
$$\Omega_{s}:=\{0\}\cup\{s+i\tau_k \mid \tau_0=0, \tau_k < \tau_{k+1}, k\in\mathbb{Z}\}$$
for a fixed $s\in\mathbb{R}_{\leq 0}$. When $s=0$, it is {assumed} to be the singular set of $\hat{\varphi}(t)=t^d\hat{\Theta}_{D_X}(t)$(see Theorem \ref{theorem about derivative regularization in introduction}), and when $s=s_0<0$, it is assumed to be the singular set of $\hat{\varphi}(t)=\hat{\Theta}_{D_X}(t-s_0)$(see Theorem \ref{thm about the computation of deform in intro}). In these cases, we define the alien operator $\Delta_{\omega_k}^-$ to describe the singularity at $\omega_k\in\Omega_s^+ := \Omega_s \cap \{\Im(\omega) > 0 \}$ along the given path shown in Figure \ref{fig: the extension path in alien}.

\begin{figure}[htbp]
    \centering 
\tikzset{every picture/.style={line width=0.75pt}} 

\begin{tikzpicture}[x=0.75pt,y=0.75pt,yscale=-1,xscale=1]

\draw [color={rgb, 255:red, 208; green, 2; blue, 27 }  ,draw opacity=1 ]   (285.75,163.51) -- (285.53,152.66) ;
\draw [shift={(285.49,150.66)}, rotate = 88.82] [color={rgb, 255:red, 208; green, 2; blue, 27 }  ,draw opacity=1 ][line width=0.75]    (10.93,-3.29) .. controls (6.95,-1.4) and (3.31,-0.3) .. (0,0) .. controls (3.31,0.3) and (6.95,1.4) .. (10.93,3.29)   ;
\draw [color={rgb, 255:red, 208; green, 2; blue, 27 }  ,draw opacity=1 ]   (286.14,132.43) -- (285.91,115.97) ;
\draw [shift={(285.88,113.97)}, rotate = 89.21] [color={rgb, 255:red, 208; green, 2; blue, 27 }  ,draw opacity=1 ][line width=0.75]    (10.93,-3.29) .. controls (6.95,-1.4) and (3.31,-0.3) .. (0,0) .. controls (3.31,0.3) and (6.95,1.4) .. (10.93,3.29)   ;
\draw [color={rgb, 255:red, 208; green, 2; blue, 27 }  ,draw opacity=1 ]   (286.17,96.05) -- (286.18,87.45) -- (286.08,80.93) ;
\draw [shift={(286.05,78.93)}, rotate = 89.14] [color={rgb, 255:red, 208; green, 2; blue, 27 }  ,draw opacity=1 ][line width=0.75]    (10.93,-3.29) .. controls (6.95,-1.4) and (3.31,-0.3) .. (0,0) .. controls (3.31,0.3) and (6.95,1.4) .. (10.93,3.29)   ;
\draw [color={rgb, 255:red, 208; green, 2; blue, 27 }  ,draw opacity=1 ]   (285.82,61.02) -- (286.01,50.19) ;
\draw [shift={(286.05,48.19)}, rotate = 91.02] [color={rgb, 255:red, 208; green, 2; blue, 27 }  ,draw opacity=1 ][line width=0.75]    (10.93,-3.29) .. controls (6.95,-1.4) and (3.31,-0.3) .. (0,0) .. controls (3.31,0.3) and (6.95,1.4) .. (10.93,3.29)   ;
\draw  [draw opacity=0] (285.43,151.31) .. controls (281.23,150.85) and (278.05,146.48) .. (278.28,141.43) .. controls (278.52,136.33) and (282.16,132.45) .. (286.43,132.73) -- (286.07,142.03) -- cycle ; \draw  [color={rgb, 255:red, 208; green, 2; blue, 27 }  ,draw opacity=1 ] (285.43,151.31) .. controls (281.23,150.85) and (278.05,146.48) .. (278.28,141.43) .. controls (278.52,136.33) and (282.16,132.45) .. (286.43,132.73) ;  
\draw  [draw opacity=0] (286.51,78.95) .. controls (282.29,78.97) and (278.83,74.98) .. (278.72,69.92) .. controls (278.6,64.81) and (281.96,60.54) .. (286.23,60.33) -- (286.52,69.64) -- cycle ; \draw  [color={rgb, 255:red, 208; green, 2; blue, 27 }  ,draw opacity=1 ] (286.51,78.95) .. controls (282.29,78.97) and (278.83,74.98) .. (278.72,69.92) .. controls (278.6,64.81) and (281.96,60.54) .. (286.23,60.33) ;  
\draw  (227,174) -- (413.25,174)(348,30.86) -- (348,234.5) (406.25,169) -- (413.25,174) -- (406.25,179) (343,37.86) -- (348,30.86) -- (353,37.86)  ;
\draw [color={rgb, 255:red, 208; green, 2; blue, 27 }  ,draw opacity=1 ]   (335.33,168) -- (287.74,163.69) ;
\draw [shift={(285.75,163.51)}, rotate = 5.17] [color={rgb, 255:red, 208; green, 2; blue, 27 }  ,draw opacity=1 ][line width=0.75]    (10.93,-3.29) .. controls (6.95,-1.4) and (3.31,-0.3) .. (0,0) .. controls (3.31,0.3) and (6.95,1.4) .. (10.93,3.29)   ;
\draw  [color={rgb, 255:red, 155; green, 155; blue, 155 }  ,draw opacity=0.27 ][fill={rgb, 255:red, 155; green, 155; blue, 155 }  ,fill opacity=0.16 ][dash pattern={on 4.5pt off 4.5pt}] (278.8,47.03) .. controls (278.77,43.02) and (281.99,39.76) .. (286,39.76) .. controls (290.01,39.76) and (293.29,43.01) .. (293.32,47.02) .. controls (293.36,51.03) and (290.14,54.28) .. (286.13,54.28) .. controls (282.12,54.29) and (278.84,51.04) .. (278.8,47.03) -- cycle ;

\draw (278.45,136.55) node [anchor=north west][inner sep=0.75pt]  [rotate=-0.81]  {$\times $};
\draw (282.24,115.2) node [anchor=north west][inner sep=0.75pt]  [rotate=-269.11,xslant=0.05]  {$\cdots $};
\draw (350,177.4) node [anchor=north west][inner sep=0.75pt]    {$0$};
\draw (278.45,168.55) node [anchor=north west][inner sep=0.75pt]  [rotate=-0.81]  {$\times $};
\draw (278.45,64.55) node [anchor=north west][inner sep=0.75pt]  [rotate=-0.81]  {$\times $};
\draw (278.45,33.98) node [anchor=north west][inner sep=0.75pt]  [rotate=-0.81]  {$\times $};
\draw (251,135.4) node [anchor=north west][inner sep=0.75pt]    {$\omega _{1}$};
\draw (248,63.4) node [anchor=north west][inner sep=0.75pt]    {$\omega _{k-1}$};
\draw (249,31.4) node [anchor=north west][inner sep=0.75pt]    {$\omega _{k}$};
\draw (243,98.4) node [anchor=north west][inner sep=0.75pt]  [color={rgb, 255:red, 208; green, 2; blue, 27 }  ,opacity=1 ]  {$\gamma _{t_{1} ,t_{k}}^{-}$};
\draw (300,154.4) node [anchor=north west][inner sep=0.75pt]  [color={rgb, 255:red, 208; green, 2; blue, 27 }  ,opacity=1 ]  {$\gamma _{s,t_{1}}$};
\draw (272.33,174.07) node [anchor=north west][inner sep=0.75pt]    {$s$};
\draw (278.45,190.89) node [anchor=north west][inner sep=0.75pt]  [rotate=-0.81]  {$\times $};
\draw (278.45,214.55) node [anchor=north west][inner sep=0.75pt]  [rotate=-0.81]  {$\times $};
\draw (298.14,44.83) node [anchor=north west][inner sep=0.75pt]  [font=\normalsize,color={rgb, 255:red, 155; green, 155; blue, 155 }  ,opacity=0.38 ]  {$D$};
\draw (400,150.4) node [anchor=north west][inner sep=0.75pt]    {$\Re ( t)$};
\draw (355,30.4) node [anchor=north west][inner sep=0.75pt]    {$\Im ( t)$};

\end{tikzpicture}

\caption{The paths $\gamma_{t_1,t_k}^-$ and $\gamma_{s,t_1}$ corresponding to the alien operator $\Delta_{\omega_k}^-$.}
\label{fig: the extension path in alien}
\end{figure}
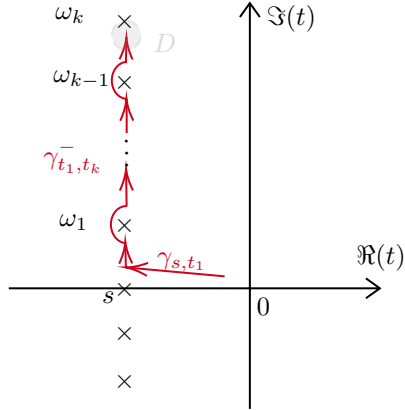

\begin{definition}{\label{Def:alien operator}}
For $\hat{\varphi} \in \hat{\mathcal{R}}_{\Omega_s}^{int}$, $\omega_k:=s+i\tau_k \in \Omega_s^+$. For any $t_1, t_k \in (0,1)$, let $\gamma_{s,t_1}$ be the straight line segment from $\frac{1}{4}\omega_1$ to $s+it_1\tau_1$, and let $\gamma_{t_1,t_k}^-$ be the path from $s+it_1\tau_1$ to $s+i(1-t_k)\tau_{k-1} + it_k\tau_{k}$ along the vertical line, but circumventing each $\omega_j$ ($1 \leq j \leq k-1$) from the left via {smaller} semi-circles (Figure \ref{fig: the extension path in alien}). If $\mathrm{cont}_{\gamma_{t_1,t_k}^- \circ \gamma_{s,t_1}} \hat{\varphi}$ has a simple singularity at $\omega_k$, we define the action of the alien operator on $\hat{\varphi}$ at $\omega_k$ by
$$(\Delta_{\omega_k}^-\hat{\varphi})(t) := -\mathrm{sing}_0 \left( \mathrm{cont}_{\gamma_{\omega_k}^-} \hat{\varphi}(\omega_k+t) \right), \quad \text{for } t \in -\omega_k+D,$$
where $\gamma_{\omega_k}^- :=\gamma_{t_1,t_k}^- \circ \gamma_{s,t_1}$. 
\end{definition}

Since $\hat{\varphi}$ is holomorphic on the open segments $(0, \omega_1)$ and $(\omega_{k-1}, \omega_k)$, this definition is independent of the choice of $t_1$ and $t_k$.
We illustrate the definition by the following example.

\begin{example}[{\cite[pp. 208]{sauzinbook}}]
If $\omega_k\in\Omega_s^+$ is a simple pole of $\hat{\varphi}$, then $(\Delta_{\omega_k}^-\hat{\varphi})(t)=-2\pi i\mathrm{Res}(\mathrm{cont}_{\gamma_{\omega_k}^-}\hat{\varphi};\omega_k)\cdot\boldsymbol{\delta}$. For more examples, especially regarding log-type singularities, see \cite[Exercise 6.51, Exercise 6.52, Exercise 6.57]{sauzinbook}.
\end{example}
The alien operator can be used to compute the discrepancy between the Laplace transforms in two directions. To suit our problems, we adapt the conditions in ({\cite[pp. 227-229]{sauzinbook}} and {\cite[pp. 84-85]{Sauzinpaper}}), and establish the results by deforming the integration contours.

\begin{lemma}{\label{lemma: Laplace and stokes}}
For a fixed $s \leq 0$, let $\Omega_s$ and $\Omega_s^+$ be defined as above. Fix two angles $\theta^+, \theta^-$ such that $0 < \theta^+ < \frac{\pi}{2} < \theta^- < \pi$. Let $a$ be the unique intersection point of the ray $[0, e^{i\theta^-}\infty)$ and the line $\{\text{Re}(t) = s\}$. Suppose there exists $0<\delta<\min(\theta^+,\pi-\theta^-)$ such that for $I^+ = (\theta^+ - \delta, \frac{\pi}{2})$ and $I^- = (\frac{\pi}{2}, \theta^- + \delta)$, the following hold:
\begin{itemize}
\item $\hat{\varphi}(t) \in \hat{\mathcal{R}}_{\Omega_s}^{int} \cap \mathcal{N}(I^+, \gamma)$ and $\hat{\varphi}(t+a)\in\mathcal{N}(I^+, \gamma) \cap \mathcal{N}(I^-, \gamma)$, where $\gamma$ is a locally bounded function on $I^+ \cup I^-$;
\item For any analytic continuation of $\hat{\varphi}$ along the path contained in the sector $\{ r e^{i \theta} : r > 0, \theta^{+}-\delta<\theta<\theta^{-}+\delta\}\setminus\Omega_{s}^+$, the extended $\hat{\varphi}$ has simple singularities at each $\omega_k \in \Omega_s^+$.
\end{itemize}
Then for any $\rho\in\Pi_{\gamma(\theta^+)}^{\theta^+}\cap\Pi_{\gamma(\theta^-)}^{\theta^-}$, we have
$$\LL^{\theta^-}\hat{\varphi}(\rho)=\LL^{\theta^+}\hat{\varphi}(\rho)+\sum_{\omega_k\in\Omega_s^+}e^{-\omega_k\rho}\LL^{\theta^+}\Delta_{\omega_k}^-\hat{\varphi}(\rho),$$
if the right hand side converges.
\end{lemma}
\begin{proof}
For $\hat{\varphi} \in \hat{\mathcal{R}}_{\Omega_s}^{int}$, we have $\lim_{r\to 0}\int_{re^{i\theta^-}}^{re^{i\theta^+}}\hat{\varphi}(t)e^{-\rho t}dt=0$. Additionally, since $\hat{\varphi} \in \mathcal{N}(I^+, \gamma)$, the integral along the arc at infinity also vanishes:$$\lim_{R\to\infty}\int_{Re^{i\theta_1}}^{Re^{i\theta_2}}\hat{\varphi}(t)e^{-\rho t}dt=0,$$
for any $\theta_1, \theta_2 \in I^+$ and $\rho \in \mathcal{D}((\theta_1, \theta_2), \gamma)$. Thus, for $\rho\in\Pi_{\gamma(\theta^+)}^{\theta^+}\cap\Pi_{\gamma(\theta^-)}^{\theta^-}$,
$$\left( (\mathcal{L}^{\theta^-} - \mathcal{L}^{\theta^+}) \hat{\varphi}\right)(\rho) = \left( \int_{a}^{a+e^{i\theta^-}\infty} - \int_{a}^{a+e^{i\theta^+}\infty} \right) \hat{\varphi}(t)e^{-\rho t}dt.$$
From $\hat{\varphi}(t) \in \hat{\mathcal{R}}_{\Omega_s}^{int}$, we have $\hat{\varphi}(t+a) \in \hat{\mathcal{R}}_{\Omega_{s-a}}^{int}$, where $\Omega_{s-a}:=\{s-a+i\tau_k:\tau_0=0,\tau_k<\tau_{k+1},k\in\ZZ\}$. Similarly, as $\hat{\varphi}(t+a)\in\mathcal{N}(I^+,\gamma)\cap\mathcal{N}(I^-,\gamma)$, $\left(\int_{a}^{a+e^{i\theta^-}\infty}-\int_{a}^{a+e^{i\theta^+}\infty}\right)\hat{\varphi}(t)e^{-\rho t}dt$ can be divided into the sum of integrals along infinite Hankel contours $\Gamma_{\omega_k,\epsilon}$ for any $0<\epsilon<\delta$ with $2\delta:=\inf_{\omega_k\in\Omega_{s}^+}d(\omega_k,\omega_{k+1})$. From the fact that $\Omega_s^+$ is a closed and discrete set, we have $\delta>0$.
 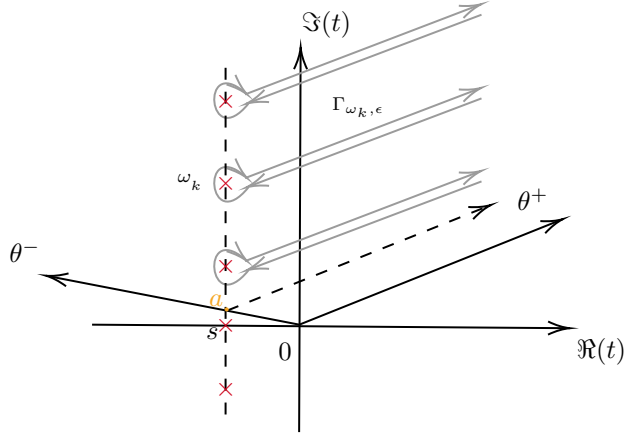
\begin{figure}[htbp]
    \centering    
\tikzset{every picture/.style={line width=0.75pt}} 
\begin{tikzpicture}[x=0.75pt,y=0.75pt,yscale=-1,xscale=1]

\draw    (227,162) -- (417,163.41) -- (464.78,163.76) ;
\draw [shift={(466.78,163.78)}, rotate = 180.42] [color={rgb, 255:red, 0; green, 0; blue, 0 }  ][line width=0.75]    (10.93,-3.29) .. controls (6.95,-1.4) and (3.31,-0.3) .. (0,0) .. controls (3.31,0.3) and (6.95,1.4) .. (10.93,3.29)   ;
\draw    (330.91,215.92) -- (331.66,24.83) ;
\draw [shift={(331.67,22.83)}, rotate = 90.22] [color={rgb, 255:red, 0; green, 0; blue, 0 }  ][line width=0.75]    (10.93,-3.29) .. controls (6.95,-1.4) and (3.31,-0.3) .. (0,0) .. controls (3.31,0.3) and (6.95,1.4) .. (10.93,3.29)   ;
\draw    (331,162) -- (372.56,145.36) -- (461.48,109.74) ;
\draw [shift={(463.33,109)}, rotate = 158.17] [color={rgb, 255:red, 0; green, 0; blue, 0 }  ][line width=0.75]    (10.93,-3.29) .. controls (6.95,-1.4) and (3.31,-0.3) .. (0,0) .. controls (3.31,0.3) and (6.95,1.4) .. (10.93,3.29)   ;
\draw    (331,162) -- (205.8,137.8) ;
\draw [shift={(203.83,137.42)}, rotate = 10.94] [color={rgb, 255:red, 0; green, 0; blue, 0 }  ][line width=0.75]    (10.93,-3.29) .. controls (6.95,-1.4) and (3.31,-0.3) .. (0,0) .. controls (3.31,0.3) and (6.95,1.4) .. (10.93,3.29)   ;
\draw [color={rgb, 255:red, 155; green, 155; blue, 155 }  ,draw opacity=1 ]   (422.44,90.11) -- (306.31,134.07) ;
\draw [shift={(304.44,134.78)}, rotate = 339.27] [color={rgb, 255:red, 155; green, 155; blue, 155 }  ,draw opacity=1 ][line width=0.75]    (10.93,-3.29) .. controls (6.95,-1.4) and (3.31,-0.3) .. (0,0) .. controls (3.31,0.3) and (6.95,1.4) .. (10.93,3.29)   ;
\draw [color={rgb, 255:red, 155; green, 155; blue, 155 }  ,draw opacity=1 ]   (303.78,130.11) -- (420.58,85.49) ;
\draw [shift={(422.44,84.78)}, rotate = 159.09] [color={rgb, 255:red, 155; green, 155; blue, 155 }  ,draw opacity=1 ][line width=0.75]    (10.93,-3.29) .. controls (6.95,-1.4) and (3.31,-0.3) .. (0,0) .. controls (3.31,0.3) and (6.95,1.4) .. (10.93,3.29)   ;
\draw [color={rgb, 255:red, 155; green, 155; blue, 155 }  ,draw opacity=1 ]   (305.44,134.78) .. controls (285.19,159.6) and (280.94,110.16) .. (303.37,128.88) ;
\draw [shift={(304.78,130.11)}, rotate = 222.51] [color={rgb, 255:red, 155; green, 155; blue, 155 }  ,draw opacity=1 ][line width=0.75]    (10.93,-3.29) .. controls (6.95,-1.4) and (3.31,-0.3) .. (0,0) .. controls (3.31,0.3) and (6.95,1.4) .. (10.93,3.29)   ;
\draw [color={rgb, 255:red, 155; green, 155; blue, 155 }  ,draw opacity=1 ]   (422.44,48.44) -- (306.31,92.4) ;
\draw [shift={(304.44,93.11)}, rotate = 339.27] [color={rgb, 255:red, 155; green, 155; blue, 155 }  ,draw opacity=1 ][line width=0.75]    (10.93,-3.29) .. controls (6.95,-1.4) and (3.31,-0.3) .. (0,0) .. controls (3.31,0.3) and (6.95,1.4) .. (10.93,3.29)   ;
\draw [color={rgb, 255:red, 155; green, 155; blue, 155 }  ,draw opacity=1 ]   (303.78,88.44) -- (420.58,43.82) ;
\draw [shift={(422.44,43.11)}, rotate = 159.09] [color={rgb, 255:red, 155; green, 155; blue, 155 }  ,draw opacity=1 ][line width=0.75]    (10.93,-3.29) .. controls (6.95,-1.4) and (3.31,-0.3) .. (0,0) .. controls (3.31,0.3) and (6.95,1.4) .. (10.93,3.29)   ;
\draw [color={rgb, 255:red, 155; green, 155; blue, 155 }  ,draw opacity=1 ]   (305.44,93.11) .. controls (285.19,117.94) and (280.94,68.49) .. (303.37,87.21) ;
\draw [shift={(304.78,88.44)}, rotate = 222.51] [color={rgb, 255:red, 155; green, 155; blue, 155 }  ,draw opacity=1 ][line width=0.75]    (10.93,-3.29) .. controls (6.95,-1.4) and (3.31,-0.3) .. (0,0) .. controls (3.31,0.3) and (6.95,1.4) .. (10.93,3.29)   ;
\draw [color={rgb, 255:red, 155; green, 155; blue, 155 }  ,draw opacity=1 ]   (422.44,6.44) -- (306.31,50.4) ;
\draw [shift={(304.44,51.11)}, rotate = 339.27] [color={rgb, 255:red, 155; green, 155; blue, 155 }  ,draw opacity=1 ][line width=0.75]    (10.93,-3.29) .. controls (6.95,-1.4) and (3.31,-0.3) .. (0,0) .. controls (3.31,0.3) and (6.95,1.4) .. (10.93,3.29)   ;
\draw [color={rgb, 255:red, 155; green, 155; blue, 155 }  ,draw opacity=1 ]   (303.78,46.44) -- (420.58,1.82) ;
\draw [shift={(422.44,1.11)}, rotate = 159.09] [color={rgb, 255:red, 155; green, 155; blue, 155 }  ,draw opacity=1 ][line width=0.75]    (10.93,-3.29) .. controls (6.95,-1.4) and (3.31,-0.3) .. (0,0) .. controls (3.31,0.3) and (6.95,1.4) .. (10.93,3.29)   ;
\draw [color={rgb, 255:red, 155; green, 155; blue, 155 }  ,draw opacity=1 ]   (305.44,51.11) .. controls (285.19,75.94) and (280.94,26.49) .. (303.37,45.21) ;
\draw [shift={(304.78,46.44)}, rotate = 222.51] [color={rgb, 255:red, 155; green, 155; blue, 155 }  ,draw opacity=1 ][line width=0.75]    (10.93,-3.29) .. controls (6.95,-1.4) and (3.31,-0.3) .. (0,0) .. controls (3.31,0.3) and (6.95,1.4) .. (10.93,3.29)   ;
\draw  [dash pattern={on 4.5pt off 4.5pt}]  (294.83,154.93) -- (336.39,138.29) -- (425.31,102.68) ;
\draw [shift={(427.17,101.93)}, rotate = 158.17] [color={rgb, 255:red, 0; green, 0; blue, 0 }  ][line width=0.75]    (10.93,-3.29) .. controls (6.95,-1.4) and (3.31,-0.3) .. (0,0) .. controls (3.31,0.3) and (6.95,1.4) .. (10.93,3.29)   ;
\draw  [dash pattern={on 4.5pt off 4.5pt}]  (294,33.08) -- (294,208) ;

\draw (319,169.4) node [anchor=north west][inner sep=0.75pt]    {$0$};
\draw (287.5,126.73) node [anchor=north west][inner sep=0.75pt]    {$\textcolor[rgb]{0.82,0.01,0.11}{\times }$};
\draw (287.5,84.97) node [anchor=north west][inner sep=0.75pt]  [color={rgb, 255:red, 0; green, 0; blue, 0 }  ,opacity=1 ]  {$\textcolor[rgb]{0.82,0.01,0.11}{\times }$};
\draw (287.5,44.07) node [anchor=north west][inner sep=0.75pt]    {$\textcolor[rgb]{0.82,0.01,0.11}{\times }$};
\draw (287.5,188.4) node [anchor=north west][inner sep=0.75pt]    {$\textcolor[rgb]{0.82,0.01,0.11}{\times }$};
\draw (439,92.4) node [anchor=north west][inner sep=0.75pt]    {$\theta ^{+}$};
\draw (184,118.4) node [anchor=north west][inner sep=0.75pt]    {$\theta ^{-}$};
\draw (346,46.4) node [anchor=north west][inner sep=0.75pt]  [font=\scriptsize]  {$\Gamma _{\omega _{k} ,\epsilon }$};
\draw (268.4,85) node [anchor=north west][inner sep=0.75pt]  [font=\scriptsize]  {$\omega _{_{k}}$};
\draw (282.73,161.81) node [anchor=north west][inner sep=0.75pt]    {$s$};
\draw (290.5,149.22) node [anchor=north west][inner sep=0.75pt]  [font=\Large,color={rgb, 255:red, 245; green, 166; blue, 35 }  ,opacity=1 ]  {$\cdot $};
\draw (284.23,145.5) node [anchor=north west][inner sep=0.75pt]  [color={rgb, 255:red, 245; green, 166; blue, 35 }  ,opacity=1 ]  {$a$};
\draw (287.5,156.5) node [anchor=north west][inner sep=0.75pt]    {$\textcolor[rgb]{0.82,0.01,0.11}{\times }$};
\draw (468.78,167.18) node [anchor=north west][inner sep=0.75pt]    {$\Re ( t)$};
\draw (330.78,3.18) node [anchor=north west][inner sep=0.75pt]    {$\Im ( t)$};

\end{tikzpicture}
    \caption{Deformation of the Laplace integration paths along $\theta^\pm$ into Hankel contours $\Gamma_{\omega_k,\epsilon}$ around singularities. }
    \label{fig:Hankel decomposition in stokes thm}
\end{figure}
Each contour $\Gamma_{\omega_k,\epsilon}$ can be defined as follows: it starts from $\omega_k+e^{i\theta^+}\infty$ and goes to $\omega_k+\epsilon e^{i\theta^+}$ along a straight line, and circles clockwise once around $\omega_k$, at last it returns to the $\omega_k+e^{i(\theta^+-2\pi)}\infty$ along the straight line from $\omega_k+\epsilon e^{i(\theta^+-2\pi)}$ (Figure \ref{fig:Hankel decomposition in stokes thm}). Thus,
$$(\LL^{\theta^-}\hat{\varphi}-\LL^{\theta^+}\hat{\varphi})(\rho)=\sum_{\omega_k\in\Omega_s^+}\int_{\Gamma_{\omega_k,\epsilon}}\mathrm{cont}_{\gamma_{\omega_k}^-}\hat{\varphi}(t)e^{-\rho t}dt,\quad\rho\in\Pi_{\gamma(\theta^+)}^{\theta^+}\cap\Pi_{\gamma(\theta^-)}^{\theta^-}.$$
Since the extended $\hat{\varphi}$ has simple singularities at each $\omega_k\in\Omega_{s}^+$, we have
$$\mathrm{cont}_{\gamma_{\omega_k}^-}\hat{\varphi}(t)=\frac{a_k}{2\pi i(t-\omega_k)}+\frac{1}{2\pi i}\hat{\Phi}_k(t-\omega_k)\log(t-\omega_k)+R_k(t-\omega_k),\quad a_k\in\CC;\hat{\Phi}_k,R_k\in\mathcal{O}_0,$$
near each singular point $\omega_k$, i.e. $$\hat{\Phi}_k(t)=\mathrm{cont}_{\gamma_{\omega_k}^-}\hat{\varphi}(\omega_k+t)-\mathrm{cont}_{\gamma_{\omega_k}^-}\hat{\varphi}(\omega_k+te^{-2\pi i})\in\mathcal{O}_0.$$
Set $M_{k,\delta}:=\sup_{t\in\{z:|z|<\delta\}}|\hat{\Phi}_k(t)|<\infty$,
then, for each path $S_{\omega_k,\epsilon}:=\{\omega_k+\epsilon e^{i\theta}:\theta\in(\theta^+-2\pi,\theta^+)\}\subseteq\Gamma_{\omega_k,\epsilon}$ and $\epsilon\to0$,
$$
\left|\int_{S_{\omega_k,\epsilon}}\hat{\Phi}_k(t-\omega_k)\log(t-\omega_k)e^{-\rho t}dt\right|\leq2\pi M_{k,\delta}e^{-\Re(\rho\omega_k)+|\rho|\epsilon}\epsilon\sqrt{\log^2\epsilon+(2\pi)^2}\in O(\epsilon\log\epsilon).
$$
For $\hat{\varphi}\in\hat{\mathcal{R}}_{\Omega_s}^{int}$, observe that, for any $0<\epsilon_1<\epsilon_2<\delta$, $\left(\int_{\Gamma_{\omega_k,\epsilon_1}}-\int_{\Gamma_{\omega_k,\epsilon_2}}\right)\hat{\varphi}(t)e^{-\rho t}dt=0$, which means that each term $\int_{\Gamma_{\omega_k,\epsilon}}\hat{\varphi}(t)e^{-\rho t}dt$ doesn't depend on the the $\epsilon>0$. Hence,
we obtain
$$\begin{aligned}
&\LL^{\theta^-}\hat{\varphi}-\LL^{\theta^+}\hat{\varphi}\\
=&\lim_{\epsilon\to 0}\sum_{\omega_k}\left(\int_{S_{\omega_k,\epsilon}}(\mathrm{cont}_{\gamma_{\omega_k}^-}\hat{\varphi})(t)e^{-\rho t}dt+\int_{\omega_k+\epsilon e^{i\theta^+}}^{\omega_k+e^{i\theta^+}\infty}e^{-\rho  t}\left((\mathrm{cont}_{\gamma_{\omega_k}^-}\hat{\varphi})(\omega_k+(t-\omega_k)e^{-2\pi i})-(\mathrm{cont}_{\gamma_{\omega_k}^-}\hat{\varphi})(t)\right)dt\right)\\
=&\sum_{\omega_k}\lim_{\epsilon\to 0}\left(\int_{S_{\omega_k,\epsilon}}(\mathrm{cont}_{\gamma_{\omega_k}^-}\hat{\varphi})(t)e^{-\rho t}dt+\int_{\omega_k+\epsilon e^{i\theta^+}}^{\omega_k+e^{i\theta^+}\infty}e^{-\rho  t}\left((\mathrm{cont}_{\gamma_{\omega_k}^-}\hat{\varphi})(\omega_k+(t-\omega_k)e^{-2\pi i})-(\mathrm{cont}_{\gamma_{\omega_k}^-}\hat{\varphi})(t)\right)dt\right)\\
=&\sum_{\omega_k}\left(-e^{-\rho\omega_k}\LL^{\theta^+}(a_k\boldsymbol{\delta})+e^{-\rho\omega_k}\int_0^{e^{i\theta^+}\infty}e^{-\rho t^{\prime}}\left((\mathrm{cont}_{\gamma_{\omega_k}^-}\hat{\varphi})(\omega_k+t^{\prime}e^{-2\pi i})-(\mathrm{cont}_{\gamma_{\omega_k}^-}\hat{\varphi})(t^{\prime}+\omega_k)\right)dt^{\prime}\right)\\
=&\sum_{\omega_k}e^{-\rho\omega_k}\int_0^{e^{i\theta^+}\infty}e^{-\rho t^{\prime}}\left(-\mathrm{sing}_0(\mathrm{cont}_{\gamma_{\omega_k}^-}\hat{\varphi})(\omega_k+t^{\prime})\right)dt^{\prime}\\
=&\sum_{\omega_k}e^{-\omega_k\rho}\LL^{\theta^+}\Delta_{\omega_k}^-\hat{\varphi}(\rho).
\end{aligned}$$
if the right side converges.
\end{proof}

This proposition is a generalization of the residue theorem.

\begin{corollary}{\label{coro: residue formula}}
Under the assumptions in Lemma \ref{lemma: Laplace and stokes}, suppose all singularities of $\hat{\varphi}$ are simple poles, then
$$\LL^{\theta^+}\hat{\varphi}=\LL^{\theta^-}\hat{\varphi}+2\pi i\sum_{\omega_k}e^{-\omega_k\rho}\mathrm{Res}(\mathrm{cont}_{\gamma_{\omega_k}^-}\hat{\varphi};\omega_k),\quad\rho\in\Pi_{\gamma(\theta^+)}^{\theta^+}\cap\Pi_{\gamma(\theta^-)}^{\theta^-},$$
if the right hand side converges.
 \end{corollary}


\section{The $1$-Gevrey Asymptotic Expansion 
}{\label{section:deformed Borel-laplace}}


In this section we apply resurgence theory to study the $1$-Gevrey asymptotic expansion of
$$\widetilde{\frac{d}{d\lambda}}\widetilde{\log}\HDet_{-\Delta_X+C}^{\sharp_1}(\lambda)=\sum_{n=1}^{\infty}\frac{e^{s_0\lambda_n}}{\lambda+\lambda_n}\quad\text{for }\lambda\in\mathcal{D}\left(\left(-\frac{\pi}{2},\frac{\pi}{2}\right),0\right)=\CC\setminus\RR_{\leq0}$$
appearing in Definition \ref{def:exp reg}. As pointed out in \S\ref{resu theoty}, the series
$$\sum_{i\geq0}\sum_{n=1}^{\infty}(-1)^ie^{s_0\lambda_n}\lambda_n^i\cdot\lambda^{-i-1}$$ is a formal expansion of $\sum_{n=1}^{\infty}\frac{e^{s_0\lambda_n}}{\lambda+\lambda_n}$,
whose radius of convergence is zero. We have following theorem to characterize their relationship.

\begin{theorem}{\label{thm:deform 1-Gevrey expansion}}
The series $\widetilde{\frac{d}{d\lambda}}\widetilde{\log}\HDet^{\sharp_1}_{-\Delta_X+C}(\lambda)$ converges uniformly and is holomorphic in any compact subset of $\mathbb{C}\setminus\{-\lambda_1,-\lambda_2,\cdots\}$. Furthermore,
$$
\widetilde{\frac{d}{d\lambda}}\widetilde{\log}\HDet^{\sharp_1}_{-\Delta_X+C}(\lambda)\sim_1\sum_{i\geq0}\left((-1)^i\sum_{n=1}^{\infty}e^{s_0\lambda_n}\lambda_n^i\right) \lambda^{-i-1}\quad{\text{for }}\lambda\in\CC\setminus\RR_{\leq0}.
$$
Or equivalently, for $\theta\in\left(-\frac{\pi}{2},\frac{\pi}{2}\right)$,
$$
\widetilde{\frac{d}{d\lambda}}\widetilde{\log}\HDet^{\sharp_1}_{-\Delta_X+C}(\lambda)=\LL^{\theta}\hat{\Theta}_{-\Delta_X+C}(t-s_0) \sim_1\sum_{i\geq0}\frac{d^i}{ds^i}\hat{\Theta}_{-\Delta_X+C}(s)\Big|_{s=-s_0}\cdot\lambda^{-i-1}.
$$
\end{theorem}
\begin{proof}
For any $i\geq0$,
$$a_i:=\sum_{n=1}^{\infty}e^{\lambda_ns_0}(-\lambda_n)^i=\frac{d^i}{ds^i}\hat{\Theta}_{-\Delta_X+C}(s)\Big|_{s=-s_0}\text{ converges if }s_0<0,$$
and there exist $\delta<|s_0|$ and $M>0$ such that
$$|a_i|=\left|\frac{d^i}{ds^i}\hat{\Theta}_{-\Delta_X+C}(s)\Big|_{s=-s_0}\right|\leq\frac{i!M}{\delta^i}$$
by the Cauchy inequality. Thus, the formal Borel transform of $\sum_{i\geq0}a_i\lambda^{-i-1}$ converges, and is holomorphic at $0$. That is,
\begin{align*}
\mathcal{B}\left(\sum_{i\geq0}a_i\lambda^{-i-1}\right)(t)=&\sum_{i\geq0}\frac{a_i}{i!}t^i\\=&\sum_{i\geq0}\frac{t^i}{i!}\frac{d^i}{ds^i}\hat{\Theta}_{-\Delta_X+C}(s)\Big|_{s=-s_0}\\=&\hat{\Theta}_{-\Delta_X+C}(t-s_0)\in\mathcal{N}\left(\left(-\frac{\pi}{2},\frac{\pi}{2}\right),0\right).
\end{align*}
By Proposition \ref{prop:Theta series in N} and Corollary \ref{relation:laplace relation four}, for $\theta\in\left(-\frac{\pi}{2},\frac{\pi}{2}\right)$,
$$\left(\LL^{\theta}\hat{\Theta}_{-\Delta_X+C}(t-s_0)\right)(\lambda)=\widetilde{\frac{d}{d\lambda}}\widetilde{\log}\HDet^{\sharp_1}_{-\Delta_X+C}(\lambda)$$
is holomorphic in $\Pi_0^{\theta}=\{\lambda\in\CC:\Re(\lambda e^{i\theta})>0\}$. Now, according to Proposition \ref{prop:1-Gevrey asymptotic}, we obtain
$$
\widetilde{\frac{d}{d\lambda}}\widetilde{\log}\HDet^{\sharp_1}_{-\Delta_X+C}(\lambda)\sim_1\sum_{i\geq0}a_i\lambda^{-i-1}
$$
for $\lambda\in\bigcup_{\theta\in(-\frac{\pi}{2},\frac{\pi}{2})}\Pi_0^{\theta}=\CC\setminus\RR_{\leq0}$.
\end{proof}

By Propositions \ref{prop:multiple in borel and derivative in laplace} and \ref{prop:1-Gevrey asymptotic}, the $m$-th derivative of the exp deformed regularization
of the determinant also admits a $1$-Gevrey expansion.

\begin{corollary}{\label{gevrey expansion of exp}}
For $m\in\ZZ_{\geq1}$, the series
$\frac{d^{m-1}}{d\lambda^{m-1}}\left(\widetilde{\frac{d}{d\lambda}}\widetilde{\log}\HDet^{\sharp_1}_{-\Delta_X+C}\right)(\lambda)$ converges uniformly and is holomorphic in any compact subset of $\CC\setminus\{-\lambda_1,-\lambda_2,\cdots\}$. Furthermore,
$$
\begin{aligned}
\frac{d^{(m-1)}}{d\lambda^{m-1}}\left(\widetilde{\frac{d}{d \lambda}}\widetilde{\log}\HDet^{\sharp_1}_{-\Delta_X+C}\right)(\lambda)
&=\LL^{(-\frac{\pi}{2},\frac{\pi}{2})}((-t)^{m-1}\hat{\Theta}_{-\Delta_X+C}(t-s_0))\\
&\sim_1\sum_{i\geq0}\frac{(-1)^{m-1}(m+i-1)!}{i!}\frac{d^i}{ds^i} \hat{\Theta}_{-\Delta_X+C}(s)\Big|_{s=-s_0}\cdot\lambda^{-i-m}\\
&=\sum_{i\geq0}\left(\frac{(-1)^{m+i-1}(m+i-1)!}{i!}\sum_{n=1}^{\infty}e^{s_0\lambda_n}\lambda_n^i\right)\lambda^{-m-i}
\end{aligned}
$$
for $\lambda\in\CC\setminus\RR_{\leq0}$.
\end{corollary}

For $\rho_n=\sqrt{\lambda_n}$, the same method can also be used to compute the formal derivative of
$$\widetilde{\log}\HDet_{D_X}^{\sharp_1}(\rho):=\sum_{n=1}^{\infty}e^{s_0\rho_n}\log(\rho+\rho_n),\quad s_0\in\RR_{<0},$$
where $\HDet_{D_X}(\rho):=\prod_{n=1}^{\infty}(\rho+\rho_n)$.

\begin{corollary}{\label{deform 1-Gevrey thm about rho}}
The series $\widetilde{\frac{d}{d\rho}}\widetilde{\log}\HDet^{\sharp_1}_{D_X}(\rho)$ converges uniformly and is holomorphic in any compact subset of $\CC\setminus\{-\rho_1,-\rho_2,\cdots\}$. Furthermore,
$$\widetilde{\frac{d}{d\rho}}\widetilde{\log}\HDet^{\sharp_1}_{D_X}(\rho)\sim_1\sum_{i\geq0}\left((-1)^i\sum_{n=1}^{\infty}e^{s_0\rho_n}\rho_n^i\right)\rho^{-i-1}\quad\text{for }\rho\in\CC\setminus\RR_{\leq0}.$$
Or equivalently, for $\theta\in\left(-\frac{\pi}{2},\frac{\pi}{2}\right)$,
$$\widetilde{\frac{d}{d\rho}}\widetilde{\log}\HDet^{\sharp_1}_{D_X}(\rho)=\LL^{\theta}\hat{\Theta}_{D_X}(t-s_0)\sim_1\sum_{i\geq0}\frac{d^i}{ds^i}\hat{\Theta}_{D_X}(s)\Big|_{s=-s_0}\cdot\rho^{-i-1}.$$
\end{corollary}


\section{Formal Logatithmic Derivative Determinant}{\label{section:computation about derivative regularization}}
In this section, we will prove the formula for computing the formal derivative determinant (Theorem \ref{theorem about derivative regularization in introduction}) by applying Stokes Lemma \ref{lemma: Laplace and stokes} with $$\hat{\varphi}(t)=t^d\hat{\Theta}_{D_X}(t).$$
As a result, the series $\widetilde{\frac{d^{m_0}}{d\rho^{m_0}}}\widetilde{\log}\HDet_{-\Delta_X+C}(\lambda)$ can be computed by the singularities and the analytic continuation of $(-t)^{m_0-1}\hat{\Theta}_{D_X}(t)$, where $m_0=\dim X+1=d+1$.


\vskip 3mm

\begin{proofof}{Theorem \ref{theorem about derivative regularization in introduction}}
By Proposition \ref{prop:Theta series in N} and Corollary \ref{relation:laplace relation four} $(iii)$, the Laplace transform of $(-t)^{m_0-1}\hat{\Theta}_{D_X}(t)$ along $\theta\in\left(-\frac{\pi}{2},\frac{\pi}{2}\right)$ converges and is holomorphic in $\Re{(\rho e^{i\theta})}>0$. It means that
\begin{equation}{\label{laplace and determinant}}
\LL^{\theta}\left((-t)^{m_0-1}\hat{\Theta}_{D_X}(t)\right)(\rho)=\sum_{n\geq1}\frac{(-1)^{m_0-1}(m_0-1)!}{(\rho+\rho_n)^{m_0}} =\widetilde{\frac{d^{m_0}}{d\rho^{m_0}}}\widetilde{\log}\HDet_{D_X}(\rho),
\end{equation}
if $\Re{(\rho e^{i\theta})}>0$ for $\theta\in\left(-\frac{\pi}{2},\frac{\pi}{2}\right)$. Since $\frac{d}{2\rho d\rho}=\frac{d}{d\lambda}$, by Definition \ref{def:deriva reg},
\begin{equation}{\label{eqn:decomposition of m-derivative}}
\begin{aligned}
\widetilde{\frac{d^{m_0}}{d\rho^{m_0}}}\widetilde{\log} \HDet_{-\Delta_X+C}(\rho^2)
&=\sum_{n\geq1}\frac{d^{m_0-1}}{d\rho^{m_0-1}}\frac{2\rho}{\rho^2+\rho_n^2}\\
&=\sum_{n\geq1}\left((-i)^{m_0}\frac{d^{m_0-1}}{d(-i\rho)^{m_0-1}}\frac{1}{\rho_n-i\rho}+i^{m_0}\frac{d^{m_0-1}}{d(i\rho)^{m_0-1}}\frac{1}{\rho_n+i\rho}\right)\\
&=(-i)^{m_0}\widetilde{\frac{d^{m_0}}{d(-i\rho)^{m_0}}}\widetilde{\log}\HDet_{D_X}(-i\rho)+i^{m_0}\widetilde{\frac{d^{m_0}}{d(i\rho)^{m_0}}}\widetilde{\log}\HDet_{D_X}(i\rho).
\end{aligned}
\end{equation}
\begin{figure}[htbp]
    \centering 
\tikzset{every picture/.style={line width=0.75pt}} 

\begin{tikzpicture}[x=0.75pt,y=0.75pt,yscale=-1,xscale=1]

\draw  (249,140.5) -- (407.67,140.5)(351.07,50.5) -- (351.07,239.5) (400.67,135.5) -- (407.67,140.5) -- (400.67,145.5) (346.07,57.5) -- (351.07,50.5) -- (356.07,57.5)  ;
\draw   (346.91,158.18) -- (354.85,166.31)(355.17,158.05) -- (346.59,166.44) ;
\draw   (346.91,178.18) -- (354.85,186.31)(355.17,178.05) -- (346.59,186.44) ;
\draw   (346.91,194.18) -- (354.85,202.31)(355.17,194.05) -- (346.59,202.44) ;
\draw   (346.91,213.18) -- (354.85,221.31)(355.17,213.05) -- (346.59,221.44) ;
\draw    (351,140.5) -- (321.66,56.89) ;
\draw [shift={(321,55)}, rotate = 70.67] [color={rgb, 255:red, 0; green, 0; blue, 0 }  ][line width=0.75]    (10.93,-3.29) .. controls (6.95,-1.4) and (3.31,-0.3) .. (0,0) .. controls (3.31,0.3) and (6.95,1.4) .. (10.93,3.29)   ;
\draw   (346.02,54.04) -- (355.55,62.27)(355.93,53.91) -- (345.64,62.4) ;
\draw   (346.02,118.04) -- (355.55,126.27)(355.93,117.91) -- (345.64,126.4) ;
\draw   (346.02,100.04) -- (355.55,108.27)(355.93,99.91) -- (345.64,108.4) ;
\draw   (346.02,78.04) -- (355.55,86.27)(355.93,77.91) -- (345.64,86.4) ;
\draw [color={rgb, 255:red, 155; green, 155; blue, 155 }  ,draw opacity=1 ][fill={rgb, 255:red, 80; green, 227; blue, 194 }  ,fill opacity=1 ] [dash pattern={on 4.5pt off 4.5pt}]  (351,140.5) -- (251.04,130.56) ;
\draw [color={rgb, 255:red, 155; green, 155; blue, 155 }  ,draw opacity=1 ] [dash pattern={on 4.5pt off 4.5pt}]  (351,140.5) -- (250.81,149.89) ;
\draw  [color={rgb, 255:red, 0; green, 0; blue, 0 }  ,draw opacity=0 ][fill={rgb, 255:red, 155; green, 155; blue, 155 }  ,fill opacity=0.17 ] (269.17,131.97) .. controls (273.06,88.91) and (306.95,54.81) .. (348.99,53.84) .. controls (394.37,52.78) and (432.06,90.73) .. (433.17,138.59) .. controls (434.28,186.45) and (398.39,226.11) .. (353.01,227.16) .. controls (309.28,228.18) and (272.69,192.98) .. (269.09,147.59) -- (351,140.5) -- cycle ;
\draw    (351.11,58.93) -- (350.83,81.42) ;
\draw [color={rgb, 255:red, 255; green, 255; blue, 255 }  ,draw opacity=1 ][line width=1.5]    (350.92,83.13) -- (350.92,103.13) ;
\draw [color={rgb, 255:red, 255; green, 255; blue, 255 }  ,draw opacity=1 ][line width=1.5]    (350.92,59.13) -- (350.83,81.42) ;
\draw [color={rgb, 255:red, 255; green, 255; blue, 255 }  ,draw opacity=1 ][line width=1.5]    (350.92,104.79) -- (350.92,121.21) ;
\draw [color={rgb, 255:red, 255; green, 255; blue, 255 }  ,draw opacity=1 ][line width=1.5]    (350.92,53.21) -- (350.92,57.72) ;
\draw [color={rgb, 255:red, 255; green, 255; blue, 255 }  ,draw opacity=1 ][line width=1.5]    (350.93,163.53) -- (350.93,180.73) ;
\draw [color={rgb, 255:red, 255; green, 255; blue, 255 }  ,draw opacity=1 ][line width=1.5]    (350.92,182.7) -- (350.93,196.64) ;
\draw [color={rgb, 255:red, 255; green, 255; blue, 255 }  ,draw opacity=1 ][line width=1.5]    (350.92,199.59) -- (350.92,216.01) ;
\draw [color={rgb, 255:red, 255; green, 255; blue, 255 }  ,draw opacity=1 ][line width=1.5]    (350.92,218.39) -- (351.07,240.73) ;

\draw (353,143.9) node [anchor=north west][inner sep=0.75pt]  [font=\scriptsize]  {$0$};
\draw (296,57.4) node [anchor=north west][inner sep=0.75pt]  [font=\tiny]  {$\frac{\pi }{2} +\epsilon $};
\draw (416.38,50.73) node [anchor=north west][inner sep=0.75pt]  [font=\scriptsize,xslant=0.02]  {$t^{m_{0} -1}\hat{\Theta }_{D_{X}}( t)$};
\draw (238.19,113.07) node [anchor=north west][inner sep=0.75pt]  [font=\tiny,color={rgb, 255:red, 155; green, 155; blue, 155 }  ,opacity=1 ]  {$\frac{\pi }{2} +\delta $};
\draw (232.85,153.07) node [anchor=north west][inner sep=0.75pt]  [font=\tiny,color={rgb, 255:red, 155; green, 155; blue, 155 }  ,opacity=1 ]  {$-\frac{\pi }{2} -\delta $};
\draw (262.2,77.67) node [anchor=north west][inner sep=0.75pt]  [font=\scriptsize,color={rgb, 255:red, 155; green, 155; blue, 155 }  ,opacity=1 ]  {$U_{\delta }$};
\draw (353.33,69.23) node [anchor=north west][inner sep=0.75pt]  [font=\tiny]  {$i\tau _{k}$};
\draw (391.53,146.48) node [anchor=north west][inner sep=0.75pt]  [font=\scriptsize]  {$\Re ( t)$};
\draw (349.93,35.15) node [anchor=north west][inner sep=0.75pt]  [font=\scriptsize]  {$\Im ( t)$};

\end{tikzpicture}
\caption{The holomorphic region $U_\delta$ and the singularities of $(-t)^d\hat{\Theta}_{D_X}(t)$.}
\label{fig: region U in derivative det computation}
\end{figure}
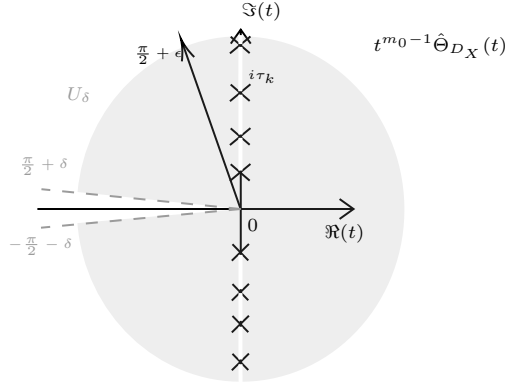
By \eqref{laplace and determinant}, for $-\frac{\pi}{2}+\epsilon<\arg(\rho)<\frac{\pi}{2}-\epsilon$,
$$
\begin{aligned}
&i^{m_0}\widetilde{\frac{d^{m_0}}{d(i\rho)^{m_0}}}\widetilde{\log}\HDet_{D_X}(i\rho)+(-i)^{m_0}\widetilde{\frac{d^{m_0}}{d(-i\rho)^{m_0}}}\widetilde{\log}\HDet_{D_X}(-i\rho)\\
=&i^{m_0}\int_0^{e^{(-\frac{\pi}{2}+\epsilon)i}\infty}(-t)^{m_0-1}\hat{\Theta}_{D_X}(t)e^{-i\rho t}dt+(-i)^{m_0}\int_0^{e^{(\frac{\pi}{2}-\epsilon)i}\infty}(-t)^{m_0-1}\hat{\Theta}_{D_X}(t)e^{i\rho t}dt\\
=&i^{m_0}\left(-\int_0^{e^{(\frac{\pi}{2}+\epsilon)i}\infty}t^{m_0-1}\hat{\Theta}_{D_X}(te^{-i\pi})e^{i\rho t}dt-\int_0^{e^{(\frac{\pi}{2}-\epsilon)i}\infty}t^{m_0-1}\hat{\Theta}_{D_X}(t)e^{i\rho t}dt\right)\\
=&i^{m_0}\left(\int_0^{e^{(\frac{\pi}{2}+\epsilon)i}\infty}t^{m_0-1}\hat{\Theta}_{D_X}(t)e^{i\rho t}dt-\int_0^{e^{(\frac{\pi}{2}-\epsilon)i}\infty}t^{m_0-1}\hat{\Theta}_{D_X}(t)e^{i\rho t}dt-\int_0^{e^{(\frac{\pi}{2}+\epsilon)i}\infty}t^{m_0-1}f(t)e^{i\rho t}dt\right)\\
=&i^{m_0}\left(\sum_{\tau_k>0}e^{-\tau_k\rho}\LL^{\left(\frac{\pi}{2}-\epsilon\right)}\left(\Delta_{i\tau_k}^-\left(t^{m_0-1}\hat{\Theta}_{D_X}(t)\right)\right)(-i\rho)\right)-i^{m_0}\LL^{\left(\frac{\pi}{2}+\epsilon\right)}\left(t^{m_0-1}f(t)\right)(-i\rho),
\end{aligned}
$$
where the last equality is obtained by Lemma \ref{lemma: Laplace and stokes}, if $t^d\hat{\Theta}_{D_X}(t)\in\mathcal{N}\left(\left(\frac{\pi}{2},\frac{\pi}{2}+\delta\right),0\right)\cap\hat{\mathcal{R}}_{\Omega^*}^{simp}$. Thus,
$$
\begin{aligned}
\widetilde{\frac{d^{m_0}}{d\rho^{m_0}}}\widetilde{\log}\HDet_{-\Delta_X+C}(\rho^2)=&i^{m_0}\left(\sum_{\tau_k>0}e^{-\tau_k\rho}\LL^{\left(\frac{\pi}{2}-\epsilon\right)}\left(\Delta_{i\tau_k}^-\left(t^{m_0-1}\hat{\Theta}_{D_X}(t)\right)\right)(-i\rho)\right)\\
&-i^{m_0}\LL^{\left(\frac{\pi}{2}+\epsilon\right)}\left(t^{m_0-1}f(t)\right)(-i\rho).
\end{aligned}
$$
\end{proofof}

\begin{remark}\label{remark on deriv them}
For conciseness, Theorem \ref{theorem about derivative regularization in introduction} is stated under relatively restrictive  conditions. However, several conditions can be relaxed to arrive the same conclusion.
\begin{itemize}
\item 
As amplified in the proof and the condition $(ii)$ of Lemma \ref{lemma: Laplace and stokes}, all arguments are restricted to $U_\delta$. Consequently, we can allow $t^d\hat{\Theta}_{D_X}$ to possess additional discrete singular points in $\mathbb{C} \setminus\left\{re^{i\theta}:r>0,-\frac{\pi}{2}-\delta<\theta<\frac{\pi}{2}+\delta\right\}$. Such additional singularities neither affect the derivation nor contribute to the sum in \eqref{formula:reg1 in intro}, provided that within $U_\delta$, $t^d\hat{\Theta}_{D_X}(t)$ remains holomorphic and has simple singularities at each $i\tau_k\in\Omega^*$. We will be in such a situation in \S\ref{section:example X}.

\item
The condition $t^d\hat{\Theta}_{D_X}(t)\in \mathcal{N}\left(\left(\frac{\pi}{2}, \frac{\pi}{2}+\delta\right), 0\right)$ in $(ii)$ of Theorem \ref{theorem about derivative regularization in introduction} can be replaced by $t^d\hat{\Theta}_{D_X}(t)\in \mathcal{N}\left(\left(\frac{\pi}{2}, \frac{\pi}{2}+\delta\right), \gamma\right)$ for a locally bounded function $\gamma: \left(\frac{\pi}{2}, \frac{\pi}{2}+\delta\right) \to \mathbb{R}$, whereby the resulting domain for formula \eqref{formula:reg1 in intro} (originally $\arg(\rho) \in (-\frac{\pi}{2}+\epsilon, \frac{\pi}{2}-\epsilon)$) is restricted to $\{\rho\in\CC:-\frac{\pi}{2}+\epsilon<\arg(\rho)<\frac{\pi}{2}+\epsilon\}\cap\{\rho\in \CC:\Re(\rho e^{i\epsilon})>\gamma(\frac{\pi}{2}+\epsilon)\}$.
\end{itemize}
\end{remark}
\section{Exponentially Deformed Determinant}{\label{section:The computation of  exp deformed determinant}}

Like in the last section \S\ref{section:computation about derivative regularization}, the exp-deformed determinant $\widetilde{\frac{d}{d\rho}}\widetilde{\log} \HDet^{\sharp_2}_{-\Delta_X+C}(\rho^2)$ can also be computed  by studying singularities of the analytic continuation of $\hat{\Theta}_{D_X}(t-s_0)$. We will prove Theorem \ref{thm about the computation of deform in intro} by applying Lemma \ref{lemma: Laplace and stokes} with $\hat{\varphi}(t)=\hat{\Theta}_{D_X}(t-s_0)$ in this section.

\vskip 3mm

\begin{proofof}{Theorem \ref{thm about the computation of deform in intro}}
Since $\rho^2=\lambda$, we have $\frac{d}{2\rho d\rho}=\frac{d}{d\lambda}$. By Definition \ref{def:exp reg},
$$
\begin{aligned}
\widetilde{\frac{d}{d\rho}}\widetilde{\log} \HDet^{\sharp_2}_{-\Delta_X+C}(\rho^2)
&=\sum_{n\geq1}\frac{2\rho e^{s_0\rho_n}}{\rho^2+\rho_n^2}\\
&=\sum_{n\geq1}(-i)\left(\frac{e^{s_0\rho_n}}{\rho_n-i\rho}-\frac{e^{s_0\rho_n}}{\rho_n+i\rho}\right)\\
&=(-i)\widetilde{\frac{d}{d(-i\rho)}}\widetilde{\log} \HDet_{D_X}^{\sharp_1}(-i\rho)+i\widetilde{\frac{d}{d(i\rho)}}\widetilde{\log} \HDet^{\sharp_1}_{D_X}(i\rho).
\end{aligned}
$$
By Corollary \ref{relation:laplace relation four} $(iv)$, for $-\frac{\pi}{2}+\epsilon<\arg(\rho)<\epsilon$, the right hand side of above formula becomes
$$
\begin{aligned}
&i\left(\int_0^{e^{(-\epsilon)i}\infty}\hat{\Theta}_{D_X}(t-s_0)e^{-i\rho t}dt -\int_0^{e^{(\frac{\pi}{2}-\epsilon)i}\infty}\hat{\Theta}_{D_X}(t-s_0)e^{i\rho t}dt\right)\\
=&i\left(-\int_0^{e^{(\pi-\epsilon)i}\infty}\hat{\Theta}_{D_X}(e^{-i\pi}t-s_0)e^{i\rho t}dt-\int_0^{e^{(\frac{\pi}{2}-\epsilon)i}\infty}\hat{\Theta}_{D_X}(t-s_0)e^{i\rho t}dt\right).
\end{aligned}
$$
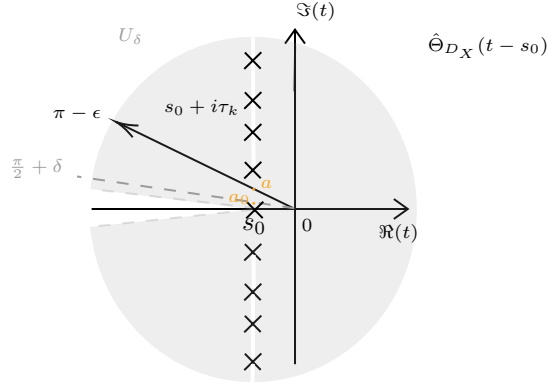
\begin{figure}[htbp]
    \centering

\tikzset{every picture/.style={line width=0.75pt}} 

\begin{tikzpicture}[x=0.75pt,y=0.75pt,yscale=-1,xscale=1]

\draw  (249,113.5) -- (407.67,113.5)(351.07,23.5) -- (351.07,212.5) (400.67,108.5) -- (407.67,113.5) -- (400.67,118.5) (346.07,30.5) -- (351.07,23.5) -- (356.07,30.5)  ;
\draw    (351.11,31.93) -- (350.83,54.42) ;
\draw [color={rgb, 255:red, 255; green, 255; blue, 255 }  ,draw opacity=1 ][line width=1.5]    (350.92,191.39) -- (351.07,213.73) ;
\draw   (325.91,131.11) -- (333.85,139.24)(334.17,130.99) -- (325.59,139.37) ;
\draw   (325.91,151.11) -- (333.85,159.24)(334.17,150.99) -- (325.59,159.37) ;
\draw   (325.91,167.11) -- (333.85,175.24)(334.17,166.99) -- (325.59,175.37) ;
\draw   (325.91,186.11) -- (333.85,194.24)(334.17,185.99) -- (325.59,194.37) ;
\draw    (351.07,113.5) -- (263.02,70.65) ;
\draw [shift={(261.22,69.78)}, rotate = 25.95] [color={rgb, 255:red, 0; green, 0; blue, 0 }  ][line width=0.75]    (10.93,-3.29) .. controls (6.95,-1.4) and (3.31,-0.3) .. (0,0) .. controls (3.31,0.3) and (6.95,1.4) .. (10.93,3.29)   ;
\draw  [color={rgb, 255:red, 0; green, 0; blue, 0 }  ,draw opacity=0 ][fill={rgb, 255:red, 155; green, 155; blue, 155 }  ,fill opacity=0.17 ] (248.37,103.46) .. controls (253.17,60.51) and (287.78,27.16) .. (329.83,27.08) .. controls (375.23,26.99) and (412.1,65.72) .. (412.19,113.57) .. controls (412.29,161.43) and (375.56,200.29) .. (330.17,200.38) .. controls (287.52,200.46) and (252.39,166.29) .. (248.22,122.47) -- (330,113.73) -- cycle ;
\draw [color={rgb, 255:red, 255; green, 255; blue, 255 }  ,draw opacity=1 ][line width=1.5]    (330,26.17) -- (330.18,83.61) -- (330.33,91.94) ;
\draw   (325.91,34.95) -- (333.85,43.07)(334.17,34.82) -- (325.59,43.2) ;
\draw   (325.91,54.95) -- (333.85,63.07)(334.17,54.82) -- (325.59,63.2) ;
\draw   (325.91,70.95) -- (333.85,79.07)(334.17,70.82) -- (325.59,79.2) ;
\draw   (325.91,89.95) -- (333.85,98.07)(334.17,89.82) -- (325.59,98.2) ;
\draw [color={rgb, 255:red, 155; green, 155; blue, 155 }  ,draw opacity=0.3 ][fill={rgb, 255:red, 80; green, 227; blue, 194 }  ,fill opacity=1 ] [dash pattern={on 4.5pt off 4.5pt}]  (331.67,113.1) -- (248.37,103.46) ;
\draw [color={rgb, 255:red, 155; green, 155; blue, 155 }  ,draw opacity=1 ][fill={rgb, 255:red, 80; green, 227; blue, 194 }  ,fill opacity=1 ] [dash pattern={on 4.5pt off 4.5pt}]  (350.67,113.1) -- (241.67,97.17) ;
\draw [color={rgb, 255:red, 155; green, 155; blue, 155 }  ,draw opacity=0.3 ][fill={rgb, 255:red, 80; green, 227; blue, 194 }  ,fill opacity=1 ] [dash pattern={on 4.5pt off 4.5pt}]  (330,113.73) -- (248.22,122.47) ;
\draw   (326.91,109.95) -- (334.85,118.07)(335.17,109.82) -- (326.59,118.2) ;
\draw [color={rgb, 255:red, 255; green, 255; blue, 255 }  ,draw opacity=1 ][line width=1.5]    (330.33,136.61) -- (330.33,153.28) ;
\draw [color={rgb, 255:red, 255; green, 255; blue, 255 }  ,draw opacity=1 ][line width=1.5]    (329.92,156.44) -- (330,169.28) ;
\draw [color={rgb, 255:red, 255; green, 255; blue, 255 }  ,draw opacity=1 ][line width=1.5]    (329.92,172.78) -- (330,188.17) ;
\draw [color={rgb, 255:red, 255; green, 255; blue, 255 }  ,draw opacity=1 ][line width=1.5]    (329.92,191.33) -- (329.67,200.94) ;

\draw (353,116.9) node [anchor=north west][inner sep=0.75pt]  [font=\scriptsize]  {$0$};
\draw (416.38,23.73) node [anchor=north west][inner sep=0.75pt]  [font=\scriptsize,xslant=0.02]  {$\hat{\Theta }_{D_{X}}( t-s_{0})$};
\draw (391.53,119.48) node [anchor=north west][inner sep=0.75pt]  [font=\scriptsize]  {$\Re ( t)$};
\draw (349.93,8.15) node [anchor=north west][inner sep=0.75pt]  [font=\scriptsize]  {$\Im ( t)$};
\draw (260.87,21.27) node [anchor=north west][inner sep=0.75pt]  [font=\scriptsize,color={rgb, 255:red, 155; green, 155; blue, 155 }  ,opacity=1 ]  {$U_{\delta }$};
\draw (283.83,57.17) node [anchor=north west][inner sep=0.75pt]  [font=\scriptsize]  {$s_{0} +i\tau _{k}$};
\draw (205.89,86.91) node [anchor=north west][inner sep=0.75pt]  [font=\scriptsize,color={rgb, 255:red, 155; green, 155; blue, 155 }  ,opacity=1 ]  {$\frac{\pi }{2} +\delta $};
\draw (327,106.55) node [anchor=north west][inner sep=0.75pt]  [color={rgb, 255:red, 245; green, 166; blue, 35 }  ,opacity=1 ]  {$\cdot $};
\draw (327,99.55) node [anchor=north west][inner sep=0.75pt]  [color={rgb, 255:red, 245; green, 166; blue, 35 }  ,opacity=1 ]  {$\cdot $};
\draw (323.04,116.69) node [anchor=north west][inner sep=0.75pt]    {$s_{0}$};
\draw (316.2,104.28) node [anchor=north west][inner sep=0.75pt]  [font=\scriptsize,color={rgb, 255:red, 245; green, 166; blue, 35 }  ,opacity=1 ]  {$a_{0}$};
\draw (332.4,97.27) node [anchor=north west][inner sep=0.75pt]  [font=\scriptsize,color={rgb, 255:red, 245; green, 166; blue, 35 }  ,opacity=1 ]  {$a$};
\draw (228,60.51) node [anchor=north west][inner sep=0.75pt]  [font=\scriptsize]  {$\pi -\epsilon $};

\end{tikzpicture}
\caption{The holomorphic region $U_\delta$ (gray shaded) and the singularity distribution of $\hat{\Theta}_{D_X}(t-s_0)$. The yellow dots $a$ and $a_0$ denote the intersection points of the line $\Re(t)=s_0$ with the rays $\{\arg t=\pi-\epsilon\}$ and $\{\arg t=\frac{\pi}{2}+\delta\}$, respectively.}
\label{fig: region U in deform det computation}
\end{figure}
As denoted in Figure \ref{fig: region U in deform det computation}, we find $a=s_0-is_0tan\epsilon$ and $a_0=s_0-is_0\cot\delta$. For $\arg(s_0+i\tau_1)<\pi-\epsilon<\frac{\pi}{2}+\delta$, we have $[a, e^{i(\pi-\epsilon)}\infty)\subseteq \{a_0+re^{i\theta}:r>0,\frac{\pi}{2}<\theta<\frac{\pi}{2}+\delta\}$. Thus, from $\hat{\Theta}_{D_X}(t-is_0\cot\delta)=\hat{\Theta}_{D_X}(t-s_0+a_0)\in \mathcal{N}\left(\left(\frac{\pi}{2}, \frac{\pi}{2}+\delta\right), 0\right)$, we have $\hat{\Theta}_{D_X}(t-s_0)$ and $f_{s_0}(t)$ belong to $\mathcal{N}(\pi-\epsilon,0)$. Then, we can take the asymmetry part $f_{s_0}(t)=\hat{\Theta}_{D_X}(e^{-i\pi}t-s_0)+\hat{\Theta}_{D_X}(t-s_0)$ into the first Laplace transform, and  $\widetilde{\frac{d}{d\rho}}\widetilde{\log} \HDet^{\sharp_2}_{-\Delta_X+C}(\rho^2)$ is equal to 
$$
i\left(-\int_0^{e^{(\pi-\epsilon)i}\infty}f_{s_0}(t)e^{i\rho t}dt+\int_0^{e^{(\pi-\epsilon)i}\infty}\hat{\Theta}_{D_X}(t-s_0)e^{i\rho t}dt-\int_0^{e^{(\frac{\pi}{2}-\epsilon)i}\infty}\hat{\Theta}_{D_X}(t-s_0)e^{i\rho t}dt\right).
$$
Under the condition $\hat{\Theta}_{D_X}(t-s_0)\in\hat{\mathcal{R}}_{\Omega_{s_0}^*}^{\text{simp}}$ and the Proposition \ref{prop:Theta series in N}, $\hat{\Theta}_{D_X}(t-s_0)\in\hat{\mathcal{R}}_{\Omega_{s_0}}^{\text{int}}\cap\mathcal{N}\left(\left(-\frac{\pi}{2}, \frac{\pi}{2}\right), 0\right)$, and we get 
$\hat{\Theta}_{D_X}(t-s_0+a)=\hat{\Theta}_{D_X}(t-is_0\tan\epsilon)\in\mathcal{N}\left(\left(-\frac{\pi}{2}, \frac{\pi}{2}\right), 0\right)$.
On the other hand, $\hat{\Theta}_{D_X}(t-s_0+a)\in \mathcal{N}\left(\left(\frac{\pi}{2}, \frac{\pi}{2}+\delta\right), 0\right)$, which is due to the fact that $\hat{\Theta}_{D_X}(t-s_0+a_0)\in \mathcal{N}\left(\left(\frac{\pi}{2}, \frac{\pi}{2}+\delta\right), 0\right)$.  Hence, we can 
apply the Lemma \ref{lemma: Laplace and stokes} to the last two terms, and obtain that
$$
\begin{aligned}
&\widetilde{\frac{d}{d\rho}}\widetilde{\log} \HDet^{\sharp_2}_{-\Delta_X+C}(\rho^2)\\
=&i\left(\sum_{\tau_k>0}e^{i(s_0+\tau_ki)\rho} \LL^{\left(\frac{\pi}{2}-\epsilon\right)}\left(\Delta_{s_0+i\tau_k}^-\left( \hat{\Theta}_{D_X}(t-s_0)\right)\right)(-i\rho)\right)-i\LL^{\left(\pi-\epsilon\right)}\left(f_{s_0}(t)\right)(-i\rho).
\end{aligned}
$$
\end{proofof}

We add various comments on the flexibility of the theorem's hypotheses.
\begin{remark}\label{remark on deform thm}
 As Remark \ref{remark on deriv them} for Theorem \ref{theorem about derivative regularization in introduction}, several conditions in Theorem \ref{thm about the computation of deform in intro} can also be relaxed.
\begin{itemize}
\item 
As shown in the proof and the condition $(ii)$ of Lemma \ref{lemma: Laplace and stokes}, all arguments are restricted to $U_\delta$. Consequently, we can allow $\hat{\Theta}_{D_X}(t-s_0)$ to possess additional discrete singular points in $\mathbb{C} \setminus\left\{s_0+re^{i\theta}:r>0,-\frac{\pi}{2}-\delta<\theta<\frac{\pi}{2}+\delta\right\}$. Such additional singularities neither affect the derivation nor contribute to the sum in \eqref{eq:compute det1 formula}, provided that within $U_\delta$, $\hat{\Theta}_{D_X}(t-s_0)$ remains holomorphic and has simple singularities at each $s_0+i\tau_k\in\Omega_{s_0}^+$.
\item
Condition $(ii)$ in Theorem \ref{thm about the computation of deform in intro} can be extended from $\hat{\Theta}_{D_X}(t-is_0\cot\delta)\in \mathcal{N}\left(\left(\frac{\pi}{2}, \frac{\pi}{2}+\delta\right), 0\right)$ to $\hat{\Theta}_{D_X}(t-is_0\cot\delta)\in \mathcal{N}\left(\left(\frac{\pi}{2}, \frac{\pi}{2}+\delta\right), \gamma\right)$ for a locally bounded function $\gamma: \left(\frac{\pi}{2}, \frac{\pi}{2}+\delta\right) \to \mathbb{R}$, provided that the domain of validity of formula \eqref{eq:compute det1 formula} is adjusted accordingly.
\end{itemize}
\end{remark}

\section{Example: Unit Circle and Poisson Summation Formula
{\label{section:example S1}}}
\subsection{Formal Logarithmic Derivative Regularization 
: The Poisson Summation Formula}
In this subsection, we apply Theorem \ref{theorem about derivative regularization in introduction} to compute the formal logarithmic derivative regularization of determinant in $S^1$, and it implies the Poisson summation formula on $S^1$.

Let $X=S^1$ be the unit circle, the non-zero spectrum of $-\Delta_{S^1}$ is $\{n^2\}_{n\in\ZZ\setminus\{0\}}$. Let $\rho_n=\sqrt{n^2}=n$, $n\geq1$, $m_0=\dim X+1=2$, then
$$\hat{\Theta}_{D_{S^1}}(t)=2\sum_{n\geq1}e^{-nt}=\frac{2}{e^t-1},\quad\Re(t)>0.$$
It can be holomorphically extended to $\CC\setminus2\pi i\ZZ$, with $t\hat{\Theta}_{D_{S^1}}(t)\in\Omega_{2\pi i\ZZ^*}^{simp}\cap\mathcal{N}\left(\left(-\frac{\pi}{2},\frac{\pi}{2}\right),0\right)\cap\mathcal{N}\left(\left(\frac{\pi}{2},\pi\right),0\right)$, and the sum satisfies
$$\hat{\Theta}_{D_{S^1}}(-t)+\hat{\Theta}_{D_{S^1}}(t)=-2.$$
In this case, for $2\pi im\in2\pi i\ZZ_{>0}$,
$$\Delta_{2\pi im}^-(t\hat{\Theta}_{D_{S^1}}(t))=-\mathrm{sing}_0\left(\left(\mathrm{cont}_{\gamma_{\omega_m}^-}t\hat{\Theta}_{D_{S^1}}(t)\right)(2\pi im+t)\right)=-2\pi i(4\pi im)\boldsymbol{\delta}_{2\pi im}.$$
Hence, for $0<\epsilon<\frac{\pi}{2}$ and $-\frac{\pi}{2}+\epsilon<\arg(\rho)<\frac{\pi}{2}-\epsilon$, we have
$$
\begin{aligned}
\widetilde{\frac{d^{2}}{d\rho^{2}}}\widetilde{\log}\HDet_{-\Delta_{S^1}}(\rho^2)= &-i^2\left(\sum_{m\geq1}e^{-2\pi m\rho}\LL^{\left(\frac{\pi}{2}-\epsilon\right)}\left(2\pi i(4\pi im)\boldsymbol{\delta}_{2\pi im}\right)(-i\rho)\right)\\
&-i^2\LL^{\left(\frac{\pi}{2}+\epsilon\right)}\left(-2t\right)(-i\rho)\\
=&-2\left(\sum_{m\geq1}4\pi^2me^{-2\pi m\rho}-\frac{1}{\rho^2}\right)\\
=&\frac{d}{d\rho}2\left(\frac{\pi}{\tanh\pi\rho}-\frac{1}{\rho}-\pi\right).
\end{aligned}
$$
On the other hand, by definition,
$$\widetilde{\frac{d^{2}}{d\rho^{2}}}\widetilde{\log}\HDet_{-\Delta_{S^1}}(\rho^2)=\frac{d}{d\rho}2\sum_{n\geq1}\frac{2\rho}{n^2+\rho^2}.$$
Integrating these two formulas, we find that there exists $C\in\CC$ such that
$$\sum_{n\in\ZZ}\frac{\rho}{n^2+\rho^2}=\frac{\pi}{\tanh\pi\rho}+C, \text{ for }-\frac{\pi}{2}+\epsilon<\arg(\rho)<\frac{\pi}{2}-\epsilon.$$ 
Let $\rho$ tend to $0$ in $\RR^+$, then $\sum_{n\in\ZZ}\frac{\rho}{n^2+\rho^2}=\frac{1}{\rho}+O(\rho)$ because
$$0\leq\frac{1}{2\rho}\left(\sum_{n\in\ZZ}\frac{\rho}{n^2+\rho^2}-\frac{1}{\rho}\right)=\sum_{n\geq1}\frac{1}{n^2+\rho^2}\leq\sum_{n\geq1}\frac{1}{n^2}=\frac{\pi^2}{6}.$$
However, when the real number $\rho\to0^+$,
$$\frac{\pi}{\tanh\pi\rho}+C=C+\frac{\pi}{\pi\rho-\frac{(\pi\rho)^3}{3}+O(\rho^5)}=C+\frac{1}{\rho}\frac{1}{1-\frac{(\pi\rho)^2}{3}+O(\rho^4)}=C+\frac{1}{\rho}+\frac{\pi^2\rho}{3}+O(\rho^3).$$
Comparing these two equalities, we find $C=0$. Performing analytic continuation，we have
\begin{equation}{\label{distribution PSF}}
\sum_{n\in\ZZ}\frac{\rho}{n^2+\rho^2}=\frac{\pi}{\tanh\pi\rho},\quad \rho\in\CC\setminus i\ZZ.
\end{equation}
This is actually the Poisson summation formula.
\begin{theorem}[Poisson Summation Formula, {\cite[Theorem 2.4]{Steinbook}}, {\cite[\S1.4]{cartierVoros1988}}]{\label{theorem:PSF}}
Assume that in the strip region $|\Im(\rho)|<2\eta$ for some $\eta>0$, function $h:\CC\to \CC$ is holomorphic and $|h(\rho)|\in O(|\rho|^{-1-\delta})$ for some $\delta>0$ when $|\rho|\to\infty$. Then,
$$
\sum_{n\in\ZZ}h(n)=2\pi\sum_{m\in\ZZ}\hat{h}(2\pi m),
$$
where $\hat{h}(\tau):=\frac{1}{2\pi}\int_{-\infty}^{\infty}h(\rho)e^{i\tau\rho}d\rho$ is the Fourier transform of $h$.
\end{theorem}
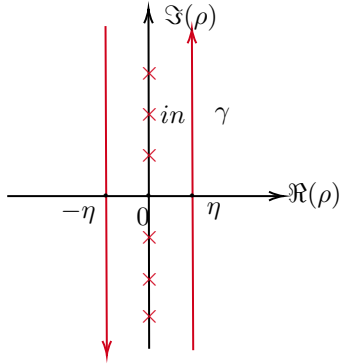
\begin{figure}[htbp]
    \centering 
\tikzset{every picture/.style={line width=0.75pt}} 

\begin{tikzpicture}[x=0.5pt,y=0.5pt,yscale=-1,xscale=1]

\draw [color={rgb, 255:red, 208; green, 2; blue, 27 }  ,draw opacity=1 ]   (364,266) -- (363.01,26) ;
\draw [shift={(363,24)}, rotate = 89.76] [color={rgb, 255:red, 208; green, 2; blue, 27 }  ,draw opacity=1 ][line width=0.75]    (10.93,-3.29) .. controls (6.95,-1.4) and (3.31,-0.3) .. (0,0) .. controls (3.31,0.3) and (6.95,1.4) .. (10.93,3.29)   ;
\draw [color={rgb, 255:red, 208; green, 2; blue, 27 }  ,draw opacity=1 ]   (298,22) -- (298.99,268) ;
\draw [shift={(299,270)}, rotate = 269.77] [color={rgb, 255:red, 208; green, 2; blue, 27 }  ,draw opacity=1 ][line width=0.75]    (10.93,-3.29) .. controls (6.95,-1.4) and (3.31,-0.3) .. (0,0) .. controls (3.31,0.3) and (6.95,1.4) .. (10.93,3.29)   ;
\draw    (224,150) -- (266,150) -- (430,150) ;
\draw [shift={(432,150)}, rotate = 180] [color={rgb, 255:red, 0; green, 0; blue, 0 }  ][line width=0.75]    (10.93,-3.29) .. controls (6.95,-1.4) and (3.31,-0.3) .. (0,0) .. controls (3.31,0.3) and (6.95,1.4) .. (10.93,3.29)   ;
\draw    (330.83,264.92) -- (330.34,10.92) ;
\draw [shift={(330.33,8.92)}, rotate = 89.89] [color={rgb, 255:red, 0; green, 0; blue, 0 }  ][line width=0.75]    (10.93,-3.29) .. controls (6.95,-1.4) and (3.31,-0.3) .. (0,0) .. controls (3.31,0.3) and (6.95,1.4) .. (10.93,3.29)   ;

\draw (372,152.4) node [anchor=north west][inner sep=0.75pt]    {$\eta $};
\draw (378,81.4) node [anchor=north west][inner sep=0.75pt]    {$\gamma $};
\draw (319,156.4) node [anchor=north west][inner sep=0.75pt]    {$0$};
\draw (324,141.4) node [anchor=north west][inner sep=0.75pt]  [font=\Large]  {$\cdot $};
\draw (321,110.4) node [anchor=north west][inner sep=0.75pt]    {$\textcolor[rgb]{0.82,0.01,0.11}{\times }$};
\draw (321,79.4) node [anchor=north west][inner sep=0.75pt]    {$\textcolor[rgb]{0.82,0.01,0.11}{\times }$};
\draw (337,79.4) node [anchor=north west][inner sep=0.75pt]    {$in$};
\draw (321,48.4) node [anchor=north west][inner sep=0.75pt]    {$\textcolor[rgb]{0.82,0.01,0.11}{\times }$};
\draw (321,172.4) node [anchor=north west][inner sep=0.75pt]    {$\textcolor[rgb]{0.82,0.01,0.11}{\times }$};
\draw (321,203.4) node [anchor=north west][inner sep=0.75pt]    {$\textcolor[rgb]{0.82,0.01,0.11}{\times }$};
\draw (433,137.4) node [anchor=north west][inner sep=0.75pt]    {$\Re ( \rho )$};
\draw (340,3.4) node [anchor=north west][inner sep=0.75pt]    {$\Im ( \rho )$};
\draw (292,141.4) node [anchor=north west][inner sep=0.75pt]  [font=\Large]  {$\cdot $};
\draw (357,141.4) node [anchor=north west][inner sep=0.75pt]  [font=\Large]  {$\cdot $};
\draw (262,152.4) node [anchor=north west][inner sep=0.75pt]    {$-\eta $};
\draw (321,231.4) node [anchor=north west][inner sep=0.75pt]    {$\textcolor[rgb]{0.82,0.01,0.11}{\times }$};
\end{tikzpicture}
\caption{The integral contour $\gamma$.}
\label{fig:the integral contour gamma for psf}
\end{figure}
The approach used in \cite[Theorem 2.4]{Steinbook} and \cite[\S1.4]{cartierVoros1988} to derive the Poisson summation formula from \eqref{distribution PSF} is summarized as follows: 
By multiplying both sides of \eqref{distribution PSF} by a test function $h(-i\rho)$ satisfying Theorem \ref{theorem:PSF} and integrating over the contour $\gamma$ shown in Figure \ref{fig:the integral contour gamma for psf}, one has
\begin{equation}{\label{distribution PSF,integral form}}
\frac{1}{2\pi i}\int_{\gamma}\sum_{n\in\ZZ}\frac{\rho h(-i \rho)}{\rho^2+n^2} d\rho=\frac{1}{2\pi i} \int_{\gamma}\frac{\pi h(-i\rho)}{\tanh\pi\rho}d\rho.
\end{equation}
Subsequently, one can apply the Residue theorem to the left-hand side of \eqref{distribution PSF,integral form} and expand $\frac{1}{\tanh(\pi\rho)}$ in terms of $e^{2\pi m\rho}$ on the right-hand side. Then the Poisson summation formula follows. 

We will use this approach to derive the deformed Poisson summation formula from the exponential deformed regularization of determinant on $S^1$ in the next subsection.


\subsection{Exponentially Deformed Regularization 
: Deformed Poisson Summation Formula}
We apply Theorem \ref{thm about the computation of deform in intro} to compute the exp deformed determinant in this subsection, which leads to a deformed Poisson summation formula. Recall that
$$\hat{\Theta}_{D_{S^1}}(t)=2\sum_{n\geq1}e^{-nt}=\frac{2}{e^t-1},\quad\Re(t)>0,$$
so for any fixed $s_0<0$,
$$\hat{\Theta}_{D_{S^1}}(t-s_0)=\frac{2}{e^{t-s_0}-1},\quad\Re(t)>s_0.$$
It implies $\hat{\Theta}_{D_X}(t-s_0)\in \hat{\mathcal{R}}_{s_0+2\pi i\ZZ}^{\text{simp}}$.  Obviously, for any $\arg(s_0+2\pi i)<\frac{\pi}{2}+\delta<\pi$, $\hat{\Theta}_{D_X}(t-is_0\cot\delta)\in\mathcal{N}\left(\left(\frac{\pi}{2},\frac{\pi}{2}+\delta\right),0\right)$, and
$$\hat{\Theta}_{D_{S^1}}(-t-s_0)+\hat{\Theta}_{D_{S^1}}(t-s_0)=-2+\frac{2\sinh s_0}{\cosh t-\cosh s_0}=:f_{s_0}(t).$$
So for any $\arg(s_0+2\pi i)<\pi-\epsilon<\pi$ and $\arg(\rho)\in\left(-\frac{\pi}{2}+\epsilon,\epsilon\right)$,
$$
\begin{aligned}
\widetilde{\frac{d}{d\rho}}\widetilde{\log} \HDet^{\sharp_2}_{-\Delta_{S^1}}(\rho^2)=&i\left(\sum_{m\geq1}e^{i(s_0+2\pi mi)\rho}\LL^{\left(\frac{\pi}{2}-\epsilon\right)}\left(\Delta_{s_0+i2\pi m}^-\left(\hat{\Theta}_{D_X}(t-s_0)\right)\right)(-i\rho)\right)\\
&-i\LL^{\left(\pi-\epsilon\right)}\left(f_{s_0}(t)\right)(-i\rho).\\
=&i\left(e^{is_o\rho}(-2\pi i)2\sum_{m\geq1}e^{-2\pi m\rho}\right)-i\int_0^{e^{i\left(\pi-\epsilon\right)\infty}}\left(-2+\frac{2\sinh s_0}{\cosh t-\cosh s_0}\right)e^{i\rho t}dt\\
=&2\pi e^{is_0\rho}\left(\frac{1}{\tanh\pi\rho}-1\right)-\frac{2}{\rho}-i\int_0^{e^{i\left(\pi-\epsilon\right)\infty}}\frac{2\sinh s_0}{\cosh t-\cosh s_0}e^{i\rho t}dt.
\end{aligned}
$$
On the other hand, by definition,
$$\widetilde{\frac{d}{d\rho}}\widetilde{\log} \HDet^{\sharp_2}_{-\Delta_{S^1}}(\rho^2)=\sum_{n\in\ZZ\setminus\{0\}}\frac{2\rho e^{s_0|n|}}{n^2+\rho^2}.$$
Hence, for $\arg(\rho)\in\left(-\frac{\pi}{2}+\epsilon,\epsilon\right)$,
\begin{equation}{\label{deformed distributon PSF}}
\sum_{n=1}^{\infty}\frac{2\rho e^{s_0n}}{n^2+\rho^2}=e^{is_0\rho}\left(\frac{\pi}{\tanh\pi\rho}\right)-\frac{1}{\rho}-i\LL^{(\pi-\epsilon)}\left(\frac{\sinh s_0}{\cosh t-\cosh s_0}\right)(-i\rho)-\pi e^{is_0\rho}.
\end{equation}

The identity \eqref{deformed distributon PSF} is similar to \eqref{distribution PSF}, but we can not compare them directly by letting $s_0\to0$. This is because the term about Laplace transform in \eqref{deformed distributon PSF} is divergent when $s_0\to0$, since $\cosh t-\cosh s_0\sim(t-s_0)\sinh s_0$ as $t\to s_0$. However, each term in the derivative of \eqref{deformed distributon PSF} is convergent if $s_0\to0$, and the derivative of \eqref{deformed distributon PSF} tends to the derivative of \eqref{distribution PSF} as $s_0\to0$. We will prove a general result in \S\ref{section:relation of two reg}, which reveals the relation between two regularization methods for determinant from the viewpoint of singularity.

Similar to Theorem \ref{theorem:PSF}, we obtain a deformed Poisson summation formula based on the identity \eqref{deformed distributon PSF}.

\begin{theorem}[Deformed Poisson Summation Formula]{\label{thm: deformed psf}}
Let $s_0$ be a negative real number, and $h(\rho):\CC\to\CC$ be a function satisfying
\begin{itemize}
\item There exists $\eta<|s_0|$, such that $h(\rho)$ is holomorphic in the strip region $|\Im (\rho)|<2\eta$;

\item There exists $\delta>0$, such that
$|h(\rho)|\in O(|\rho|^{-1-\delta})$ if $\rho\to-\infty+is$, for $s\in(-2\eta,2\eta)$;

\item There exists $\delta^{\prime}<s_0$, such that $|h(\rho)|\in O(e^{\delta^\prime|\rho|})$ if $\rho\to+\infty+is$, for $s\in(-2\eta,2\eta)$.
\end{itemize}
Then,
$$\sum_{n\in\ZZ}h(n)e^{-s_0n}=2\pi\sum_{m\in\ZZ}\hat{h}(2\pi m+is_0).$$
\end{theorem}

We sketch the idea of proof here, and leave the complete argument to Appendix A. First, we extend \eqref{deformed distributon PSF} meromorphically to the complex plane. In particular, the meromorphic extension of the Laplace transform term in \eqref{deformed distributon PSF}, denoted by $K_{s_0}$, is non-trivial since it depends on the singularity of the integrand $\frac{\sinh s_0}{\cosh t - \cosh s_0}$ (see Lemma \ref{lemma: hol-ext of K}). Next, following the approach in (\cite[Theorem 2.4]{Steinbook} and \cite[\S1.4]{cartierVoros1988}), we multiply \eqref{deformed distributon PSF} by a test function $h(-i\rho)$ and take the contour integration. We obtain 
$$\int_\gamma\sum_{n\in\mathbb{Z}}^{\infty}\frac{\rho e^{s_0 |n|}}{n^2+\rho^2}h(-i\rho)d\rho=\int_\gamma e^{is_0\rho}\left(\frac{\pi}{\tanh\pi\rho}\right)h(-i\rho)d\rho-\int_\gamma iK_{s_0}(\rho)h(-i\rho)d\rho-\int_\gamma \pi e^{is_0\rho}h(-i\rho)d\rho.$$
Applying the residue theorem to the left-hand side and the term about $K_{s_0}$, we get $\sum_{n \in \mathbb{Z}} h(n)e^{s_0|n|}$ and $\sum_{n=1}^{\infty}(e^{-s_0n}-e^{s_0n})h(n)$, respectively. Meanwhile, by expanding $\frac{1}{\tanh \pi\rho}$ into the sum of Fourier basis in the first term on the right-hand side, we arrive at the deformed Poisson summation formula
$$\sum_{n\in\ZZ}h(n)e^{-s_0n}=2\pi\sum_{m\in\ZZ}\hat{h}(2\pi m+is_0).$$

\begin{remark}
\begin{itemize}
\item[(i)]
Although the deformed Poisson summation formula can be derived via standard Fourier properties and Poisson summation formula, the preceding argument reveals a deeper insight: the left-hand side arises from the contour integration of a deformed regularized determinant, whereas the right-hand side is induced by the sum of alien operators acting on a shifted theta series on $S^1$. These operators precisely encode the singularity distribution of $\hat{\Theta}_{D_{S^1}}(t-s_0)$. In other words, the right-hand side originates from the Stokes structure related to the singularities of $\hat{\Theta}_{D_{S^1}}(t-s_0)$. 

The same analysis can be applied to explain the Selberg trace formula (see \S\ref{section:example X}) and others derived from our regularization.

\item[(ii)] Identities \eqref{distribution PSF} and \eqref{deformed distributon PSF} derived from regularized formulas \eqref{formula:reg1 in intro} and \eqref{eq:compute det1 formula} are more refined than the Poisson summation formula and its deformed version. 

Same remarks apply to the Selberg trace formula in the next section. 

\end{itemize}
\end{remark}

\section{Example: 
Higher Genus Riemann Surface and Selberg Trace Formula}{\label{section:example X}}
Let $X$ be a closed Riemann surface with genus $g\geq 2$. Suppose $\{\lambda_n\}$ is the spectrum of $-\Delta_X-\frac{1}{4}$, and set $\rho_n=\sqrt{\lambda_n}$, then $\hat{\Theta}_{D_X}(t)=\sum_{n\geq1}e^{-\rho_nt}$, for $\Re(t)>0$. 

With the help of Selberg's Trace Formula, Cartier and Voros (\cite[Theorem 2]{cartierVoros1988}) verified that:

(1) If $\mathcal{P}=\{\tau_m,m\in\ZZ_{>0}\}$ is the set of lengths of all primitive periodic geodesics in $X$, then $\hat{\Theta}_{D_X}(t)$ can be holomorphically extended to $\CC\setminus(\{k\tau_mi:k\in\ZZ\setminus\{0\},\tau_m\in\mathcal{P}\}\cup\{-2\pi m,m\in\ZZ_{\geq0}\})$, and in this region,
\begin{equation}{\label{formula:extension of theta in STF}}
\hat{\Theta}_{D_{X}}(-t)+\hat{\Theta}_{D_{X}}(t)=(2-2g)\hat{\Theta}_{D_{S^2}}(it),
\end{equation}
where $\hat{\Theta}_{D_{S^2}}$ is the Theta series defined by the square roots of the spectrum of $-\Delta_{S^2}+\frac{1}{4}$ on the Riemann sphere, i.e. 
$$\hat{\Theta}_{D_{S^2}}(t)=\sum_{j=0}^{\infty}(2j+1)e^{-t(j+\frac{1}{2})}=\frac{\cosh t/2}{2\sinh^2t/2},\quad\Re(t)>0.
$$

(2) All singular points in $\{k\tau_mi:k\in\ZZ\setminus\{0\},\tau_m\in\mathcal{P}\}$ are simple poles, and the residue of $\hat{\Theta}_{D_X}(t)$ at each $k\tau_mi$ is $\frac{\tau_m}{4\pi\sinh(k\frac{\tau_m}{2})}$, so the alien derivative $\Delta_{k\tau_mi}^-(t^2\hat{\Theta}_{D_X}(t))=\frac{ik^2\tau_m^3}{2\sinh(\frac{k\tau_m}{2})}\boldsymbol{\delta}_{ik\tau_m}$. Other singular points are $\{-2\pi m:m\in\ZZ_{\geq0}\}$, where $\hat{\Theta}_{D_X}(t)$ has second-order poles.

Let $m_0= \dim X+1=3$. From (\ref{formula:extension of theta in STF}), one can verify that the analytic continuation $t^2\hat{\Theta}_{D_X}(t)$ belongs to $\mathcal{N}\left((\frac{\pi}{2}, \pi),0\right)$. Additionally, this extension has a simple singularity at each $ik\tau_m$, where $k \in \mathbb{Z}_{>0}$ and $\tau_m \in \mathcal{P}$. According to Theorem \ref{theorem about derivative regularization in introduction} and Remark \ref{remark on deriv them}, for $\epsilon\in\left(0,\frac{\pi}{2}\right)$, $\arg(\rho)\in\left(-\frac{\pi}{2}+\epsilon,\frac{\pi}{2}-\epsilon\right)$,
\begin{equation}\label{computation of det in X and Selberg zeta}
\begin{aligned}
\widetilde{\frac{d^3}{d\rho^3}}\widetilde{\log}\HDet_{-\Delta_X+\frac{1}{4}}(\rho^2)
=&i^3\left(\sum_{k\tau_m>0}e^{-k\tau_m\rho}\LL^{\left(\frac{\pi}{2}-\epsilon\right)}\left(\frac{ik^2\tau_m^3\boldsymbol{\delta}_{ik\tau_m}}{2\sinh\left(\frac{k\tau_m}{2}\right)}\right)(-i\rho)\right)\\
&-(2-2g)i^3\LL^{\left(\frac{\pi}{2}+\epsilon\right)}\left(t^2\hat{\Theta}_{D_{S^2}}(it)\right)(-i\rho)\\
=&\sum_{k=1}^{\infty}\sum_{\tau_m\in\mathcal{P}}e^{-k\tau_m\rho}\frac{k^2\tau_m^3}{2\sinh\left(\frac{k\tau_m}{2}\right)}+(2-2g)\LL^{\epsilon}\left(t^2\hat{\Theta}_{D_{S^2}}(t)\right)(\rho).
\end{aligned}
\end{equation}
The first term is the third derivative of $\log \mathcal{Z}_X(\frac{1}{2}+\rho)$ for Selberg zeta function
\footnote{The relation between Selberg zeta function and the first term of formula \eqref{computation of det in X and Selberg zeta} is obtained as follows: we expand each term $\log(1-e^{-\tau_m(\rho+\frac{1}{2}+n)})$ in $\log \mathcal{Z}_X(\frac{1}{2}+\rho)$ as the series $\sum_{k=1}^{\infty} -\frac{1}{k} e^{-k\tau_m(\rho+\frac{1}{2}+n)}$. By interchanging the summations over $k$ and $n$ and taking the third derivative in the holomorphic region of $\mathcal{Z}_X(\frac{1}{2}+\rho)$, the first term of formula \eqref{computation of det in X and Selberg zeta} emerges.
$$
\left(\frac{d}{d \rho}\right)^3\left(\log \mathcal{Z}_X\left(\frac{1}{2}+\rho\right)\right)
 =\sum_{k=1}^{\infty} \sum_{\tau_m\in\mathcal{P}} k^2 \tau_m^3 \frac{e^{-k \tau_m \rho}}{e^{\frac{k \tau_m}{2}}-e^{-\frac{k \tau_m}{2}}}
=\sum_{k=1}^{\infty} \sum_{\tau_m\in\mathcal{P}} e^{-k \tau_m \rho} \frac{k^2 \tau_m^3}{2 \sinh \left(\frac{k \tau_m}{2}\right)},\text{ for }\Re(\rho)>\frac{1}{2}.
$$
}. Recall that the Selberg zeta function $\mathcal{Z}_X$ (\cite[(2.22)]{cartierVoros1988}) is defined by the Euler product over the lengths of primitive periodic geodesics:
$$\mathcal{Z}_X(s):=\prod_{\tau_m\in\mathcal{P}}\prod_{n=0}^{\infty}\left(1-e^{-\tau_m(s+n)}\right).$$
It is convergent and holomorphic in $\{\Re (s)>1:s\in\CC\}$.

For the second term of \eqref{computation of det in X and Selberg zeta}, note that $t^2\hat{\Theta}_{D_{S^2}} \in \mathcal{N}\left((-\frac{\pi}{2}, \frac{\pi}{2}), 0\right)$. Similar to \eqref{laplace and determinant}, we obtain $\LL^{\epsilon}\left(t^2\hat{\Theta}_{D_{S^2}}(t)\right)(\rho)=\widetilde{\frac{d^3}{d \rho^3}} \widetilde{\log} \HDet_{D_{S^2}}\left(\rho\right)$.
Thus,
 \begin{equation}{\label{3-derivative of reg det in Rie-Surface.}}
\begin{aligned}
\widetilde{\frac{d^{3}}{d\rho^{3}}}\widetilde{\log} \HDet_{-\Delta_{X}+\frac{1}{4}}(\rho^2)
=&\frac{d^3}{d\rho^3}\left(\log \mathcal{Z}_X(\frac{1}{2}+\rho)\right)-(2g-2)\widetilde{\frac{d^3}{d\rho^3}}\widetilde{\log} \HDet_{D_{S^2}}(\rho),\text{ when } \Re (\rho)>\frac{1}{2}.
\end{aligned}
\end{equation}
This formula is the first derivative of Eq. (3.55) in Cartier and Voros (\cite{cartierVoros1988}). By multiplying both sides of \eqref{3-derivative of reg det in Rie-Surface.} by a test function and integrating over a suitable contour, they rederived Selberg trace formula:
\begin{corollary}[Selberg trace formula,{\cite[Theorem 1]{cartierVoros1988}}]\label{Selberg trace formula}
Suppose that
\begin{itemize}
\item $h(\rho)$ is a holomorphic function from $V$ to $\mathbb{C}$, where $V$ is an open region such that $\{\rho\in\mathbb{C}:\Re(\rho)\geq 0\}\subset V\subset\mathbb{C}$.
\item There exist $c>0$ and $\alpha>2$, such that when $|\Im(\rho)|<c$, $|\rho|\to \infty$, $h(\rho)\in O(|\rho|^{-\alpha})$.
\item The differential form $h(\rho)d\log \mathcal{Z}_X(\frac{1}{2}+i\rho)$ is integrable when $-\frac{\pi}{2}\leq\arg\rho\leq0$; and $h(\rho)d\log \mathcal{Z}_X(\frac{1}{2}-i\rho)$ is integrable when $0\leq\arg\rho\leq\frac{\pi}{2}$. Then
\end{itemize}
$$\sum_{n=0}^{\infty} h\left(\rho_n\right)=(2 g-2) \int_0^{\infty} h(\rho) \rho \tanh \pi \rho d \rho+\int_0^{\infty} \frac{h(i\rho)+h(-i\rho)}{2\pi i} d \log \mathcal{Z}_X\left(\frac{1}{2}+\rho\right).$$
\end{corollary}
\begin{remark}
Cartier and Voros obtained (\ref{3-derivative of reg det in Rie-Surface.}) by plugging the special test function $h_\rho(z)=\left(\rho^2+z^2\right)^{-2}$ into the Selberg trace formula. Different from them, we get the functional equation of Selberg zeta function and determinant (\ref{3-derivative of reg det in Rie-Surface.}) by computing regularized determinant (Theorem \ref{theorem about derivative regularization in introduction}). Here, the term about Selberg zeta function can be interpreted as the Stokes structure of $\hat{\Theta}_{D_X}$-series on the Riemann surface $X$.
\end{remark}


\section{Relation between Two Regularizations
}{\label{section:relation of two reg}}

We take $\varepsilon$ as in Theorem \ref{theorem about derivative regularization in introduction} and Theorem \ref{thm about the computation of deform in intro}, then
\begin{theorem}{\label{relation Reg1 and Reg2}}
For all $ m\geq d$ and $\rho\in\{\rho\in\mathbb{C}:-\frac{\pi}{2}+\varepsilon<\arg(\rho)<\varepsilon\}$,
\begin{equation}
 \lim _{s_0 \rightarrow 0}\left(\frac{d}{d \rho}\right)^m \frac{\widetilde{d}}{d \rho} \widetilde{\log }\HDet_{-\Delta_X+C}^{\sharp2}\left(\rho^2\right)=\frac{\widetilde{d^{m+1}}}{d \rho^{m+1}} \widetilde{\log }\HDet_{-\Delta_X+C}\left(\rho^2\right).
\end{equation}
\end{theorem}

\begin{proof}
From the definition of $\frac{\widetilde{d}}{d \rho} \widetilde{\log }\HDet_{-\Delta_X+C}^{\sharp2}\left(\rho^2\right)$ in \S\ref{section:deformed Borel-laplace}, and the fact that $\sum_{n\geq1}\frac{ 2 \rho e^{s_0 \rho_n}}{\rho^2+\rho_n^2}$ uniformly convergent in any compact subsets of $\mathbb{C}-\{-\rho_1,-\rho_2,\cdots\}$,
$$
\begin{aligned}
\left(\frac{d}{d \rho}\right)^m\frac{\widetilde{d}}{d \rho} \widetilde{\log }\HDet_{-\Delta_X+C}^{\sharp2}\left(\rho^2\right)
&=\left(\frac{d}{d \rho}\right)^m \sum_{n\geq1}\frac{ 2 \rho e^{s_0 \rho_n}}{\rho^2+\rho_n^2}\\
&=\sum_{n\geq1}(-i)\left(\frac{d}{d \rho}\right)^m \left(\frac{ e^{s_0 \rho_n}}{\rho_n-i\rho}-\frac{ e^{s_0 \rho_n}}{\rho_n+i\rho}\right)\\
&=(-i)^{m+1}\frac{\widetilde{d^{m+1}}}{d (-i\rho)^{m+1}} \widetilde{\log }\HDet_{D_X}^{\sharp1}\left(-i\rho\right)+i^{m+1}\frac{\widetilde{d^{m+1}}}{d (i\rho)^{m+1}} \widetilde{\log }\HDet_{D_X}^{\sharp1}\left(-i\rho\right),
\end{aligned}
$$
from Corollary \ref{relation:laplace relation four}$(iv)$ and Proposition \ref{prop:multiple in borel and derivative in laplace}, the right hand side is
$$
(-i)^{m+1}\mathcal{L}^{\frac{\pi}{2}-\epsilon}\left((-t)^{m}\hat{\Theta}_{D_X}(t-s_0)\right)(-i\rho)+i^{m+1}\mathcal{L}^{-\frac{\pi}{2}+\epsilon}\left((-t)^{m}\hat{\Theta}_{D_X}(t-s_0)\right)(i\rho).
$$
Let $s_0\to 0$, using the Lebesgue dominant convergence theorem and Proposition \ref{proposition:Dirac-theta property}, we have
$$
\begin{aligned}
\lim_{s_0\to0}\left(\frac{d}{d \rho}\right)^m\frac{\widetilde{d}}{d \rho} \widetilde{\log }\HDet_{-\Delta_X+C}^{\sharp2}\left(\rho^2\right)
&=(-i)^{m+1}\mathcal{L}^{\frac{\pi}{2}-\epsilon}\left((-t)^{m}\hat{\Theta}_{D_X}(t)\right)(-i\rho)+i^{m+1}\mathcal{L}^{-\frac{\pi}{2}+\epsilon}\left((-t)^{m}\hat{\Theta}_{D_X}(t)\right)(i\rho)\\
&=(-i)^{m+1}\frac{\widetilde{d^{m+1}}}{d (-i\rho)^m} \widetilde{\log }\HDet_{D_X}\left(-i\rho\right)+i^{m+1}\frac{\widetilde{d^{m+1}}}{d (i\rho)^m} \widetilde{\log }\HDet_{D_X}\left(i\rho\right)\\
&=\frac{\widetilde{d^{m+1}}}{d \rho^{m+1}} \widetilde{\log }\HDet_{-\Delta_X+C}\left(\rho^2\right).
\end{aligned}
$$
The last two equalities follow from Eqs. (\ref{laplace and determinant}) and (\ref{eqn:decomposition of m-derivative}) respectively.
\end{proof}
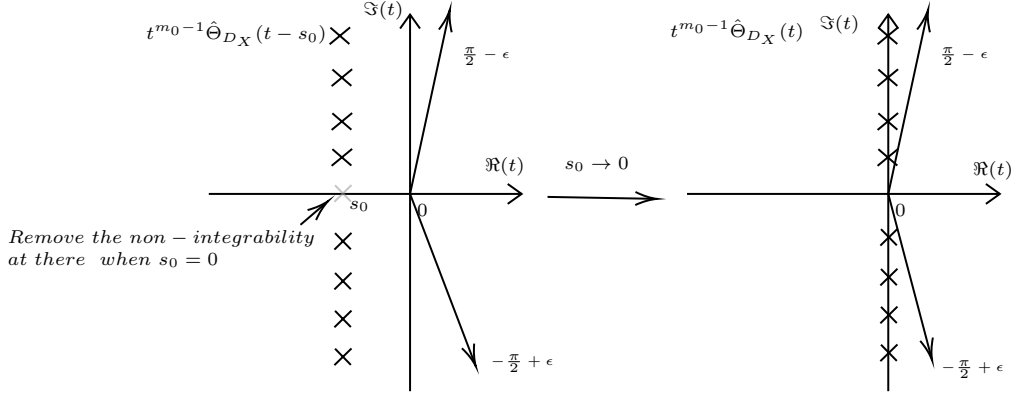
\begin{figure}[htbp]
    \centering  
\tikzset{every picture/.style={line width=0.75pt}} 

\begin{tikzpicture}[x=0.75pt,y=0.75pt,yscale=-1,xscale=1]

\draw  (114,100.5) -- (271,100.5)(215,10.5) -- (215,199.5) (264,95.5) -- (271,100.5) -- (264,105.5) (210,17.5) -- (215,10.5) -- (220,17.5)  ;
\draw  [color={rgb, 255:red, 155; green, 155; blue, 155 }  ,draw opacity=0.58 ] (177.35,96.18) -- (185.29,104.31)(185.61,96.05) -- (177.03,104.44) ;
\draw   (177.35,120.18) -- (185.29,128.31)(185.61,120.05) -- (177.03,128.44) ;
\draw   (177.35,140.18) -- (185.29,148.31)(185.61,140.05) -- (177.03,148.44) ;
\draw   (177.35,159.18) -- (185.29,167.31)(185.61,159.05) -- (177.03,167.44) ;
\draw   (177.35,178.18) -- (185.29,186.31)(185.61,178.05) -- (177.03,186.44) ;
\draw [color={rgb, 255:red, 0; green, 0; blue, 0 }  ,draw opacity=1 ]   (215,100.5) -- (234.58,9.95) ;
\draw [shift={(235,8)}, rotate = 102.2] [color={rgb, 255:red, 0; green, 0; blue, 0 }  ,draw opacity=1 ][line width=0.75]    (10.93,-3.29) .. controls (6.95,-1.4) and (3.31,-0.3) .. (0,0) .. controls (3.31,0.3) and (6.95,1.4) .. (10.93,3.29)   ;
\draw    (215,100.5) -- (247.29,185.13) ;
\draw [shift={(248,187)}, rotate = 249.12] [color={rgb, 255:red, 0; green, 0; blue, 0 }  ][line width=0.75]    (10.93,-3.29) .. controls (6.95,-1.4) and (3.31,-0.3) .. (0,0) .. controls (3.31,0.3) and (6.95,1.4) .. (10.93,3.29)   ;
\draw   (175.02,17.04) -- (184.55,25.27)(184.93,16.91) -- (174.64,25.4) ;
\draw   (176.02,78.04) -- (185.55,86.27)(185.93,77.91) -- (175.64,86.4) ;
\draw   (176.02,60.04) -- (185.55,68.27)(185.93,59.91) -- (175.64,68.4) ;
\draw   (176.02,38.04) -- (185.55,46.27)(185.93,37.91) -- (175.64,46.4) ;
\draw  (354,100.5) -- (511,100.5)(455,10.5) -- (455,199.5) (504,95.5) -- (511,100.5) -- (504,105.5) (450,17.5) -- (455,10.5) -- (460,17.5)  ;
\draw   (451.35,118.18) -- (459.29,126.31)(459.61,118.05) -- (451.03,126.44) ;
\draw   (451.35,138.18) -- (459.29,146.31)(459.61,138.05) -- (451.03,146.44) ;
\draw   (451.35,157.18) -- (459.29,165.31)(459.61,157.05) -- (451.03,165.44) ;
\draw   (451.35,176.18) -- (459.29,184.31)(459.61,176.05) -- (451.03,184.44) ;
\draw [color={rgb, 255:red, 0; green, 0; blue, 0 }  ,draw opacity=1 ]   (455,100.5) -- (474.58,9.95) ;
\draw [shift={(475,8)}, rotate = 102.2] [color={rgb, 255:red, 0; green, 0; blue, 0 }  ,draw opacity=1 ][line width=0.75]    (10.93,-3.29) .. controls (6.95,-1.4) and (3.31,-0.3) .. (0,0) .. controls (3.31,0.3) and (6.95,1.4) .. (10.93,3.29)   ;
\draw    (455,100.5) -- (476.49,182.07) ;
\draw [shift={(477,184)}, rotate = 255.24] [color={rgb, 255:red, 0; green, 0; blue, 0 }  ][line width=0.75]    (10.93,-3.29) .. controls (6.95,-1.4) and (3.31,-0.3) .. (0,0) .. controls (3.31,0.3) and (6.95,1.4) .. (10.93,3.29)   ;
\draw   (450.02,17.04) -- (459.55,25.27)(459.93,16.91) -- (449.64,25.4) ;
\draw   (450.02,78.04) -- (459.55,86.27)(459.93,77.91) -- (449.64,86.4) ;
\draw   (450.02,60.04) -- (459.55,68.27)(459.93,59.91) -- (449.64,68.4) ;
\draw   (450.02,38.04) -- (459.55,46.27)(459.93,37.91) -- (449.64,46.4) ;
\draw    (160,116) -- (172.48,105.3) ;
\draw [shift={(174,104)}, rotate = 139.4] [color={rgb, 255:red, 0; green, 0; blue, 0 }  ][line width=0.75]    (10.93,-3.29) .. controls (6.95,-1.4) and (3.31,-0.3) .. (0,0) .. controls (3.31,0.3) and (6.95,1.4) .. (10.93,3.29)   ;
\draw    (284,102) -- (337,102.96) ;
\draw [shift={(339,103)}, rotate = 181.04] [color={rgb, 255:red, 0; green, 0; blue, 0 }  ][line width=0.75]    (10.93,-3.29) .. controls (6.95,-1.4) and (3.31,-0.3) .. (0,0) .. controls (3.31,0.3) and (6.95,1.4) .. (10.93,3.29)   ;

\draw (183,102.4) node [anchor=north west][inner sep=0.75pt]  [font=\scriptsize]  {$s_{0}$};
\draw (217,103.9) node [anchor=north west][inner sep=0.75pt]  [font=\scriptsize]  {$0$};
\draw (239,25.4) node [anchor=north west][inner sep=0.75pt]  [font=\tiny]  {$\frac{\pi }{2} -\epsilon $};
\draw (254,179.4) node [anchor=north west][inner sep=0.75pt]  [font=\tiny]  {$-\frac{\pi }{2} +\epsilon $};
\draw (81.27,12.4) node [anchor=north west][inner sep=0.75pt]  [font=\scriptsize,xslant=0.02]  {$t^{m_{0} -1}\hat{\Theta }_{D_{X}}( t-s_{0})$};
\draw (457,103.9) node [anchor=north west][inner sep=0.75pt]  [font=\scriptsize]  {$0$};
\draw (479,25.4) node [anchor=north west][inner sep=0.75pt]  [font=\tiny]  {$\frac{\pi }{2} -\epsilon $};
\draw (480,182.4) node [anchor=north west][inner sep=0.75pt]  [font=\tiny]  {$-\frac{\pi }{2} +\epsilon $};
\draw (344.27,11.4) node [anchor=north west][inner sep=0.75pt]  [font=\scriptsize,xslant=0.02]  {$t^{m_{0} -1}\hat{\Theta }_{D_{X}}( t)$};
\draw (5,116.4) node [anchor=north west][inner sep=0.75pt]  [font=\scriptsize]  {$ \begin{array}{l}
Remove\ the\ non-integrability\ \\
at\ there\ \ when\ s_{0} =0
\end{array}$};
\draw (291,80.4) node [anchor=north west][inner sep=0.75pt]  [font=\scriptsize]  {$s_{0}\rightarrow 0$};
\draw (251,80.4) node [anchor=north west][inner sep=0.75pt]  [font=\scriptsize]  {$\Re ( t)$};
\draw (496,82.4) node [anchor=north west][inner sep=0.75pt]  [font=\scriptsize]  {$\Re ( t)$};
\draw (190,2.4) node [anchor=north west][inner sep=0.75pt]  [font=\scriptsize]  {$\Im ( t)$};
\draw (419,9.4) node [anchor=north west][inner sep=0.75pt]  [font=\scriptsize]  {$\Im ( t)$};

\end{tikzpicture}
\caption{The relationship between two regularized determinants from the perspective of the singularities of $\hat{\Theta}_{D_X}(t)$.}
\label{fig: relation between two regularized determinants}
\end{figure}
\begin{remark}
In fact, the left hand side of Theorem \ref{relation Reg1 and Reg2} is uniformly convergent for any $s_0<0$, and the theorem can be proved directly by sending the parameter $s_0\to 0$.

However, the reason for employing a more intricate proof strategy, specifically the Laplace transform approach presented here, is to provide the insight into the choice of these two regularizations via the resurgent analysis of singularities.  This perspective also explains why the two methods coincide after differentiation and {sending the deformation parameter to zero}.

Both regularizations can be viewed as sums of Laplace transforms in two directions regarding certain modified Theta series (recall the proofs in Theorem \ref{theorem about derivative regularization in introduction} and \ref{thm about the computation of deform in intro}); they differ, however, in how they establish integrability at the origin. The derivative regularization achieves this goal by multiplying the series by the asymptotic power term at zero (e.g.$(-t)^{d}\hat{\Theta}_{D_X}$). In contrast, the exponential deformation regularization shifts the Theta series itself (e.g. $\hat{\Theta}_{D_X}(t-s_0)$). The shift moves all singularities to the left half-plane, making the Laplace transform well-defined in the right half-plane.

It deserves to point out that if we simply let the deformation parameter $s_0$ approach $0$, the deformation regularization does not converge to the derivative regularization. This is because the non-integrable singularity at $t=s_0$ shifts to the origin ($t=0$), causing both directions of the Laplace transform to be ill-defined. Consequently, we must multiply the integrand by a factor of at least $(-t)^{d}$ to ensure that $(-t)^{d}\hat{\Theta}_{D_X}(t-s_0)$ remains integrable at $t=0$ as $s_0 \to 0$. after the Laplace transform, multiplying the integrand--$\hat{\Theta}_{D_X}(t-s_0)$ by $(-t)^{d}$ corresponds precisely to take the $d$-th derivative of the exponential deformation regularization.

\end{remark}


\section{Concluding Remarks}


We studied two ways to regularize the ill defined functional determinant of Laplacian on Riemannian manifolds in the unified framework of resurgence theory. Inspired by Voros（\cite{voros1992spectral}) for previous regularizations, we clarify and establish similar results for two regularized determinants proposed in the current work. These regularized determinants are characterized by both the Stokes structure and the analytic continuation of the corresponding modified theta series. Specifically, our formulas for these determinants reproduce the Poisson summation formula, its deformed form, and the Selberg trace formula when applied to the unit circle and compact Riemann surfaces of higher genus. One may expect that the current regularization methods could be generalized to more general differential operators. In fact we hope to study the spectra of the Schr\"odinger operator and Gutzwiller semiclassical trace formula (SCTF) through resurgence theory as shown by Selberg trace formula.  

The singularities and Stokes structure of theta series reveal finer information on the spectra as one can see from  identities \eqref{distribution PSF} and \eqref{deformed distributon PSF} for the case of Poisson summation formula, and \eqref{computation of det in X and Selberg zeta} for the case of Selberg trace formula. We will address this issue by varying the integration contours.

{As shown in the formal solutions of Airy equation 
(\cite{liyong2025resurgence}), the modularity of the partial theta series appearing in Chern-Simons theory (\cite{HanLi2023}), and as did by Voros (\cite{voros1992spectral}) for the spectral functions in the cases of quartic oscillator, to name a few examples, resurgence theory will show its true power for much more general singularities beyond poles.
}We will come back to this issue in the following work.

\section*{Acknowledgments}
{The authors are deeply grateful to David Sauzin, who is supposed to be coauthor at the early stage,  and Yong Li for their valuable discussions and suggestions on the paper. We also appreciate Ziyang Zhu for his kind assistance during the preparation of the manuscript. Both authors are partially supported by National Key R\&D Program of China (2020YFA0713300) and NSFC (No.s 12171327, 12261131498). }


\appendix
\section{Proof of the Deformed Poisson Summation Formula 
}{\label{subsection:Deformed PSF}}


In the appendix, we prove Theorem \ref{thm: deformed psf} by using a method similar to that employed for the Poisson summation formula in (\cite{Steinbook, cartierVoros1988}). However, the meromorphic extension of identity \eqref{deformed distributon PSF}
$$\sum_{n\in\mathbb{Z}}^{\infty}\frac{\rho e^{s_0 |n|}}{n^2+\rho^2}=e^{is_0\rho}\left(\frac{\pi}{\tanh\pi\rho}
-\pi
\right)-i\mathcal{L}^{\left(\pi-\varepsilon\right)}\left(\frac{\sinh s_0}{\cosh t-\cosh s_0}\right)(-i\rho)-\pi e^{is_0\rho},$$
from the sector $\arg(\rho)\in(-\frac{\pi}{2}+\varepsilon,\varepsilon)$ is non-trivial. Specifically, the Laplace transform term $\mathcal{L}^{\left(\pi-\varepsilon\right)}\left(\frac{\sinh s_0}{\cosh t-\cosh s_0}\right)(-i\rho)$ develops singularities during this process, which we shall first discuss explicitly.

~\\
\textbf{The meromorphic extension of $\mathcal{L}^{\left(\pi-\varepsilon\right)}\left(\frac{\sinh s_0}{\cosh t-\cosh s_0}\right)(-i\rho)$.}





\begin{lemma}{\label{lemma: hol-ext of K}}
$\mathcal{L}^{\left(\pi-\varepsilon\right)}\left(\frac{\sinh s_0}{\cosh t-\cosh s_0}\right)(-i\rho)$ can be extended holomorphically to $\{\rho\in\mathbb{C}:-\pi+\epsilon<\arg\rho\leq\pi+\epsilon\}-i\ZZ_{\geq 0}$, and we denote it by $K_{s_0}(\rho)$. When $\rho\in i\Pi_0^{-\epsilon}:=\{|\rho|e^{i\theta}:\epsilon<\theta<\epsilon+\pi,|\rho|>0.\}$,
\begin{equation}{\label{extension of K}}
K_{s_0}(\rho)=\mathcal{L}^{-\epsilon}\left(\frac{\sinh s_0}{\cosh t-\cosh s_0}\right)(-i\rho)-\frac{2\pi i\left(e^{is_0\rho}-e^{-is_0\rho}\right)}{e^{2\pi \rho}-1}+2\pi ie^{-i\rho s_0};
\end{equation}
when $\arg\rho=\pi+\epsilon$ and any $\arg(-s_0-2\pi i)<-\epsilon^\prime<-\epsilon$,
\begin{equation}{\label{extension of K in line}}
K_{s_0}(\rho)=\mathcal{L}^{-\epsilon^\prime}\left(\frac{\sinh s_0}{\cosh t-\cosh s_0}\right)(-i\rho)-\frac{2\pi i\left(e^{is_0\rho}-e^{-is_0\rho}\right)}{e^{2\pi \rho}-1}+2\pi ie^{-i\rho s_0}.
\end{equation}
As a result, the singular set of $K_{s_0}(\rho)$ is $i\ZZ_{\geq0}$, with $\operatorname{Res}(K_{s_0}(\rho);in)=2\pi(e^{-s_0n}-e^{s_0n})$.
\end{lemma}

\begin{remark}
Formula \eqref{extension of K in line} can be explained as follows. Note that $$\frac{\sinh s_0}{\cosh t - \cosh s_0} \in \mathcal{N}\left((\arg(-s_0-2\pi i), 0), 0\right).$$ By Proposition \ref{def:Laplace transform is welldefined}, the Laplace transforms $\mathcal{L}^{-\epsilon^\prime}$ and $\mathcal{L}^{-\epsilon}$ coincide on the common domain $i\Pi_0^{-\epsilon} \cap i\Pi_0^{-\epsilon^\prime}$ for any $\arg(-s_0-2\pi i) < -\epsilon^\prime < -\epsilon$ which ensures that $K_{s_0}(\rho)$ can be extended meromorphically to  $\arg\rho = \pi + \epsilon$, retaining the same form as in \eqref{extension of K}.

\end{remark}

\begin{proof}
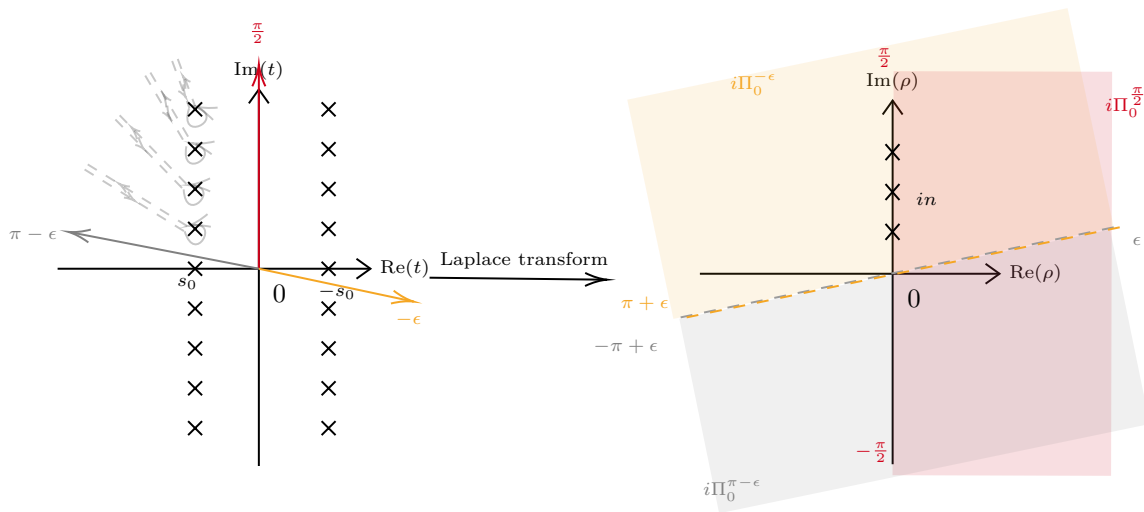
\begin{figure}[htbp]
    \centering    
    \tikzset{every picture/.style={line width=0.75pt}} 
\begin{tikzpicture}[x=0.75pt,y=0.75pt,yscale=-1,xscale=1]

\draw  (58,137.5) -- (215,137.5) node [anchor=west] {\scriptsize $\text{Re}(t)$};
\draw  (159,236.5) -- (159,47.5) node [anchor=south] {\scriptsize $\text{Im}(t)$};
\draw (208,132.5) -- (215,137.5) -- (208,142.5) (154,54.5) -- (159,47.5) -- (164,54.5);

\draw  (380.34,140.02) -- (530.72,140.02) node [anchor=west] {\scriptsize $\text{Re}(\rho)$};
\draw  (477.08,235.62) -- (477.08,53.11) node [anchor=south] {\scriptsize $\text{Im}(\rho)$};
\draw (523.72,135.02) -- (530.72,140.02) -- (523.72,145.02) (472.08,60.11) -- (477.08,53.11) -- (482.08,60.11);

\draw [color={rgb, 255:red, 208; green, 2; blue, 27 }  ,draw opacity=1 ]  (159,137.5) -- (159,37) ;
\draw [shift={(159,35)}, rotate = 90] [color={rgb, 255:red, 208; green, 2; blue, 27 }  ,draw opacity=1 ][line width=0.75]    (10.93,-3.29) .. controls (6.95,-1.4) and (3.31,-0.3) .. (0,0) .. controls (3.31,0.3) and (6.95,1.4) .. (10.93,3.29)   ;
\draw  [color={rgb, 255:red, 0; green, 0; blue, 0 }  ,draw opacity=0 ][fill={rgb, 255:red, 208; green, 2; blue, 27 }  ,fill opacity=0.13 ] (477.55,38.35) -- (477.02,241.15) -- (586.71,241.44) -- (587.24,38.63) -- cycle ;
\draw [color={rgb, 255:red, 128; green, 128; blue, 128 }  ,draw opacity=1 ]   (159,137.5) -- (65.3,119.38) ;
\draw [shift={(63.33,119)}, rotate = 10.94] [color={rgb, 255:red, 128; green, 128; blue, 128 }  ,draw opacity=1 ][line width=0.75]    (10.93,-3.29) .. controls (6.95,-1.4) and (3.31,-0.3) .. (0,0) .. controls (3.31,0.3) and (6.95,1.4) .. (10.93,3.29)   ;
\draw  [color={rgb, 255:red, 0; green, 0; blue, 0 }  ,draw opacity=0 ][fill={rgb, 255:red, 128; green, 128; blue, 128 }  ,fill opacity=0.12 ] (586.18,117.84) -- (370.04,162.68) -- (390,260.26) -- (606.14,215.42) -- cycle ;
\draw [color={rgb, 255:red, 245; green, 166; blue, 35 }  ,draw opacity=1 ]   (159,137.5) -- (234.71,153.58) ;
\draw [shift={(236.67,154)}, rotate = 191.99] [color={rgb, 255:red, 245; green, 166; blue, 35 }  ,draw opacity=1 ][line width=0.75]    (10.93,-3.29) .. controls (6.95,-1.4) and (3.31,-0.3) .. (0,0) .. controls (3.31,0.3) and (6.95,1.4) .. (10.93,3.29)   ;
\draw  [color={rgb, 255:red, 0; green, 0; blue, 0 }  ,draw opacity=0 ][fill={rgb, 255:red, 245; green, 166; blue, 35 }  ,fill opacity=0.11 ] (367.04,162.68) -- (588.03,116.71) -- (564.87,6.71) -- (343.88,52.67) -- cycle ;
\draw    (244.67,142) -- (331.67,142.98) ;
\draw [shift={(333.67,143)}, rotate = 180.64] [color={rgb, 255:red, 0; green, 0; blue, 0 }  ][line width=0.75]    (10.93,-3.29) .. controls (6.95,-1.4) and (3.31,-0.3) .. (0,0) .. controls (3.31,0.3) and (6.95,1.4) .. (10.93,3.29)   ;
\draw [color={rgb, 255:red, 155; green, 155; blue, 155 }  ,draw opacity=1 ] [dash pattern={on 4.5pt off 4.5pt}]  (588.03,116.71) -- (367.04,162.68) ;
\draw [color={rgb, 255:red, 245; green, 166; blue, 35 }  ,draw opacity=1 ] [dash pattern={on 4.5pt off 4.5pt}]  (591.03,116.71) -- (370.04,162.68) ;


\draw [color={rgb, 255:red, 0; green, 0; blue, 0 }  ,draw opacity=0.22 ]   (121.28,97.69) .. controls (122.24,115.77) and (140.09,96.99) .. (125.6,95.61) ;
\draw [shift={(123.61,95.53)}, rotate = 359.47] [color={rgb, 255:red, 0; green, 0; blue, 0 }  ,draw opacity=0.22 ][line width=0.75]    (10.93,-3.29) .. controls (6.95,-1.4) and (3.31,-0.3) .. (0,0) .. controls (3.31,0.3) and (6.95,1.4) .. (10.93,3.29)   ;

\draw [color={rgb, 255:red, 0; green, 0; blue, 0 }  ,draw opacity=0.22, dash pattern={on 4pt off 4pt} ]   (118.3,32.21) -- (127.03,55.99)(115.48,33.24) -- (124.22,57.03) ;

\draw [color={rgb, 255:red, 0; green, 0; blue, 0 }  ,draw opacity=0.22 ]   (123.96,57.13) .. controls (118.44,74.37) and (141.8,63.16) .. (128.74,56.71) ;
\draw [shift={(126.92,55.93)}, rotate = 20.29] [color={rgb, 255:red, 0; green, 0; blue, 0 }  ,draw opacity=0.22 ][line width=0.75]    (10.93,-3.29) .. controls (6.95,-1.4) and (3.31,-0.3) .. (0,0) .. controls (3.31,0.3) and (6.95,1.4) .. (10.93,3.29)   ;

\draw [color={rgb, 255:red, 0; green, 0; blue, 0 }  ,draw opacity=0.22, dash pattern={on 4pt off 4pt} ]   (73.12,85.14) -- (123.56,116.42)(71.54,87.69) -- (121.98,118.97) ;

\draw [color={rgb, 255:red, 0; green, 0; blue, 0 }  ,draw opacity=0.22 ]   (121.77,118.69) .. controls (127.1,135.99) and (139.85,113.43) .. (125.45,115.62) ;
\draw [shift={(123.5,116.02)}, rotate = 345.39] [color={rgb, 255:red, 0; green, 0; blue, 0 }  ,draw opacity=0.22 ][line width=0.75]    (10.93,-3.29) .. controls (6.95,-1.4) and (3.31,-0.3) .. (0,0) .. controls (3.31,0.3) and (6.95,1.4) .. (10.93,3.29)   ;

\draw [color={rgb, 255:red, 0; green, 0; blue, 0 }  ,draw opacity=0.22 ]   (122.62,76.05) .. controls (119.44,93.88) and (141.1,79.65) .. (127.3,75) ;
\draw [shift={(125.38,74.47)}, rotate = 12.62] [color={rgb, 255:red, 0; green, 0; blue, 0 }  ,draw opacity=0.22 ][line width=0.75]    (10.93,-3.29) .. controls (6.95,-1.4) and (3.31,-0.3) .. (0,0) .. controls (3.31,0.3) and (6.95,1.4) .. (10.93,3.29)   ;

\draw [color={rgb, 255:red, 0; green, 0; blue, 0 }  ,draw opacity=0.22 ] (93.34,98.86) .. controls (94.25,100.72) and (95.39,102.25) .. (96.76,103.45) .. controls (95.16,102.57) and (93.33,102.03) .. (91.28,101.8) ;
\draw [color={rgb, 255:red, 0; green, 0; blue, 0 }  ,draw opacity=0.22 ] (92.31,99.25) .. controls (91.35,97.42) and (90.17,95.92) .. (88.77,94.76) .. controls (90.39,95.59) and (92.23,96.09) .. (94.29,96.26) ;

\draw [color={rgb, 255:red, 0; green, 0; blue, 0 }  ,draw opacity=0.22 ] (98.63,71.51) .. controls (99.37,73.97) and (100.29,76) .. (101.4,77.6) .. controls (100.1,76.44) and (98.63,75.71) .. (96.96,75.41) ;
\draw [color={rgb, 255:red, 0; green, 0; blue, 0 }  ,draw opacity=0.22 ] (97.9,72.46) .. controls (97.12,70.03) and (96.16,68.04) .. (95.03,66.49) .. controls (96.34,67.59) and (97.83,68.26) .. (99.5,68.48) ;

\draw [color={rgb, 255:red, 0; green, 0; blue, 0 }  ,draw opacity=0.22 ] (113.1,55.62) .. controls (113.23,58.19) and (113.65,60.38) .. (114.35,62.2) .. controls (113.37,60.76) and (112.1,59.7) .. (110.56,59.02) ;
\draw [color={rgb, 255:red, 0; green, 0; blue, 0 }  ,draw opacity=0.22 ] (112.52,56.52) .. controls (112.33,53.98) and (111.88,51.82) .. (111.14,50.04) .. controls (112.15,51.43) and (113.44,52.43) .. (115.01,53.04) ;

\draw [color={rgb, 255:red, 0; green, 0; blue, 0 }  ,draw opacity=0.22 ] (121.1,44.95) .. controls (121.23,47.52) and (121.65,49.71) .. (122.35,51.53) .. controls (121.37,50.09) and (120.1,49.03) .. (118.56,48.36) ;
\draw [color={rgb, 255:red, 0; green, 0; blue, 0 }  ,draw opacity=0.22 ] (122.18,47.19) .. controls (122,44.64) and (121.54,42.49) .. (120.8,40.71) .. controls (121.82,42.1) and (123.11,43.09) .. (124.68,43.7) ;

\draw [color={rgb, 255:red, 0; green, 0; blue, 0 }  ,draw opacity=0.22, dash pattern={on 4pt off 4pt} ]   (89.14,60.49) -- (123.85,96.19)(86.99,62.58) -- (121.7,98.28) ;
\draw [color={rgb, 255:red, 0; green, 0; blue, 0 }  ,draw opacity=0.22, dash pattern={on 4pt off 4pt} ]   (104.56,39.04) -- (125.97,75.4)(101.97,40.56) -- (123.39,76.92) ;

\foreach \y in {79, 99, 119} {
    \draw (473.5, \y-4) -- (480.5, \y+4) (480.5, \y-4) -- (473.5, \y+4);
}
\foreach \y in {57.5, 77.5, 97.5, 117.5, 137.5, 157.5, 177.5,197.5, 217.5} {
    \draw (123.5, \y-3.5) -- (130.5, \y+3.5) (130.5, \y-3.5) -- (123.5, \y+3.5);
}
\foreach \y in {57.5, 77.5, 97.5, 117.5, 137.5, 157.5, 177.5, 197.5, 217.5} {
    \draw (190.5, \y-3.5) -- (197.5, \y+3.5) (197.5, \y-3.5) -- (190.5, \y+3.5);
}

\draw (161,140.9) node [anchor=north west] {$0$};
\draw (149,9.4) node [anchor=north west] [font=\scriptsize,color={rgb, 255:red, 126; green, 211; blue, 33 }] {$\textcolor[rgb]{0.82,0.01,0.11}{\frac{\pi }{2}}$};
\draw (479.53,143.1) node [anchor=north west] {$0$};
\draw (578.85,41.39) node [anchor=north west] [font=\scriptsize] {$\textcolor[rgb]{0.82,0.01,0.11}{i\Pi }\textcolor[rgb]{0.82,0.01,0.11}{_{0}^{\frac{\pi }{2}}}$};
\draw (28.6,112) node [anchor=north west] [font=\scriptsize] {$\textcolor[rgb]{0.5,0.5,0.5}{\pi -\epsilon }$};
\draw (322.04,168.04) node [anchor=north west] [font=\scriptsize] {$\textcolor[rgb]{0.5,0.5,0.5}{-\pi +\epsilon }$};
\draw (592.47,116.67) node [anchor=north west] [font=\scriptsize] {$\textcolor[rgb]{0.5,0.5,0.5}{\epsilon }$};
\draw (463.42,21.74) node [anchor=north west] [font=\scriptsize] {$\textcolor[rgb]{0.82,0.01,0.11}{\frac{\pi }{2}}$};
\draw (453.27,219.86) node [anchor=north west] [font=\scriptsize] {$\textcolor[rgb]{0.82,0.01,0.11}{-}\textcolor[rgb]{0.82,0.01,0.11}{\frac{\pi }{2}}$};
\draw (377.11,237.26) node [anchor=north west] [font=\scriptsize] {$\textcolor[rgb]{0.5,0.5,0.5}{i\Pi }\textcolor[rgb]{0.5,0.5,0.5}{_{0}^{\pi -\epsilon }}$};
\draw (222.8,154.6) node [anchor=north west] [font=\scriptsize] {$\textcolor[rgb]{0.96,0.65,0.14}{-\epsilon }$};
\draw (335.67,146.4) node [anchor=north west] [font=\scriptsize] {$\textcolor[rgb]{0.96,0.65,0.14}{\pi +\epsilon }$};
\draw (390.52,33.51) node [anchor=north west] [font=\scriptsize] {$\textcolor[rgb]{0.96,0.65,0.14}{i\Pi }\textcolor[rgb]{0.96,0.65,0.14}{_{0}^{-\epsilon }}$};
\draw (245,124) node [anchor=north west] [font=\scriptsize] {Laplace transform};
\draw (484,94.4) node [anchor=north west] [font=\scriptsize] {$in$};
\draw (184,140.6) node [anchor=north west] [font=\scriptsize] {$-s_{0}$};
\draw (113.07,137.13) node [anchor=north west] [font=\scriptsize] {$s_{0}$};

\end{tikzpicture}
    
    \caption{Meromorphic extension of $\mathcal{L}^{\left(\pi-\varepsilon\right)}\left(\frac{\sinh s_0}{\cosh t-\cosh s_0}\right)(-i\rho)$ by varying directions of the Laplace transform. \textbf{Left ($t$-plane):} Gray, red, and yellow rays represent transform directions $\pi-\epsilon$, $\pi/2$, and $-\epsilon$, respectively; the dashed gray rays and circles with arrows indicate the Hankel contours decomposed from the difference between the transforms along $\pi/2$ and $\pi-\epsilon$. \textbf{Right ($\rho$-plane):} Shaded regions denote the domains of convergence for the Laplace transforms along the three directions of corresponding colors. By computing the differences between these transforms, the convergence domains are glued together, thus extending the holomorphic domain of $\mathcal{L}^{\left(\pi-\varepsilon\right)}\left(\frac{\sinh s_0}{\cosh t-\cosh s_0}\right)(-i\rho)$.}
    
    \label{fig:Holomorphic continuation of K}
\end{figure}
Notice that the function $\frac{\sinh s_0}{\cosh t-\cosh s_0}$ is holomorphic and bounded on each ray emanating from the origin, except for those containing the points $\pm s_0+2\pi ik$ ($k\in \mathbb{Z}$). These points are simple poles, and the residues are given by
$\operatorname{Res}\left(\frac{\sinh s_0}{\cosh t-\cosh s_0}; t = \pm s_0 + 2\pi i k\right) = \pm 1.$
Thus, for $J_k^+:=\left(\arg(s_0+2\pi(k+1) i),\arg (s_0+2\pi ki)\right)$, $J_k^-:=\left(\arg(-s_0+2k\pi i),\arg (-s_0+2(k+1)\pi i)\right)$ with $k\in\ZZ$, we have 
$$
\frac{\sinh s_0}{\cosh t-\cosh s_0}\in\left(\bigcap_{k\in\ZZ}\mathcal{N}(J_k^+,0)\cap\mathcal{N}(J_k^-,0)\right)\cap\mathcal{R}_{\pm s_0+2\pi i\ZZ}^{simp}.
$$
Similarly to the proof of Lemma \ref{lemma: Laplace and stokes}, we decompose $(\mathcal{L}^{\pi-\epsilon} - \mathcal{L}^{\frac{\pi}{2}})\left(\frac{\sinh s_0}{\cosh t-\cosh s_0}\right)(-i\rho)$ into a sum of integrals along infinite Hankel contours\footnote{ Given $\theta_k^+:=\arg(s_0+2\pi i k)$ and $\epsilon^{''}$ small enough, for any $k\in\ZZ_{>0}$, each contour $\Gamma_{k,\epsilon^{''}}^+$ is defined as follows: it starts from $s_0+2\pi ik+e^{i\theta_k^+}\infty$ and goes to $s_0+2\pi ik+\epsilon^{''}e^{i\theta_k^+}$ along a straight line, and circles anticlockwise once around $s_0+2\pi ik$, at last, it returns to the $s_0+2\pi ik+e^{i(\theta_k^{+}+2\pi)}\infty$ along the ray starting from $s_0+2\pi ik+\epsilon^{''}e^{i(\theta_k^{+}+2\pi)}$. Look at the dashed gray rays and gray circles in the Figure \ref{fig:Holomorphic continuation of K}.}. Since $s_0 + 2k\pi i$ are simple poles, only the contour integrals around small circles enclosing these poles contribute to the sum. Consequently, by the Residue formula, we obtain that
\begin{equation}\label{formula:extend K 1}
\begin{aligned} K_{s_0}(\rho) &= \mathcal{L}^{\frac{\pi}{2}}\left(\frac{\sinh s_0}{\cosh t-\cosh s_0}\right)(-i\rho) - 2\pi i \sum_{k=1}^{\infty} e^{is_0\rho} e^{-2\pi\rho k}\\
&=\mathcal{L}^{\frac{\pi}{2}}\left(\frac{\sinh s_0}{\cosh t-\cosh s_0}\right)(-i\rho)-\frac{2\pi ie^{i\rho s_0}}{e^{2\pi \rho}-1},\quad \rho \in i\Pi_0^{\pi-\epsilon} \cap i\Pi_0^{\frac{\pi}{2}}.
\end{aligned}
\end{equation}
Using the same method, for $\arg(-s_0-2\pi i)<-\epsilon<0$ and any $\rho \in i\Pi_0^{\pi-\epsilon} \cap i\Pi_0^{\frac{\pi}{2}}$, we also have
\begin{equation}\label{formula:extend K 2}
\left(\mathcal{L}^{-\epsilon}-\mathcal{L}^{\frac{\pi}{2}}\right)\left(\frac{\sinh s_0}{\cosh t-\cosh s_0}\right)(-i\rho)=-2\pi i \sum_{k=0}^{\infty}e^{-is_0\rho}e^{-2\pi\rho k}=\frac{-2\pi ie^{-i\rho s_0}}{e^{2\pi \rho}-1}-2\pi ie^{-i\rho s_0}.    
\end{equation}

Thus, by the uniqueness of holomorphic extension, we find that $\mathcal{L}^{\left(\pi-\varepsilon\right)}\left(\frac{\sinh s_0}{\cosh t-\cosh s_0}\right)(-i\rho)$ can be meromorphically extended to $i\Pi_0^{-\epsilon} \cup i\Pi_0^{\pi-\epsilon}$ via the relations \eqref{formula:extend K 1} and \eqref{formula:extend K 2}. As a result,
$$
K_{s_0}(\rho)=\mathcal{L}^{-\epsilon}\left(\frac{\sinh s_0}{\cosh t-\cosh s_0}\right)(-i\rho)-\frac{2\pi i\left(e^{is_0\rho}-e^{-is_0\rho}\right)}{e^{2\pi \rho}-1}+2\pi ie^{-i\rho s_0},
$$
when $\rho\in i\Pi_0^{-\epsilon}:=\{|\rho|e^{i\theta}:\epsilon<\theta<\epsilon+\pi,|\rho|>0.\}$.
\end{proof}

Now we can start to prove Theorem \ref{thm: deformed psf} by Lemma \ref{lemma: hol-ext of K}.

\begin{proof}
 Observing that $ \sum\limits_{n \in \mathbb{Z}}\left|\frac{2 \rho e^{s_0|n|} }{\rho^2+n^2}\right|$ is uniformly convergent in any compact subset of $\mathbb{C}\setminus i\mathbb{Z}$, then for any fixed $N\in\ZZ_{>0}$ and the path $\gamma_N:=\{\pm\eta+it|t\in(-N-\frac{1}{2},N+\frac{1}{2})\}\cup\{s\pm i(N+\frac{1}{2})|s\in[-\eta,\eta]\}$ (see Figure \ref{fig:the integral contour gamma-N}), we have
$$
\sum_{n \in \mathbb{Z}} h(n)e^{s_0 |n|}
=\lim _{N \rightarrow \infty} \frac{1}{2 \pi i} \int_{\gamma_N} \sum_{n \in \mathbb{Z}} \frac{\rho h(-i \rho)e^{s_0|n|}}{\rho^2+n^2} d \rho.
$$

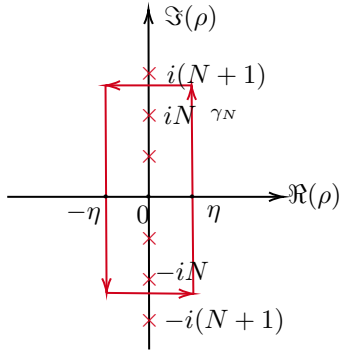
\begin{figure}[htbp]
    \centering
\tikzset{every picture/.style={line width=0.75pt}} 

\begin{tikzpicture}[x=0.50pt,y=0.50pt,yscale=-1,xscale=1]

\draw [color={rgb, 255:red, 208; green, 2; blue, 27 }  ,draw opacity=1 ]   (364.11,223) -- (362.79,68.11) ;
\draw [shift={(362.78,66.11)}, rotate = 89.51] [color={rgb, 255:red, 208; green, 2; blue, 27 }  ,draw opacity=1 ][line width=0.75]    (10.93,-3.29) .. controls (6.95,-1.4) and (3.31,-0.3) .. (0,0) .. controls (3.31,0.3) and (6.95,1.4) .. (10.93,3.29)   ;
\draw [color={rgb, 255:red, 208; green, 2; blue, 27 }  ,draw opacity=1 ]   (298,66.11) -- (298.77,220.33) ;
\draw [shift={(298.78,222.33)}, rotate = 269.71] [color={rgb, 255:red, 208; green, 2; blue, 27 }  ,draw opacity=1 ][line width=0.75]    (10.93,-3.29) .. controls (6.95,-1.4) and (3.31,-0.3) .. (0,0) .. controls (3.31,0.3) and (6.95,1.4) .. (10.93,3.29)   ;
\draw    (224,150) -- (266,150) -- (430,150) ;
\draw [shift={(432,150)}, rotate = 180] [color={rgb, 255:red, 0; green, 0; blue, 0 }  ][line width=0.75]    (10.93,-3.29) .. controls (6.95,-1.4) and (3.31,-0.3) .. (0,0) .. controls (3.31,0.3) and (6.95,1.4) .. (10.93,3.29)   ;
\draw    (330.83,264.92) -- (330.34,10.92) ;
\draw [shift={(330.33,8.92)}, rotate = 89.89] [color={rgb, 255:red, 0; green, 0; blue, 0 }  ][line width=0.75]    (10.93,-3.29) .. controls (6.95,-1.4) and (3.31,-0.3) .. (0,0) .. controls (3.31,0.3) and (6.95,1.4) .. (10.93,3.29)   ;
\draw [color={rgb, 255:red, 208; green, 2; blue, 27 }  ,draw opacity=1 ]   (362.78,66.11) -- (300,66.11) ;
\draw [shift={(298,66.11)}, rotate = 360] [color={rgb, 255:red, 208; green, 2; blue, 27 }  ,draw opacity=1 ][line width=0.75]    (10.93,-3.29) .. controls (6.95,-1.4) and (3.31,-0.3) .. (0,0) .. controls (3.31,0.3) and (6.95,1.4) .. (10.93,3.29)   ;
\draw [color={rgb, 255:red, 208; green, 2; blue, 27 }  ,draw opacity=1 ]   (298.78,222.33) -- (362.11,222.98) ;
\draw [shift={(364.11,223)}, rotate = 180.58] [color={rgb, 255:red, 208; green, 2; blue, 27 }  ,draw opacity=1 ][line width=0.75]    (10.93,-3.29) .. controls (6.95,-1.4) and (3.31,-0.3) .. (0,0) .. controls (3.31,0.3) and (6.95,1.4) .. (10.93,3.29)   ;

\draw (372,154.4) node [anchor=north west][inner sep=0.75pt]  [font=\normalsize]  {$\eta $};
\draw (374.67,81.4) node [anchor=north west][inner sep=0.75pt]  [font=\scriptsize]  {$\gamma _{N}{}$};
\draw (319,154.4) node [anchor=north west][inner sep=0.75pt]    {$0$};
\draw (324,142) node [anchor=north west][inner sep=0.75pt]  [font=\Large]  {$\cdot $};
\draw (321,110.4) node [anchor=north west][inner sep=0.75pt]    {$\textcolor[rgb]{0.82,0.01,0.11}{\times }$};
\draw (321,79.4) node [anchor=north west][inner sep=0.75pt]    {$\textcolor[rgb]{0.82,0.01,0.11}{\times }$};
\draw (321,48.4) node [anchor=north west][inner sep=0.75pt]    {$\textcolor[rgb]{0.82,0.01,0.11}{\times }$};
\draw (321,172.4) node [anchor=north west][inner sep=0.75pt]    {$\textcolor[rgb]{0.82,0.01,0.11}{\times }$};
\draw (321,203.4) node [anchor=north west][inner sep=0.75pt]    {$\textcolor[rgb]{0.82,0.01,0.11}{\times }$};
\draw (433,137.4) node [anchor=north west][inner sep=0.75pt]    {$\Re ( \rho )$};
\draw (340,3.4) node [anchor=north west][inner sep=0.75pt]    {$\Im ( \rho )$};
\draw (292,142) node [anchor=north west][inner sep=0.75pt]  [font=\Large]  {$\cdot $};
\draw (357,142) node [anchor=north west][inner sep=0.75pt]  [font=\Large]  {$\cdot $};
\draw (266,152.4) node [anchor=north west][inner sep=0.75pt]  [font=\normalsize]  {$-\eta $};
\draw (321,234.4) node [anchor=north west][inner sep=0.75pt]    {$\textcolor[rgb]{0.82,0.01,0.11}{\times }$};
\draw (339.33,79.4) node [anchor=north west][inner sep=0.75pt]  [font=\normalsize]  {$iN$};
\draw (342,46.73) node [anchor=north west][inner sep=0.75pt]  [font=\normalsize]  {$i( N+1)$};
\draw (333,197.07) node [anchor=north west][inner sep=0.75pt]  [font=\normalsize]  {$-iN$};
\draw (340,231.73) node [anchor=north west][inner sep=0.75pt]  [font=\normalsize]  {$-i( N+1)$};

\end{tikzpicture}
\caption{The integral contour $\gamma_N$.}
\label{fig:the integral contour gamma-N}
\end{figure}

Plugging the identity  (\ref{deformed distributon PSF}) 
$$\sum_{n\in\mathbb{Z}}^{\infty}\frac{\rho e^{s_0 |n|}}{n^2+\rho^2}h(-i\rho)=e^{is_0\rho}\left(\frac{\pi}{\tanh\pi\rho}\right)h(-i\rho)-i K_{s_0}(\rho)h(-i\rho)
-\pi e^{is_0\rho}h(-i\rho),\quad \rho\in\mathbb{C}-i\mathbb{Z}$$
into it and using the Residue formula for contour $\gamma_N$ and Lemma \ref{lemma: hol-ext of K}, we get
$$
\sum_{n \in \mathbb{Z}} h(n)e^{s_0|n|}
=\lim _{N \rightarrow \infty} \frac{1}{2 \pi i} \int_{\gamma_N} e^{i s_0\rho}h(-i \rho)\left(\frac{\pi }{\tanh \pi \rho}\right)d \rho -\lim_{N\to\infty}\sum_{n=1}^{N}(e^{-s_0n}-e^{s_0n})h(n),
$$
which can in turn be simplified to
$$\sum_{N \in \mathbb{Z}} h(n)e^{-s_0n}=\lim _{N \rightarrow \infty} \frac{1}{2 \pi i} \int_{\gamma_N} e^{i s_0\rho}h(-i \rho)\left(\frac{\pi }{\tanh \pi \rho}\right)d \rho.
$$

On the horizontal segments of $\gamma_N$, where $\rho = s \pm i(N + \frac{1}{2})$ for $s \in (-\eta, \eta)$, the term $\frac{\pi}{\tanh \pi \rho}$ is bounded. Given that $h(-i\rho) = O(|\rho|^{-1-\delta})$ as $\Im(\rho) \to -\infty$ and $h(-i\rho) = O(e^{\delta^\prime |\rho|})$ as $\Im(\rho) \to +\infty$, the integrals along these segments satisfy:
$$\left|\int_{- i(N+\frac{1}{2})}^{\pm\eta- i(N+\frac{1}{2})} e^{is_0\rho}h(-i \rho)\frac{\pi }{\tanh \pi \rho} d \rho\right| = O\left(N^{-1-\delta}\right),$$
and
$$\left|\int_{i(N+\frac{1}{2})}^{\pm\eta+ i(N+\frac{1}{2})} e^{is_0\rho}h(-i \rho)\frac{\pi }{\tanh \pi \rho} d \rho\right| = O\left(e^{(\delta^\prime-s_0)N}\right).$$
Noting that $s_0 > \delta^\prime$, then as $N \to \infty$, both contributions vanish, and we have
$$\lim _{N \rightarrow \infty} \int_{\pm i(N+\frac{1}{2})}^{\pm\eta\pm i(N+\frac{1}{2})} e^{is_0\rho}h(-i \rho)\frac{\pi }{\tanh \pi \rho} d \rho = 0.$$
As $N\to\infty$, the integral over $\gamma_N$ reduces to the contributions from its vertical segments. By expanding $\frac{\pi}{\tanh\pi\rho}$ into an exponential series uniformly on these segments and changing the variable $\rho = i\rho^\prime$, we obtain:
\begin{equation} \label{deform computation in appendix}\begin{aligned}
\sum_{n\in \ZZ} h(n)e^{-s_0n}&=\frac{1}{2\pi i} \left(\int_{\eta-i\infty}^{\eta+i\infty}-\int_{-\eta-i\infty}^{-\eta+i\infty} \right) e^{is_0\rho}h(-i \rho)\left(\frac{\pi}{\tanh\pi \rho}\right)d\rho \\
&\overset{\rho = i\rho^\prime}{=} \sum_{m\geq 1} \int_{-i\eta-\infty}^{-i\eta+\infty} e^{-s_0\rho}h(\rho)e^{-2\pi mi \rho} d\rho + \frac{1}{2}\int_{-i\eta-\infty}^{-i\eta+\infty} e^{-s_0\rho}h(\rho) d\rho \\
&\quad + \sum_{m\geq 1} \int_{i\eta-\infty}^{i\eta+\infty} e^{-s_0\rho}h(\rho)e^{2\pi m i\rho} d\rho + \frac{1}{2} \int_{i\eta-\infty}^{i\eta+\infty} e^{-s_0\rho}h(\rho) d\rho.\end{aligned}\end{equation}

For all $m\geq0$, when $R\to\infty$, $$\int_{R-i\eta}^{R} |e^{-s_0\rho}h(\rho)e^{-2\pi mi \rho}| d \rho,\int_{R}^{R+i\eta} |e^{-s_0\rho}h(\rho)e^{2\pi mi \rho}| d \rho\in O(e^{(\delta^\prime-s_0)R}),$$
$$\int_{-R-i\eta}^{-R} |e^{-s_0\rho}h(\rho)e^{-2\pi mi \rho}| d \rho,\int_{-R}^{-R+i\eta} |e^{-s_0\rho}h(\rho)e^{2\pi mi \rho}| d \rho\in O(R^{-1-\delta}),$$
so,$$\int_{-i\eta-\infty}^{-i\eta+\infty} e^{-s_0\rho}h(\rho)e^{-2\pi mi \rho} d \rho=\int_{-\infty}^{+\infty} e^{-s_0\rho}h(\rho)e^{-2\pi mi \rho} d \rho,$$
 $$\int_{i\eta-\infty}^{i\eta+\infty} e^{-s_0\rho}h(\rho)e^{2\pi mi \rho} d \rho=\int_{-\infty}^{+\infty} e^{-s_0\rho}h(\rho)e^{2\pi mi \rho} d \rho.$$
Now (\ref{deform computation in appendix}) becomes $\sum_{m\in\mathbb{Z}}\int_{-\infty}^{+\infty} h(\rho)e^{(2\pi m +is_0)i\rho} d \rho$,
which gives us the desired deformed Poisson Summation formula:
$$\sum_{n\in\mathbb{Z}}h(n)e^{-s_0n}=2\pi\sum_{m\in\mathbb{Z}}\hat{h}(2\pi m+is_0).$$
\end{proof}

\bibliographystyle{plain}
\bibliography{ref}


 \end{document}